\documentclass[12pt,twoside]{article}
\usepackage[utf8]{inputenc}
\usepackage{paralist}
\usepackage{amsmath, amsthm, amssymb, bbm}
\usepackage{geometry}
\usepackage{graphicx}
\usepackage[font=normalsize,skip=0pt]{caption}
\usepackage[colorlinks=true, allcolors=blue]{hyperref}
\usepackage{optidef}
\usepackage{multirow}
\usepackage{authblk}

\def\red{}

\newcommand{\cF}{{\mathcal F}}
\newcommand{\E}{{\mathbb E}}

\newcommand{\N}{{\mathbb N}}
\newcommand{\R}{{\mathbb R}}
\newcommand{\bfK}{{\mathbf K}}

\newtheorem{prop}{Proposition}[section]
\newtheorem{lemma}{Lemma}[section]
\newtheorem{corollary}{Corollary}[section]

\newtheorem{remark}{Remark}[section]

\title{Risk valuation of quanto derivatives on temperature and electricity}

\begin{document}

\author[1]{Aurélien Alfonsi}
\author[1,2]{Nerea Vadillo}
\affil[1]{CERMICS, Ecole des Ponts, Marne-la-Vall\'ee, France. MathRisk, Inria, Paris,
  France.  }
\affil[2]{AXA Climate, Paris, France. }
\affil[ ]{email: \texttt{aurelien.alfonsi@enpc.fr}, \texttt{nerea.vadillo@axaclimate.com} }
\date{\today}

\maketitle

\begin{abstract}
    This paper develops a coupled model for day-ahead electricity prices and average daily temperature which allows to model quanto weather and energy derivatives. These products have gained on popularity as they enable to hedge against both volumetric and price risks. Electricity day-ahead prices and average daily temperatures are modelled through non homogeneous Ornstein-Uhlenbeck processes driven by a Brownian motion and a Normal Inverse Gaussian Lévy process, which allows to include dependence between them. A Conditional Least Square method is developed to estimate the different parameters of the model and used on real data. Then, explicit and semi-explicit formulas are obtained for derivatives including quanto options and compared with Monte Carlo simulations. Last, we develop explicit formulas to hedge statically single and double sided quanto options by a portfolio of electricity options and temperature options (CDD or HDD). 
\end{abstract}
{\bf Keywords:} Energy quanto options, Weather derivatives, Joint Temperature-Electricity model, Risk hedging

\section*{Introduction}

The increasing impact of climate change on businesses has led to a growing demand for risk transfer instruments to hedge against its consequences. The energy sector is particularly affected by such weather variability. On the one hand, weather variability affects energy production. The availability of wind and solar radiation impacts the production of renewable electricity~\cite{benth2017stochastic}.  Similarly, experienced and predicted temperatures influence demand, as cold snaps increase heating demand in winter and heat waves increase cooling demand in summer~\cite{benth2009information}. This exposure to weather variability is often referred to as volumetric risk~\cite{muller2000weather}. On the other hand, weather forecasts can have a direct impact on energy prices as actors anticipate demand increases and act in advance. This less frequently discussed risk is referred to as price risk~\cite{brockett2005weather}.

Weather derivatives emerged in the 1990s as a response to this need for risk transfer. These financial instruments are based on an underlying weather index and trigger a claim depending on the value of the index at maturity, similar to other financial market derivatives. These instruments experienced significant success in the early 2000s, reaching \$45 billion in notional volume traded in the market in 2006 according to the World Risk Management Association~\cite{WRMA_2006}. Mainly dominated by temperature-based derivatives, up to 95\% of the market, the weather market remained illiquid with small volumes traded in the standardized open market and most of the volume traded OTC~\cite{weagley2019financial}. It also led to extensive research into the modeling of weather derivatives and best pricing methodologies~\cite{jewson2005weather}~\cite{benth2012modeling}~\cite{benth2007putting}~\cite{alaton2002modelling}~\cite{cao1999pricing}~\cite{campbell2005weather}~\cite{brody2002dynamical}.

By 2008, the weather market experienced a significant slowdown, with trading volumes declining to \$11.8 billion in 2011~\cite{WRMA_2011}. This corresponded to a general slowdown of financial markets, but also, according to Pérez-González and Yun~\cite{perez2013risk}, to the birth of new hybrid derivatives that could combine both volumetric and price risk. These new products, also called quantos, were indexed to two underlying parameters, one proxying the volumetric risk, typically a weather parameter, and one proxying the price risk, typically the spot price of electricity, gas or oil. These double-indexed products already existed in the market for other financial assets (foreign exchange, bonds, commodities)~\cite{baxter1996financial}~\cite{hull2003options}. They are technically challenging because they require a convincing model of the joint distribution of the underlyings. Our analysis will focus on finding a model to price temperature and spot electricity price quantos.

Unfortunately, the literature exploring weather quantos is thin. Benth \red{et al.}~\cite{benth2015pricing} use a Heath-Jarrow-Morton approach to price hybrid derivatives combining New York Mercantile Exchange-traded natural gas futures and Chicago Mercantile Exchange-traded heating degree day futures for New York. Matsumoto and Yamada study the optimal design of mixed weather derivatives on wind indices and electricity prices~\cite{yamada2023construction}. Benth and Ibrahim~\cite{benth2017stochastic} develop continuous-time models combining spot prices and logarithmic photovoltaic power production. For quantos combining energy prices and temperature, we should mention Caporin \red{et al.}~\cite{caporin2012model}, who develop a two-dimensional daily ARFIMA-FIGARCH model for energy prices and temperature. They consider both an actuarial and a financial approach and perform simulation-based pricing that leads to important price differences~\cite{caporin2012model}. Cucu \red{et al.}~\cite{cucu2016managing} develop a combined natural gas spot price and temperature model. They address calibration and pricing challenges for temperature-gas swaps. Finally, Benth \red{et al.}~\cite{benth2023pricing} consider bivariate Markov-modulated additive processes with independent non-stationary increments to model quantos combining temperature and energy and electricity and gas prices. Given a known analytical joint characteristic function for the logarithmic futures prices, they derive quanto pricing formulas for the Fast Fourier Transform (FFT) technique.

We begin our analysis by exploring various marginal models for spot energy price and daily temperature. In particular, we dive deep into a large literature on energy and commodity modeling~\cite{deschatre2021survey}~\cite{weron2008forecasting}. First, we examine mean-reverting diffusion models. Pioneering models by Gibson and Schwartz~\cite{gibson1990stochastic}, Schwartz~\cite{schwartz1997stochastic}, and Lucia and Schwartz~\cite{lucia2002electricity} propose two- or three-factors Gaussian diffusion dynamics to model commodity assets. However, the presence of non-Gaussian behaviors, including spikes, jumps, and heavy tails, has led to a refinement of these initial models. One proposal is to extend mean-reverting diffusion processes to Levy noises. Thus, compound Poisson processes have been studied by Geman and Roncoroni~\cite{geman2006understanding}, Cartea and Figueroa~\cite{cartea2005pricing}, and Meyer-Brandis and Tankov~\cite{meyer2008multi}. A second widespread proposal is to move to multi-factor models with Brownian~\cite{lucia2002electricity}~\cite{benth2023pricing} or Levy increments~\cite{benth2007non}~\cite{bjork2002term}. Finally, Benth and Benth explore the relevance of mean-reverting diffusion processes with Normal Inverse Gaussian (NIG) increments~\cite{benth2004normal}. We compare these models and consider different estimation and process characterization challenges for day-ahead auction market clearing prices~\cite{weron2014electricity} for the French and Northern Italian electricity markets. Finally, we propose to model the daily day-ahead log spot prices with mean-reverting processes and NIG increments. 

For daily temperature models, our analysis is less extensive as the reader can refer to Alfonsi and Vadillo~\cite{alfonsi2022stochastic} for a more detailed presentation of daily temperature modeling applied to temperature derivatives pricing. We mainly suggest using a simple mean-reverting Gaussian model as in Benth and Benth~\cite{benth2007putting} to model the daily average temperature for Charles de Gaulle and Milano Linate weather station data. 

Third, we address the challenge of the joint temperature and log spot energy price distribution by proposing a coupled model on the dynamics. In particular, we introduce the Brownian noise of the temperature dynamics into the energy process. This allows the integration of weather information available at the time of price formation, as suggested by Benth and Meyer-Brandis~\cite{benth2009information}, while maintaining flexibility and tractability in both processes. We estimate the marginals and dependence parameters of the joint model using Condition Least Square estimation applied to the characteristic function. $\chi^2$ tests comparing the simulated and observed joint distributions confirm the goodness of fit of the combined model for both French and Northern Italian datasets.

Next, we introduce the pricing of quanto derivatives. Contrary to Benth \red{et al.}~\cite{benth2023pricing}, we do not consider quanto and temperature derivatives market as arbitrage-free complete markets. As noted above, most exchanges are OTC and CME standardized weather derivatives lack daily trading volume~\cite{weagley2019financial}. Temperature and energy quantos do not exist in any open market. Therefore, risk-neutral pricing is arguable and we stick to analyse the risk under the historical probability. Given our combined model, we derive explicit and semi-explicit formulas for the average payoff of futures, swaps and single sided options, here called $\mathcal{E}$-options, and double-sided options on temperature indices (HDD and CDD) and spot electricity price. These formulas are compared with payoff distributions derived from Monte Carlo simulations. Finally, we discuss the static hedging of $\mathcal{E}$-HDD and double sided quanto options in an self-financing portfolio framework, where the option is hedged by HDD and energy spot derivatives. We show that by using our model we can hedge most of the risk of quanto options and reduce their variances.

Hence the contributions of this paper are multiple. First, it develops a convincing joint model for spot energy prices and daily average temperatures. Second, it proposes a method to estimate all the parameters of the model and assess the goodness of fit. Third, it develops pricing formulas under historical probability for futures, swaps, single and double-sided quanto options. Finally, it shows the hedging capability of single and double-sided quanto options.

The paper is organized as follows. Section 1 presents the models for the univariate and combined dynamics of the logarithmic day-ahead spot price and the average daily temperature. Section 2 explores different dynamics for the log day-ahead energy spot price and justifies the modeling choice. Section 3 discusses the estimation challenges. Section 4 introduces the combined model and confirms its goodness of fit. Section 5 addresses the risk valuation of quantos that depend on both energy and temperature and develops a framework for static hedging of $\mathcal{E}$-HDDs and quanto options.

\section{Model and data description}

In this section we introduce our combined models to describe the dynamics of daily day-ahead energy log spot price $(X_t)_{t\geq0}$ and the temperature $(T_t)_{t\geq0}$. In the following, we will also note $S_t=e^{X_t}, \ t\geq0,$ the daily day-ahead energy spot price. We will consider time-continuous  models with the time unit of one day ($\Delta=1$), which follows literature practices as noted by Deschatre~\cite{deschatre2021survey}. Thus, $T_{i\Delta}$ will model the average daily temperature of the $i$-th day, usually defined in financial contracts as the average between the hourly minimum and maximum temperature. 

We will consider the below Model~\eqref{eq:comb} as the combined model for daily day-ahead energy log spot price and average daily temperature.

\begin{equation} \label{eq:comb}
\tag{ETM}
   \begin{dcases}
    d(X_t-\mu_X(t))  &= -\kappa_X (X_t-\mu_X(t)) dt + \lambda \sigma_T dW^T_t + dL^X_t\\
    d(T_t-\mu_T(t))  &= -\kappa_T (T_t-\mu_T(t)) dt+ \sigma_T dW^T_t
 \end{dcases} 
\end{equation}
When $\lambda=0$, the dynamics of $X$ and $T$ are independent. The elements characterizing the dynamics of the log-price~$X$ are: 
\begin{itemize}
    \item The deterministic function $\mu_X:\R_+\to \R$ represents the trend and seasonality component. We assume that
    \begin{equation}
    \mu_X(t) = \beta^X_0 t + \alpha^X_{1} \sin(\xi t) + \beta^X_{1} \cos(\xi t) +  \alpha^{X,DoW}_{DoW(t)}
    \label{eq:mu}
    \end{equation}
    where  $\xi = \frac{2 \pi}{365}$ and $DoW(t)=\lfloor \frac{t}{\Delta} \rfloor$ mod $p$ where $p \in \N^*$. In practice, $p=7$ and $\alpha^{X,DoW}_{DoW(\cdot)}$ corresponds to the constant depending of the day in the week.
    \item The parameter $\kappa_X>0$ corresponds to the mean-reverting (or autoregressive) behaviour\red{, or more precisely to the speed of mean-reversion}.
    \item $L^X$ is a Normal Inverse Gaussian distribution of parameters $(\alpha^X, \beta^X, \delta^X, m^X)$ which properties are described in Appendix~\ref{appendix:NIG}. We will assume that this process is centered ($\E[L^X_t]=0$), which means
$$ m^X+\delta^X\frac{\beta^X}{\gamma^X}=0.$$
\end{itemize}
Similarly, the elements characterizing the dynamics of the temperature~$T$ are: 
\begin{itemize}
    \item The function $\mu_T$ represents the trend and seasonality component. We assume that 
    \begin{equation}
    \mu_T(t) = \alpha^T_0 + \beta^T_0 t + \alpha^T_{1} \sin(\xi t) + \beta^T_{1} \cos(\xi t) \mbox{, where } \xi = \frac{2 \pi}{365}.
    \end{equation}
    \item The parameter $\kappa_T$ corresponds to \red{the speed of mean-reversion toward the trend component.}
    \item $W^T$ is a Brownian motion independent of $L^X$, and $\sigma_T>0$ to the standard deviation of the noise.
\end{itemize}
Last, the parameter $\lambda \in \R$ allows for some dependence between both processes. 
For the temperature, the Ornstein-Uhlenbeck (OU) model with Brownian noise corresponds to a well established model developed by Benth et al.~\cite{benth2007putting} and largely spread on literature. We refer to Alfonsi and Vadillo~\cite{alfonsi2022stochastic} for a recent discussion on temperature models. Section~\ref{Sec_modele1D} presents different models for the electricity spot price and justifies the choice of dynamics of $X$ in~\eqref{eq:comb} when $\lambda=0$. Then, Section~\ref{Sec_modele2D} explores the pertinence of Model~\eqref{eq:comb} and shows that it reproduces well the features of our data.

As it will be often useful in the calculations, we write here the integrated version of Model~\eqref{eq:comb}
\begin{equation} \label{eq:comb_int}
   \begin{dcases}
           X_t - \mu_X(t) &= e^{-\kappa_X (t-s)}(X_s-\mu_X(s))+ \lambda \sigma_T \int_s^t e^{-\kappa_X(t-u)}  dW^T_u + \int_s^t e^{-\kappa_X(t-u)}dL^X_u \\
           T_t - \mu_T(t) &= e^{-\kappa_T (t-s)}(T_s-\mu_T(s))+  \sigma_T \int_s^t e^{-\kappa_T(t-u)} dW^T_u  ,
    \end{dcases}
\end{equation} 
and introduce the notation $\tilde{X}_t=X_t - \mu_X(t)$ and $\tilde{T}_t=T_t - \mu_T(t)$ that will be used through the paper. \red{The covariance of the residuals is given by
\begin{align}\label{cov_residuals}
&Cov\left(X_{t+\Delta}-\mu_X(t+\Delta)-e^{-\kappa_X \Delta}(X_t-\mu_X(t)),T_{t+\Delta}-\mu_T(t+\Delta)-e^{-\kappa_T \Delta}(T_t-\mu_T(t)) \right)\\
&=\lambda \sigma^2_T   \frac{1-e^{-(\kappa_{X}+ \kappa_{T})\Delta}}{\kappa_{X} + \kappa_{T}}. \notag
\end{align}}
\red{
\begin{remark} We present an alternative to Model~\eqref{eq:comb}. Let us consider for the spot electricity price a mean-reverting NIG-Levy process
$$d(X_t-\mu_X(t))  = -\kappa_X (X_t-\mu_X(t)) dt +  dL^X_t,$$
and then use the approximation of Asmussen and Rosinski~\cite[Theorem 2.1 and Example 2.5]{AsRo}:
$$L^X_t \underset{\varepsilon \to 0^+}{\approx} m^\varepsilon t+ \sqrt{\frac{2\delta^X \varepsilon}{\pi}}W^X_t + N^{X,\varepsilon}_t, $$
where $W^X$ is a Brownian motion and $N^{X,\varepsilon}_t= \sum_{s\le t} \mathbf{1}_{|\Delta L^X_s|>\varepsilon}\Delta L^X_s$ is an independent compound Poisson process with jump measure
$\nu^{X,\varepsilon}(dx)=\mathbf{1}_{|x|>\varepsilon}\frac{\delta^X \alpha^X }{\pi |x|}e^{\beta^X x}K_1(\alpha^X|x|)dx.$
Since $L^X_t$ is assumed to be centered, $\widetilde{N}^{X,\varepsilon}_t=N^{X,\varepsilon}_t + m^\varepsilon t$ is a compensated Poisson process. An alternative to Model~\eqref{eq:comb} would then be the following joint model: 
\begin{equation} \label{eq:comb2}
\tag{ETM2}
   \begin{dcases}
    d(X_t-\mu_X(t))  &= -\kappa_X (X_t-\mu_X(t)) dt +   \sqrt{\frac{2\delta^X \varepsilon}{\pi}}dW^X_t + d\widetilde{N}^{X,\varepsilon}_t\\
    d(T_t-\mu_T(t))  &= -\kappa_T (T_t-\mu_T(t)) dt+ \sigma_T dW^T_t
 \end{dcases} 
\end{equation}
with $dW^T_tdW^X_t=\rho dt$. 
By a similar calculation, the covariance between residuals is equal to $\rho \sqrt{\frac{2\delta \varepsilon}{\pi}}\sigma_T   \frac{1-e^{-(\kappa_{X}+ \kappa_{T})\Delta}}{\kappa_{X} + \kappa_{T}}$ and matches exactly~\eqref{cov_residuals} if, and only if
$$\rho^2 \varepsilon= \frac{\pi {\sigma_T}^2 \lambda^2}{2 \delta^X }.  $$
To have $\varepsilon$ as small as possible (i.e. to stay as close as possible to a NIG distribution), one should take $\rho=\text{sign}(\lambda)\in \{-1,1\}$ and $\varepsilon= \frac{\pi {\sigma_T}^2 \lambda^2}{2 \delta^X }$. For the relatively small values of $\lambda$ as we observed on data, models~\eqref{eq:comb} and~\eqref{eq:comb2}
are rather close, but~\eqref{eq:comb} produces larger tail events as shown by Figure~\ref{fig:qq_CPoisson_NIG}. A potential advantage of~\eqref{eq:comb2} is to allow for a separate estimation for the electricity spot price and the temperature. However, the  parameter estimation that we propose for Model~\eqref{eq:comb} is rather simple to implement as shown in Section~\ref{sec:calibration}. Besides, the use of the NIG distribution in Model~\eqref{eq:comb} is more convenient than the Poisson compound distribution to get tractable formulas in Section~\ref{Sec_risk_quanto} and for the simulation.    
\end{remark}}
\begin{figure}
    \centering
    \includegraphics[width=0.5\linewidth]{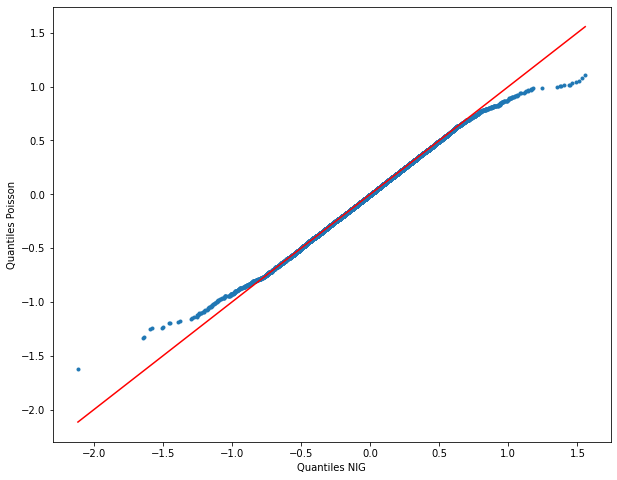}
    \caption{Quantile-quantile plot of empirical distributions with 100,000 samples of $L^X_1$ (abscissa) and $\sqrt{\frac{2\delta^X \varepsilon}{\pi}}W^X_1 + N^{X,\varepsilon}_1$ (ordinate), with the parameters of Table~\ref{tab:parameters_simu} and $\varepsilon= \frac{\pi {\sigma_T}^2 \lambda^2}{2 \delta^X }$. }
    \label{fig:qq_CPoisson_NIG}
\end{figure}

\subsection*{Data description}

The above model is tested in real world data. In particular, we study day-ahead log spot energy prices in France and North Italy from 5th January 2015 to 31st December 2018. This data is extracted from the ENTSO-E Transparency Platform and Gestore Mercati Energetici (GME) and are available hourly until 31st December 2022. We decided to average hourly data into daily data to avoid additional intra-day noise and follow literature practices~\cite{deschatre2021survey}. This granularity choice will equally enable us to match the granularity of other dynamics like the temperature data's. Additionally, we exclude 2019 to 2022 years as energy price time series show considerably erratic paths due two major macroeconomic shocks: the COVID-19 pandemic and the Ukrainian war.

For temperature data, we choose to extract average daily temperature time series for Paris-Charles de Gaulle airport and Milano-Linate airport weather stations. These weather stations are referenced in WMO with the following identification numbers 7157 and 16080. Daily average temperature is defined as the average between the maximum and minimum hourly temperatures. Data is extracted from a private data provider platform. This latter is in charge of the removal of outliers. The data is therefore considered as cleaned in the following of this study.

\section{Overview of different energy models}\label{Sec_modele1D}

The literature on energy modeling is large. We would particularly recommend the surveys of Weron \cite{weron2008forecasting} and Deschatre et al.\cite{deschatre2021survey}. Although there can be exceptions, experts usually focus on either day ahead daily spot or forwards prices. The granularity kept is hence the day and is seen as the average of hourly spot prices.  While forward price market modeling has been explored successfully through HJM-modeling paradigm~\cite{benth2008stochastic}, we will focus on spot or log spot price modeling. In the study of this section, we do not consider structural models nor neural networks models, but we focus rather on stochastic models. Indeed, our objective is to combine energy dynamic modeling with temperature modeling to handle the risk of hybrid options and have a clear understanding of the model parameters. In the following we will consider log spot price to ensure positivity of the energy dynamics.

\subsection{Mean-reverting diffusion models}

The first models describing electricity dynamics are mean-reverting diffusion models. They were first developed by Gibson and Schwartz~\cite{gibson1990stochastic}, Schwartz~\cite{schwartz1997stochastic}  and Lucia and Schwartz~\cite{lucia2002electricity}. They do not focus only on electricity but apply these models to  wider range of energy commodities (crude oil, on-peak electricity spot prices). They are built around the concept of \textit{convenience yield} and model daily commodity through a Ornstein-Uhlenbeck (OU) process with Brownian noise~\cite{vasicek1977equilibrium} as follows:
\begin{equation}
 \begin{cases}
X_t &= \mu(t) + \tilde{X}_t \\
\tilde{X}_t &= -\kappa \tilde{X}_t dt + \sigma d W_t
\end{cases}  \label{eq:diff}
\end{equation}
where $\mu(\cdot)$ corresponds to a deterministic function including trend and seasonality and $W$ to a Brownian noise.\\

There is not much discussion on the form of the deterministic function $\mu(\cdot)$. While several papers reduce this function to a simple constant~\cite{schwartz1997stochastic}~\cite{gibson1990stochastic}~\cite{deng2000stochastic}~\cite{knittel2005empirical}, other suggest different order Fourier expansions~\cite{benth2004normal}~\cite{cartea2005pricing} or piece-wise stepped functions~\cite{lucia2002electricity}. These latter enable to include annual seasonality and capture differences in winter and summer prices, a phenomenon agreed upon literature~\cite{knittel2005empirical}~\cite{escribano2011modelling}. Pawlowsky and Nowak~\cite{pawlowski2015modelling} justifies the presence of a trend component on the deterministic component. We suggest to keep this constant, trend and seasonality deterministic components and test their significance such that we define:
\begin{equation}
    \mu(t) = \beta_0 t + \alpha_{1} \sin(\xi t) + \beta_{1} \cos(\xi t) \mbox{, where } \xi = \frac{2 \pi}{365}
    \label{eq:mu-1}
\end{equation}
This first deterministic equation is implemented and tested on our data. For this we consider the results of the following regression function.
\begin{equation}
\sum_{i=0}^{N-1} \left( X_{i+1} - \mu(i+1) - e^{-\kappa }(X_i-\mu(i)) \right)^2,
\label{eq:ref-1}
\end{equation} 
where $\mu(\cdot)$ is defined as in Equation~\eqref{eq:mu-1}.

Following minimisation of Equation~\eqref{eq:ref-1}, we check significance of the coefficients and residual plots. While all coefficients show  5\% significance, residual plots are less satisfactory.
Indeed, as shown in Figure~\ref{fig:pacf}, residuals show important weekly dependencies. This phenomenon has already been observed by several papers. Following this observation, some suggest to distinguish week and week-end effects~\cite{lucia2002electricity}~\cite{meyer2008multi} while others show statistical significance of daily dummy integration~\cite{bosco2007deregulated}. We also considered alternatives such as additional weekly seasonality terms. The form of $\mu(\cdot)$ minimising the AIC criteria turned to be the defined in Equation~\eqref{eq:mu}. 
\begin{figure}[ht]
    \begin{minipage}[b]{0.45\linewidth}
        \centering
        \includegraphics[width=\textwidth]{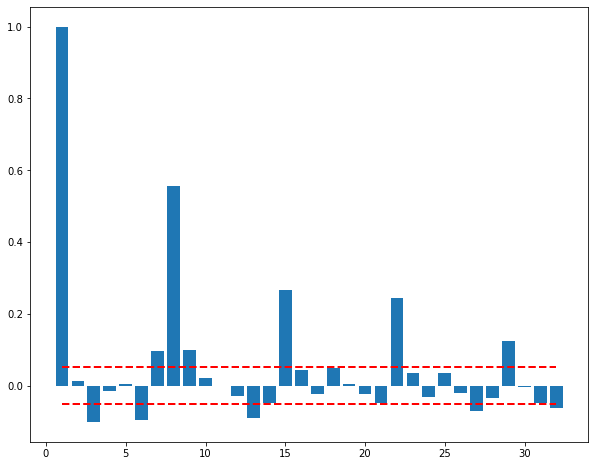}
    \end{minipage}
    \hspace{0.3cm}
    \begin{minipage}[b]{0.45\linewidth}
        \centering
        \includegraphics[width=\textwidth]{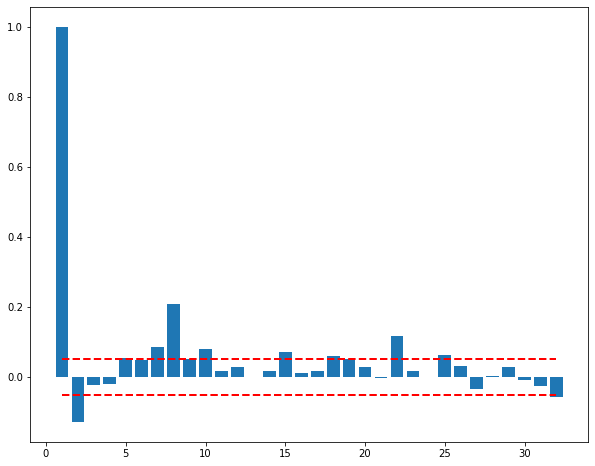}
    \end{minipage}
    \caption{Partial autocorrelation plots of residuals of the Regression~\eqref{eq:ref-1} where $\mu(\cdot)$ defined as in Equation~\eqref{eq:mu-1} (left) and as in Equation~\eqref{eq:mu} (right). The dashed red line corresponds to the 95\% confidence interval from which we can consider the partial autocorrelation coefficient is significantly different from 0.} \label{fig:pacf}
\end{figure} \\
Now let turn to the noise $W$, in our case we consider the residuals of Regression~\eqref{eq:ref-1} where $\mu(.)$ as defined in~\eqref{eq:mu} to assess the characteristics of the distribution of $W$. Initially, the first models suggested Brownian dynamics for such noise following the model of Vasicek~\cite{vasicek1977equilibrium}. However, this proposal has been considerably challenged. Indeed, as can be seen in Figure~\ref{fig:qqplot}, the qqplot of the residuals show significant deviation from normal theoretical quantiles. The residuals seem indeed to present heavier tails than normal residuals. 

\begin{figure}[htb!]  
  \begin{minipage}[b]{0.5\linewidth}
    \centering
    \includegraphics[width=0.8\textwidth]{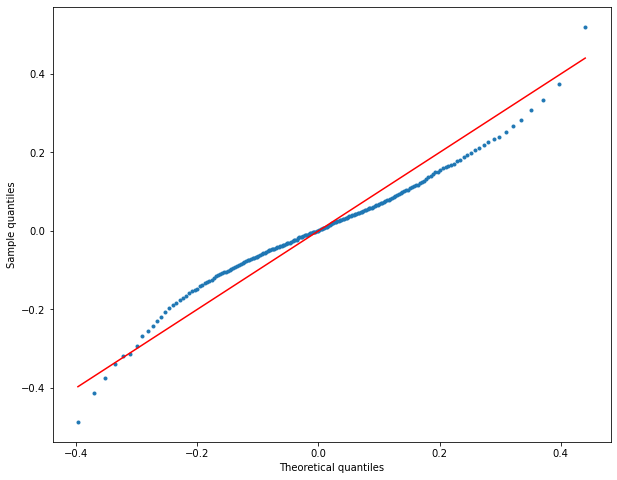}
  \end{minipage}
  \begin{minipage}[b]{0.5\linewidth}
    \centering
    \includegraphics[width=0.8\textwidth]{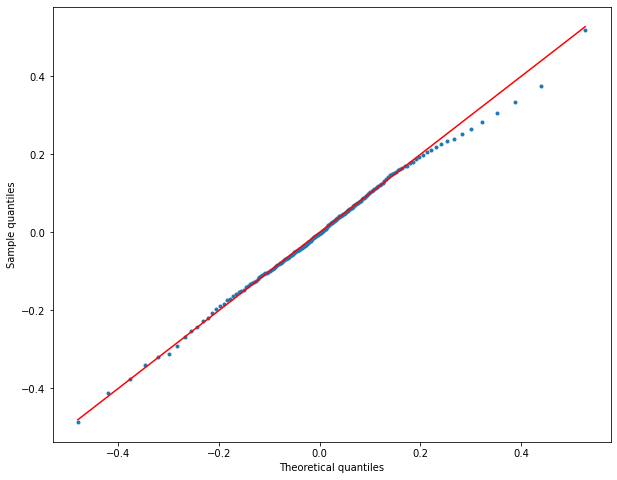}
  \end{minipage} 
  \begin{minipage}[b]{0.5\linewidth}
    \centering
    \includegraphics[width=0.8\textwidth]{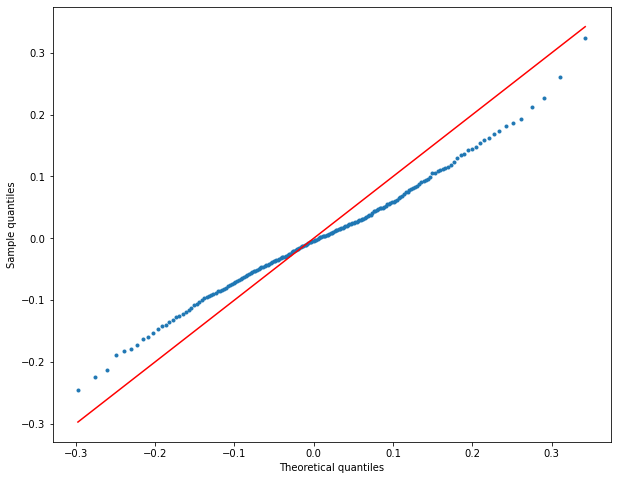}
  \end{minipage}
  \begin{minipage}[b]{0.5\linewidth}
    \centering
    \includegraphics[width=0.8\textwidth]{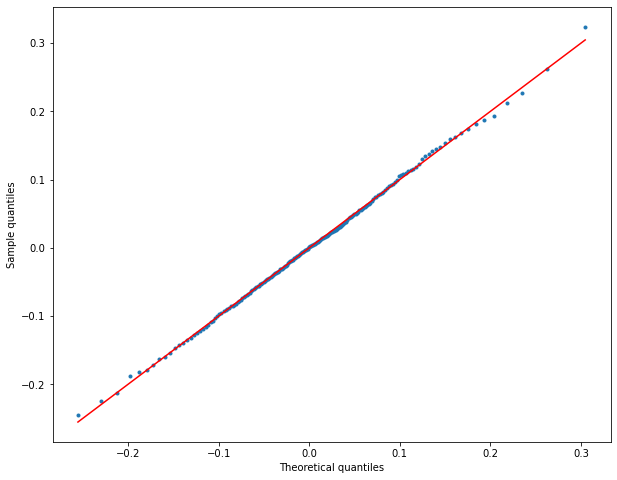} 
  \end{minipage} 
  \caption{Quantile quantile plots for residuals of Regression~\eqref{eq:ref-1} compared with a theoretical quantiles of a normal distribution (left) and of a normal inverse gaussian distribution (right) for French energy (first row) and North Italian Energy (second row).} \label{fig:qqplot}
\end{figure}
Different answers have been giving to this limitations. Some authors suggest to introduce price spikes thanks to jump-diffusion processes~\cite{deng2000stochastic}~\cite{cartea2005pricing}~\cite{geman2006understanding} while others explore multi-factor jump-diffusion models~\cite{meyer2008multi} or alternative distributions for the residuals~\cite{benth2004normal}. Next sections will concentrate on these proposals and on their estimation power.

\subsection{Mean-reverting jump-diffusion models (MRJD)}
Several papers have studied the possibility to consider non-Gaussian increments. A particularly popular one is to combine Brownian motion with a compound Poisson process \cite{deng2000stochastic} that would capture the price spikes usually observed in energy prices, extending~\eqref{eq:diff} as follows
\begin{equation}
 \begin{aligned}
X_t &= \mu(t) + \tilde{X}_t \\
\tilde{X}_t &= -\kappa \tilde{X}_t dt + \sigma dW_t + dJ_t,
\end{aligned}   
\end{equation}
where $J_t$ is a Poisson process of intensity $\lambda$ such that $J_t = \sum_{i=1}^{N_t} \xi_i$ where $\xi_i$ are i.i.d. jump magnitudes that can follow distributions such as log-Normal~\cite{cartea2005pricing}, exponential~\cite{deng2000stochastic} or mixture of exponential distributions~\cite{pawlowski2015modelling}. One of the main drivers for distribution selection is the ability to obtain explicit formulas for forward prices which mainly depends on the jump magnitude assumptions and the time-dependence of the other parameters. Indeed more flexible models, such us Geman and Roncoroni~\cite{geman2006understanding}'s, offer more flexibility on the properties of the Poisson process but do not enable explicit formulas of the forward prices.

However, only few papers analyse the challenge of estimating the parameters of such dynamics. Indeed, in order to estimate the parameters of both the continuous and the spiked noise, we need first to be able to distinguish them. There exist different methods of jump filtering. A first intuitive one is to settle a threshold, for example 3 standard deviations, such that data points within this threshold are considered to belong to the continuous part while the data points above correspond to the jumps. Cartea and Figueroa \cite{cartea2005pricing} and Pawłowski and Nowak~\cite{pawlowski2015modelling} use an iterative method of filtering based on such threshold.  However, the choice of the threshold seems arbitrary and integrates a standard deviation that itself combines continuous and spiked noises. Figure~\ref{fig:filter_jumps} shows the residuals of Regression~\eqref{eq:ref-1} filtered through the above method. The residuals categorized as continuous suit well the normal quantiles however the jumps are pretty sparse and, hence, difficult to fit. We explored this method but we hardly could estimate the jump parameters convincingly and found the filtering criteria rather arbitrary.
\begin{figure}[ht]
    \begin{minipage}[b]{0.3\linewidth}
        \centering
        \includegraphics[width=\textwidth]{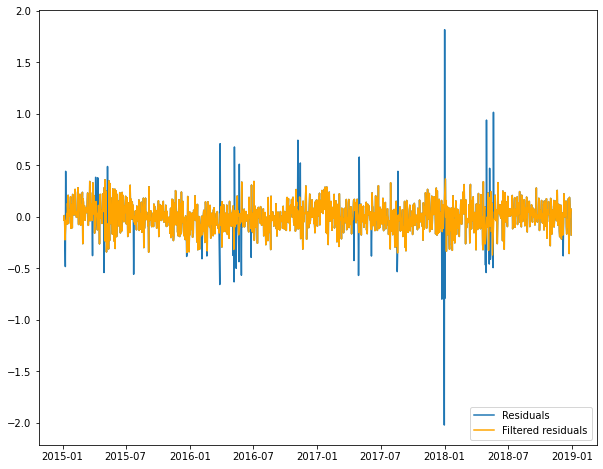}
    \end{minipage}
    \hspace{0.3cm}
    \begin{minipage}[b]{0.3\linewidth}
        \centering
        \includegraphics[width=\textwidth]{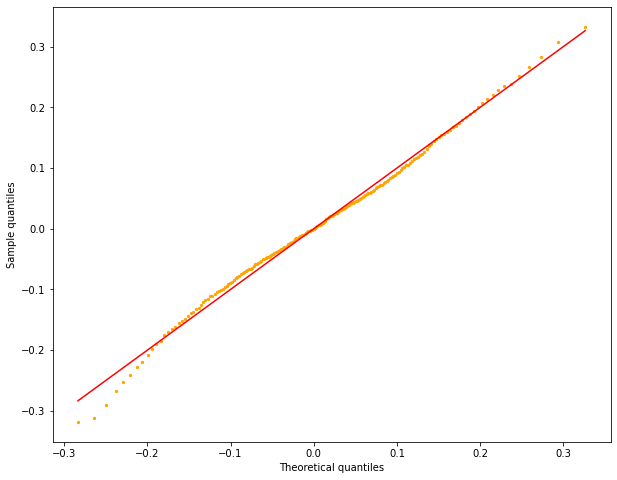}
    \end{minipage}
        \begin{minipage}[b]{0.3\linewidth}
        \centering
        \includegraphics[width=\textwidth]{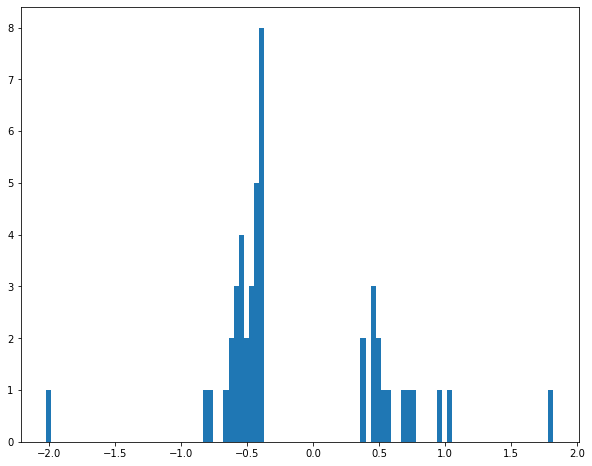}
    \end{minipage}
    \caption{Jump filtering method as described by Cartea and Figueroa~\cite{cartea2005pricing} (left). Filtered continuous residuals qqplot against a normal distribution (center). Histogram of jumps (right).}\label{fig:filter_jumps}
\end{figure}
We also considered alternative jump estimation methods. Meyer-Brandis and Tankov~\cite{meyer2008multi} challenge the threshold method in their two factor model an suggest two alternative filtering algorithm. Their implementation involves a hypothesis on the stochastic nature of the processes and the number of spikes. Alternatively, Deng~\cite{deng2000stochastic} suggests to implement the method of moments which can introduce bias. Finally, Geman and Roncoroni~\cite{geman2006understanding} apply a maximum likelihood estimation applied to an unknown process using a prior reference process. This method is highly dependent on the underlying distribution and choice of priors.

\subsection{Multi-factor mean-reverting models}

Multi-factor models with non-Gaussian increments represent another popular alternative to model erratic dynamics. Two factors and three factors models with Gaussian increments were developed by Schwartz through different collaborations~\cite{schwartz1997stochastic}~\cite{schwartz2000short}~\cite{gibson1990stochastic}~\cite{lucia2002electricity}. The idea behind is that the spot prices could be driven by a long-term and a short-term dynamics, so that the spot price would integrate long-run terms and potential circumstantial tensions on the energy market. The estimation of such models can be performed through classic Kalman filtering. However, they do not answer the issue of bad fitting of the residuals with normal distribution as can be observed on the right figure of Figure~\ref{fig:fit_rs}.

Taking inspiration of Schwartz's models, several papers have explored the possibility to combine multi-factor models with Levy processes~\cite{bjork2002term}~\cite{meyer2008multi}. The adaptation of~\eqref{eq:diff} to a multi-factor model of $n$ factors takes the following form:
\begin{equation}
 \begin{cases}
X_t &= \mu(t) + \sum \tilde{X}_t^n \\
\tilde{X}_t^n &= -\kappa_n \tilde{X}_t^n dt + d L_t^n
\end{cases} \label{eq:diff_rs}
\end{equation}
where $L_t^n$ corresponds to a Levy process.\\
While Björk and Landén~\cite{bjork2002term} deduce analytical expression for forward contracts with compound Poisson processes, they do not address the estimation challenges. On their side, Meyer-Brandis and Tankov~\cite{meyer2008multi} study two factor models with non-defined Levy processes and suggest a calibration based on the autocorrelation function. Figure~\ref{fig:fit_rs} shows the autocorrelation function and the exponential fitting enabling to compute the autoregressive parameters. It can be observed that the autocorrelation function does not present a clear exponential shape. Furthermore the fitting is extremely sensitive to the time horizon considered for the autocorrelation function.

\begin{figure}[!htb] 
    \begin{minipage}[b]{0.48\linewidth}
        \centering
        \includegraphics[width=\textwidth]{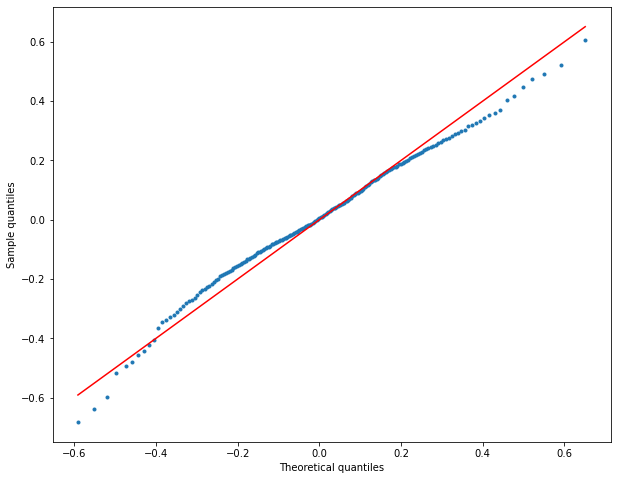}
    \end{minipage}
    \begin{minipage}[b]{0.46\linewidth}
        \centering
        \includegraphics[width=\textwidth]{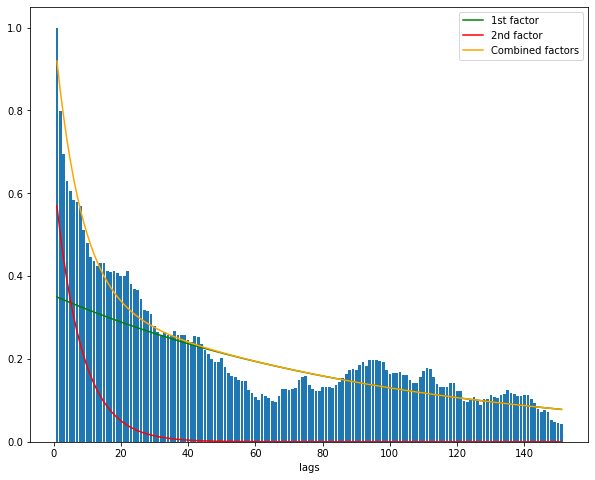}
    \end{minipage}
    \begin{minipage}[b]{0.48\linewidth}
        \centering
        \includegraphics[width=\textwidth]{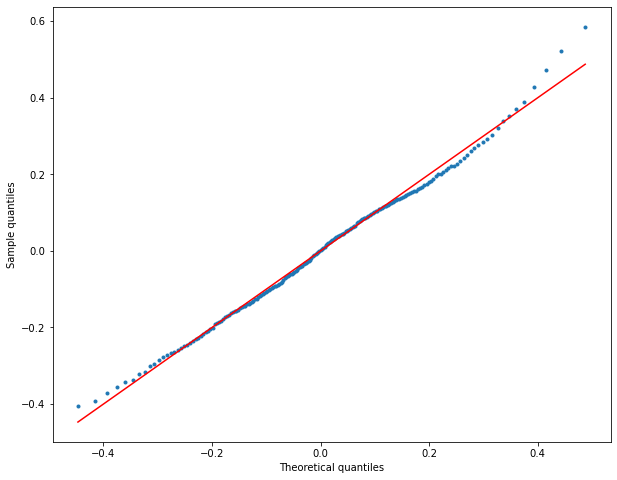}
    \end{minipage}
    \hspace{0.75cm}
    \begin{minipage}[b]{0.46\linewidth}
        \centering
        \includegraphics[width=\textwidth]{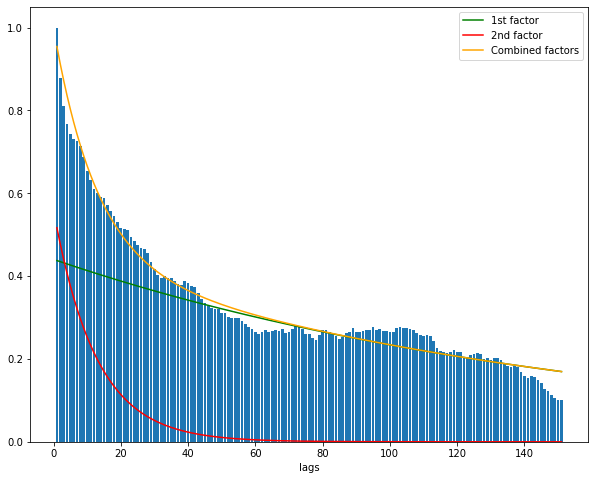}
    \end{minipage}
    \caption{On the left, quantile quantile plots for residuals after removal of the two autocorrelation dynamics compared with a theoretical quantiles of a normal distribution. On the right, two factor model fitted through the autocorrelation function as described in ~\cite{meyer2008multi}. First row corresponds to French energy data while second row to North Italy.} \label{fig:fit_rs}
\end{figure}

Finally, the inclusion of more than one autoregressive factors blocks from obtaining an integrant form of the dynamics. Given the latter and the lack of robust fitting method for multi-factor models with Levy processes, we prefer to focus on one factor models.

\subsection{Mean-reverting diffusion models with NIG noise}
As suggested by Benth and Benth~\cite{benth2004normal}, we explore non-Gaussian Ornstein-Uhlenbeck process. The Normal Inverse Gaussian distribution was first introduced by Barndorff-Nielsen~\cite{barndorff1997normal} and presents the distribution and properties in Appendix~\ref{appendix:NIG}. 
The first motivation to keep this model is a rather good fitting of the residuals as can been seen in Figure~\ref{fig:qqplot}. While residuals corresponding to $\int_s^t e^{-\kappa(t-u)} dL_u$ are not exactly distributed in NIG, the estimation of tinniness $\kappa$ suggested a very close distribution that we find convincing.
The second motivation to keep this model consisted in the relative easy estimation of the parameters. Section~\ref{sec:calibration} will focus on these challenges.\\
Finally, this model can be easily generalized to multivariate distributions which will enable us to foresee combinations of energy and weather parameters dynamics in order to compute hybrid options.\\
As rightly noted by Benth and Benth~\cite{benth2004normal}, OU-NIG dynamics correspond to a first step towards stochastic volatility models that are common models for commodity derivative modelling \cite{deng2000stochastic} \cite{trolle2009unspanned} \cite{benth2011stochastic}. While these models are usually developed for commodity forwards under the Heath–Jarrow–Morton (HJM) framework, the methodology can be replicated to spot price modeling \cite{lari2001mean}. Nevertheless, stochastic volatility models imply observing daily volatility on the modeled variable. While this is possible for energy spot prices, our objective is to combine this dynamic with daily average temperature dynamics for which the volatility is unobserved~\cite{alfonsi2022stochastic}. The following section will therefore focus on a two dimensional Ornstein-Uhlenbeck dynamic with NIG noise.

\section{Separate parameter estimation of the two marginal processes}\label{sec:calibration}

This section focuses on the estimation challenges of our two marginals: the daily day-ahead energy log spot price $(X_t)_{t\geq0}$ and the average daily temperature $(T_t)_{t\geq0}$. We assume that $\lambda=0$ in Section~\ref{sec:calibration}, which gives the independence of these processes and allow to estimate their parameters separately. The joint estimation when $\lambda \not =0$ will be discussed in the next section.  

\subsection{Estimation of $\kappa_X$ and $\mu_X(\cdot)$}\label{subsec:kappa_mu}
Following Klimko {and} Nelson \cite{klimko1978conditional}, the objective of this section is to estimate $\kappa_X$ and $\mu_X(\cdot)$ using Conditional Least Squares Estimation (CLSE). For this, we first write the conditional expectation of  $(X_t)_{t\geq0}$. From
$$X_{t+\Delta}-\mu_X(t+\Delta)=e^{-\kappa_X \Delta }(
  X_t-\mu_X(t))+ \int_t^{t+\Delta} e^{-\kappa_X({t+\Delta}-u)} dL^X_u, $$
  we get
\begin{equation}
    \E[X_{t+\Delta}-\mu_X(t+\Delta)|\cF_t]=e^{-\kappa_X \Delta}(
  X_t-\mu_X(t)),
\end{equation}
since we consider that $L^X$ is centered, that is $m^X +\delta^X \beta^X /\gamma^X=0$. We then get the following expression for the conditional expectation:
\begin{equation} \label{eq:expectation}
    \E[X_{t+\Delta}|\cF_t]=\mu_X(t+\Delta) + e^{-\kappa_X \Delta}(X_t-\mu_X(t))
\end{equation}
where $\mu_X(t) = \beta^X_0 t + \alpha^X_{1} \sin(\xi t) + \beta^X_{1} \cos(\xi t) +  \alpha^{X,DoW}_{DoW(t)}$ where  $\xi = \frac{2 \pi}{365}$ and $DoW(t)=\lfloor \frac{t}{\Delta} \rfloor$ mod $7$.

We can now apply CLSE to the discrete for form of Equation~\eqref{eq:expectation} which boils down to minimise 
\begin{equation} \label{eq:LR_stop}
\sum_{i=0}^{N-1} \left( X_{(i+1)\Delta} - \mathbb{E}  [X_{(i + 1)\Delta} | X_{i\Delta} ] \right)^2.
\end{equation} 
\noindent This can be solved through linear regression, and Proposition~\ref{prop_estimX} gives:
\begin{equation}
\begin{cases}\label{def_estim}
  \hat{\kappa}_X &= - \ln{\hat{\eta} _2} \\
  \hat{\beta}^X_0&= \frac{\hat{\eta}_1}{1-\hat{\eta}_2} \\
  \hat{\alpha}^X_1 &= \frac{\hat{\eta}_3 (\cos(\xi \Delta )-e^{-\hat{\kappa}_X\Delta}) + \hat{\eta}_4 \sin(\xi \Delta )}{(\cos(\xi \Delta )-e^{-\hat{\kappa}_X\Delta})^2 + \sin^2(\xi \Delta )}\\
  \hat{\beta}^X_1 &= \frac{\hat{\eta}_4 (\cos(\xi \Delta )-e^{-\hat{\kappa}_X \Delta }) - \hat{\eta}_3 \sin(\xi \Delta )}{(\cos(\xi \Delta )-e^{-\hat{\kappa}_X \Delta})^2 + \sin^2(\xi \Delta )}\\
  \hat{\alpha}^{X,DoW}_{j}&= \frac{1}{1-e^{- 7\hat{\kappa}_X \Delta }} \sum_{k=0}^6 (\hat{\eta}^{DoW}_{j + k} - \hat{\beta}_0) e^{-(6-k)\hat{\kappa}_X \Delta }, \\
\end{cases}
\end{equation}
where 
\begin{equation}
\label{eq:def_hlambda}
\hat{\eta} = \left( \sum_{i=0}^{N-1} \Xi_{i \Delta} \Xi^\top_{i \Delta}  \right)^{-1} \left(\sum_{i=0}^{N-1} \Xi_{i \Delta}  X_{(i + 1)\Delta} \right) ,
\end{equation}
with $\Xi_{i\Delta} = (i\Delta,  X_{i\Delta}, \sin(\xi  i\Delta), \cos(\xi  i\Delta), (\mathbbm{1}_{\{DoW(i\Delta)=j\}})_{0\le j \le 6} ) \in \R^4\times \{0,1\}^7$ for $i \in \N$ and \\
$(\hat{\eta}^{DoW}_0,\dots,\hat{\eta}^{DoW}_6)=(\hat{\eta}_5,\dots,\hat{\eta}_{11})$ and $\hat{\eta}^{DoW}_j=\hat{\eta}^{DoW}_{\tilde{j}}$, with $\tilde{j}\in \{0,\dots 6\}$ such that $j=\tilde{j} \mod 7$.

Figure~\ref{fig:fit_energy} represents the fitted trend and seasonality component $\mu_X(\cdot)$ in the original logarithm of spot prices. We can observe an important difference of $\mu_X(\cdot)$ between week-ends and week days. In particular, Saturdays and Sundays are particularly cheaper days. Although energy still is traded on weekends, the volume is smaller which explains the different behaviours also noted by Meyer-Brandis and Tankov~\cite{meyer2008multi}.
Finally, Table~\ref{tab:LR_parameter_estimates} shows estimated parameters for French and North Italian data.

\begin{figure}[!htb]
    \begin{minipage}[b]{0.45\linewidth}
        \centering
        \includegraphics[width=\textwidth]{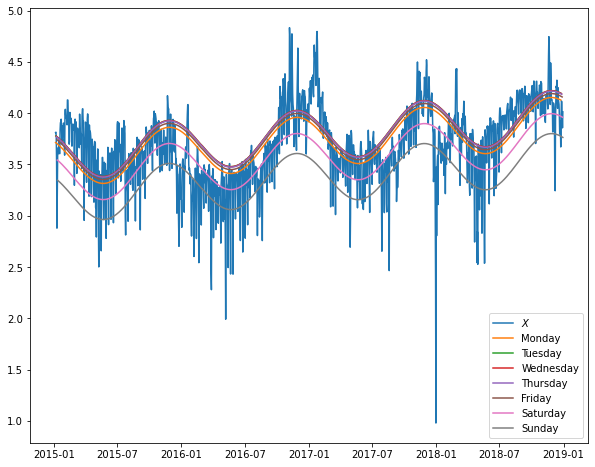}
    \end{minipage}
    \hspace{0.3cm}
    \begin{minipage}[b]{0.45\linewidth}
        \centering
        \includegraphics[width=\textwidth]{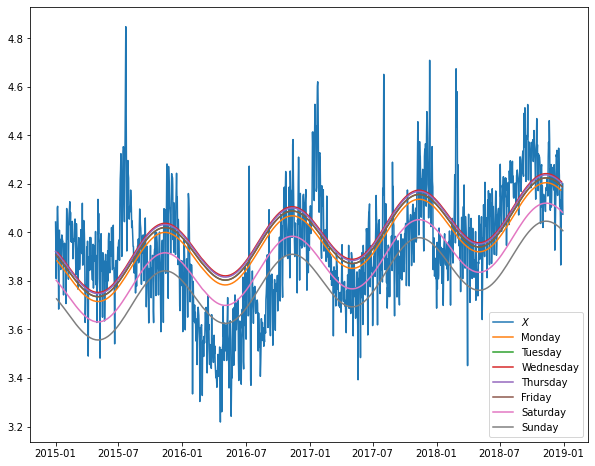}
    \end{minipage}
    \caption{Fitted deterministic curve $\mu_X(\cdot)$ for France (left) and Italy (right).} \label{fig:fit_energy}
\end{figure}

\begin{table}[h]
\centering
\begin{tabular}{cccc}
 & $\hat{\kappa}$ & Residuals mean   & Residuals SD \\ \hline \hline
France & $0.226$ & $9.614 \times 10^{-16}$ & $0.171$  \\
North Italy & $0.129$ & $5.056 \times 10^{-16}$ & $0.099$  
\end{tabular}
\caption{Fitted $\hat{\kappa}$ and residuals of Regression~\eqref{eq:LR_stop} for France (left) and Italy (right) }
\label{tab:LR_parameter_estimates}
\end{table}

\subsection{Parameter estimation for the NIG noise} \label{subsec:est_NIG}
We now move to the estimation of the parameters of the NIG process. Let first recall that:
\begin{align}
\tilde{X}_{t+\Delta}=e^{-\kappa_X \Delta}\tilde{X}_t+ \int_t^{t+\Delta} e^{-\kappa_X (t+\Delta-v)} dL^X_v \mbox{ where } \tilde{X}_t = X_{t} - \mu_X(t)
\end{align}
We can first observe that the residuals we study correspond to $\int_s^t e^{-\kappa_X(t-v)} dL^X_v$ and not the Levy noise per se. There is some study on the distribution of this integral particularly as Tempered Stable Processes, see Sabino~\cite{sabino2022fast}. However, given the inability to compute an explicit formula, we decided to work with approximations of this integrant. The following section will concentrate in three estimation methodologies: CLS, maximum likelihood and EM algorithm applied to a second order approximation of the characteristic function.

Let first study the form of the characteristic function of our process $(\tilde{X}_t)_{t\geq0}$. We have
\begin{equation*}
    \E(e^{i u \tilde{X}_{t+\Delta}})=\E\left(\exp(i u [e^{-\kappa_X \Delta} \tilde{X}_t + \int_t^{t+\Delta} e^{-\kappa_X (t+\Delta-v)} dL^X_v ])\right) 
\end{equation*}
If we focus on the second term and use Lemma 4.1. on Benth and Benth~\cite{benth2004normal}. We call $\varphi$ the characteristic function and we can write as follows:
\begin{equation}\label{eq_def_phi}
    \varphi(u ;\Delta) = \E(\exp(i u  \int_t^{t+\Delta} e^{-\kappa_X (t+\Delta-v)} dL^X_v )) = \exp \left( \int_t^{t+\Delta} \psi(iue^{-\kappa_X(t+\Delta-v)}) dv \right)
\end{equation}
where $\psi$ corresponds to the cumulant function of NIG distributions and is given by
\begin{equation*}
    \psi(x) =  x m^X + \delta^X(\gamma^X - \sqrt{(\alpha^X)^2- (\beta^X+ x)^2})
\end{equation*}
Finally, we have:
\begin{equation} \label{eq:charac_nig}
    \begin{aligned}
 \varphi(u ; \Delta) 
 &= \exp \Bigg(i u m^X \frac{1- e^{-\kappa_X\Delta}}{\kappa_X} + \delta \gamma^X \Delta - \delta^X \int_t^{t+\Delta} \sqrt{(\alpha^X)^2- (\beta^X+ iue^{-\kappa_X(t +\Delta- v)})^2} dv \Bigg)
\end{aligned}
\end{equation}

\paragraph{Estimation through CLS}
In this paragraph, we apply the CLS method developed by Klimko and Nelson~\cite{klimko1978conditional} to the characteristic function $\varphi$. Our objective is to minimise the below function for $u \in \R$.
\begin{equation} \label{eq:CLS_charac}
     \Big|e^{i u ( \tilde{X}_{t+\Delta} -e^{-\kappa_X \Delta} \tilde{X}_t)}- \varphi(u ; \Delta) \Big|^2
\end{equation}
Now let consider a discrete time interval, the objective function in Equation~\eqref{eq:CLS_charac} becomes:
\begin{equation*}
    \sum_{t=0}^{N-1}  \Big|e^{i u ( \tilde{X}_{t+\Delta} -e^{-\kappa_X \Delta} \tilde{X}_t)}- \varphi(u ; \Delta) \Big|^2
\end{equation*}
In our case, we minimise for different values of $u$ the below objective function:
\begin{equation}
    \sum_u  \sum_{t=0}^{N-1}  \Big|e^{i u ( \tilde{X}_{t+\Delta} -e^{-\kappa_X \Delta} \tilde{X}_t)}- \varphi(u ; \Delta) \Big|^2
    \label{eq:charac}
\end{equation}
We compute the characteristic function through numerical integration using the function \texttt{quad} in \texttt{Python} and $u$ is taken in $\{-5,-4,..., 5\}$.

\begin{remark}
    We contemplated to approximate the characteristic function through a Simpson's integration method. We tested how the choice of this approximation impacts the estimation by comparing with with the exact numerical integration. This method lead to similar  results (see Table~\ref{tab:parameter_estimates}) and not particularly quicker (it takes 0.53s (resp. 0.63s) to minimize~\eqref{eq:charac} with Simpson's on France (resp. North Italy) data) instead of 1.14s (resp. 1.25s) with exact integration). Therefore, on the following we keep using the exact numerical integration.
\end{remark}

The minimisation algorithm we apply is the Nelder–Mead algorithm~\cite{nelder1965simplex}. We choose this method because it enables to integrate boundary constraints such as $\alpha^X>0$ and $\alpha^X \geq |\beta^X|$ and shows good convergence. Table~\ref{tab:parameter_estimates} summarizes the results of the minimisation.\\

The reader can also note that the above methodology was validated with simulated data. We simulated $100,000$ Normal Inverse Gaussian simulations with predefined parameters, computed the corresponding integral and verify that the minimisation of metric~\eqref{eq:charac} lead to the correct parameter estimates. The sensitivity of this method was also tested with convincing results. We leave the mathematical proof of the characteristics of these estimators as a possible further research. 
\paragraph{MLE and EM-algorithm estimation through second order approximation}
An alternative approach for Normal Inverse Gaussian parameter estimation is to implement the Expectation-Maximization algorithm (EM) or maximum likelihood estimation (MLE). Unfortunately we do not know explicitly the density of $\int_t^{t+\Delta} e^{-\kappa_X (t+\Delta-v)} dL^X_v$. We hence proceed to an approximation of this density.

For this, let consider the characteristic function $\varphi(\cdot, \Delta)$ in Equation~\eqref{eq:charac_nig}. \red{We use the midpoint rule to approximate the integral and get:
\begin{align*}
    \varphi(u; \Delta) &=  \exp\Bigg(\Delta \Bigg(i u m^X e^{-\kappa_X \Delta/2} + \delta^X \gamma^X  - \delta^X \sqrt{(\alpha^X)^2- (\beta^X+ iue^{-\kappa_X\Delta/2})^2}  \Bigg)  +O(\Delta^3) \Bigg) \\
    &=(1+O(\Delta^2))\times \exp\Bigg(\Delta \Bigg(i u m^X e^{-\kappa_X \Delta/2} + \delta^X \gamma^X  - \delta^X \sqrt{(\alpha^X)^2- (\beta^X+ iue^{-\kappa_X\Delta/2})^2}  \Bigg) 
\end{align*}
We now  apply the transformation,
\begin{equation*}
 \begin{cases}
    \tilde{m}^X &= \Delta m^X \\
    \tilde{\delta}^X &= \Delta \delta^X \\
    \tilde{u} &= u e^{-\kappa_X \Delta/2}
\end{cases}   
\end{equation*}
to obtain the below approximation for the characteristic function:
\begin{align*}
    \E(\exp(i \tilde{u} e^{\kappa_X \Delta/2}  &\int_t^{t +\Delta} e^{-\kappa_X(t+\Delta-v)} dL^X_v )) \\
    &=(1+O(\Delta^2))\times \exp\left(i \tilde{u}\tilde{m}^X  + \tilde{\delta}^X \gamma^X  - \tilde{\delta}^X  \sqrt{(\alpha^X)^2- (\beta^X+ i\tilde{u})^2}  \right).
\end{align*}
Therefore, $\left(e^{\kappa_X \Delta/2}\int_{\ell \Delta}^{(\ell+1)\Delta} e^{-\kappa_X((\ell+1)\Delta-v)} dL^X_v \right)_{\ell\ge 1}$ behaves, at second order approximation in~$\Delta$, as a sample of independent NIG distribution with parameters $\alpha^X$, $\beta^X$, $\tilde{m}^X$ and $\tilde{\delta}^X$.} The EM algorithm or maximum likelihood estimation is therefore applied to this time series. The Maximum Likelihood Estimation (MLE) is applied through the function pre-implemented in the Python library \texttt{scipy.stats} while the EM-algorithm is implemented following Karlis~\cite{karlis2002type}. 
\begin{table}[h]
\begin{minipage}{.5\linewidth}
      \centering
\begin{tabular}{ccccc}
Method & $\hat{\alpha}^X$ & $\hat{\beta}^X$   & $\hat{m}^X$ & $\hat{\delta}^X$ \\ \hline \hline
CLS Simp & 4.222 & -0.361 & 0.011& 0.128\\
CLS Num & 4.222 & -0.361 & 0.011& 0.128\\
EM & 4.066 & -0.340& 0.011 & 0.131 \\
MLE & 4.013 & -0.203 & 0.007 & 0.130 
\end{tabular}
\end{minipage}%
    \begin{minipage}{.5\linewidth}  
    \centering
    \begin{tabular}{ccccc}
Method & $\hat{\alpha}^X$ & $\hat{\beta}^X$   & $\hat{m}^X$ & $\hat{\delta}^X$ \\ \hline \hline
CLS Simp & 13.736 & 0.319 & 0.004 & 0.152  \\
CLS Num & 13.736 & 0.319 & 0.004 & 0.152 \\
EM & 12.576 & 0.533 & -0.006 & 0.140\\
MLE & 12.575 & 0.533 & -0.006 & 0.140
\end{tabular}
\end{minipage}
\caption{Parameter estimation through the CLS, EM and maximum likelihood estimation for French (left) and Italian (right) energy log spot price.}
\label{tab:parameter_estimates}
\end{table}

Table~\ref{tab:parameter_estimates} shows the parameter estimations through the CLS method, EM algorithm and maximum likelihood minimisation. We choose to move forward with the CLS Numerical method as all method lead to similar results and EM and MLE method cannot be directly applied to the combined model. From here onwards, MLE estimates are mainly used to initialize the CLS minimisation problem~\eqref{eq:charac}.

\subsection{Estimation of $\kappa_T$ and $\mu_T(\cdot)$}
As introduced in Section 1, we consider that $(T_t)_{t\geq0}$ follows an Ornstein Uhlenbeck with Gaussian residuals such that:
\begin{equation*}
d(T_t-\mu_T(t))=-\kappa_T(T_t-\mu_T(t))dt + dW^T_t
\label{eq:TNIG}    
\end{equation*}
where $\mu_T(\cdot)$ represents the trend and seasonality component such that $\mu_T(t) = \alpha^T_0 + \beta^T_0 t + \alpha^T_{1} \sin(\xi t) + \beta^T_{1} \cos(\xi t)$, $\kappa_T$ the autoregressive parameter of the OU and $W^T_t$ is a Brownian motion.

First, similarly to Subsection~\ref{subsec:kappa_mu}, we implement CLS estimation to $(T_t)_{t\geq0}$ by minimising:
\begin{equation} \label{eq:LR_temp}
\sum_{i=0}^{N-1} \left( T_{(i+1)\Delta} - \mathbb{E}  [T_{(i + 1)\Delta} | T_{i\Delta} ] \right)^2.
\end{equation}
This can be solved through linear regression to obtain (see \cite[Proposition C.1]{alfonsi2022stochastic})
\begin{equation}
\begin{cases}\label{def_estim_T}
  \hat{\kappa}_T &= - \ln{\hat{\zeta} _2} \\
  \hat{\alpha}^T_0&= \frac{\hat{\zeta}_0 }{1-\hat{\zeta}_2} -  \frac{\hat{\zeta}_1 }{(1-\hat{\zeta}_2)^2}\\
  \hat{\beta}^T_0&= \frac{\hat{\zeta}_1}{1-\hat{\zeta}_2} \\
  \hat{\alpha}^T_1 &= \frac{\hat{\zeta}_3 (\cos(\xi \Delta )-e^{-\hat{\kappa}_T \Delta }) + \hat{\zeta}_4 \sin(\xi \Delta)}{(\cos(\xi \Delta )-e^{-\hat{\kappa}_T \Delta})^2 + \sin^2(\xi \Delta )}\\
  \hat{\beta}^T_1 &= \frac{\hat{\zeta}_4 (\cos(\xi \Delta )-e^{-\hat{\kappa}_T \Delta }) - \hat{\zeta}_3 \sin(\xi \Delta )}{(\cos(\xi \Delta )-e^{-\hat{\kappa}_T\Delta })^2 + \sin^2(\xi \Delta )}, \\
\end{cases}
\end{equation}
where 
\begin{equation}
\label{eq:def_hzeta}
\hat{\zeta} = \left( \sum_{i=0}^{N-1} \Pi_{i \Delta} \Pi^\top_{i \Delta}  \right)^{-1} \left(\sum_{i=0}^{N-1} \Pi_{i \Delta}  T_{(i + 1) \Delta} \right)
\end{equation}
with 
$\Pi_{i\Delta} = (1, i \Delta,  T_{i\Delta}, \sin(\xi  i \Delta), \cos(\xi  i \Delta)) \in \R^5$ for $i \in \N$. \\
Figure~\ref{fig:fit_temp} represents the daily average temperature and fitted function $\mu_T(\cdot)$. Table~\ref{tab:LR_temp_parameter_estimates} presents the estimates of $\kappa_T$ and residuals of Regression~\eqref{eq:LR_temp} for Paris and Milan. We can observe that the residuals are well centered.

\begin{table}[h]
\centering
\begin{tabular}{cccc}
 & $\hat{\kappa}_T$ & Residuals mean   & Residuals SD \\ \hline \hline
France & $0.254$ & $-1.741 \times 10^{-15}$ & $2.138$  \\
North Italy & $0.250$  & $-1.138 \times 10^{-16}$ & $1.638$  
\end{tabular}
\caption{Fitted $\hat{\kappa}$ and residuals of Regression~\eqref{eq:LR_temp} for France (left) and Italy (right).}
\label{tab:LR_temp_parameter_estimates}
\end{table}

\begin{figure}[!htb]
    \begin{minipage}[b]{0.45\linewidth}
        \centering
        \includegraphics[width=\textwidth]{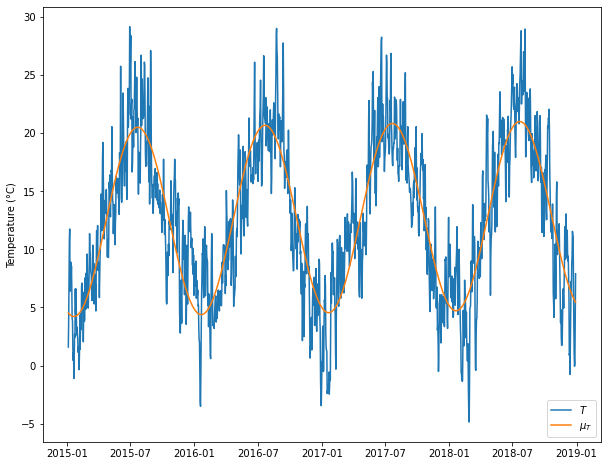}
    \end{minipage}
    \hspace{0.3cm}
    \begin{minipage}[b]{0.45\linewidth}
        \centering
        \includegraphics[width=\textwidth]{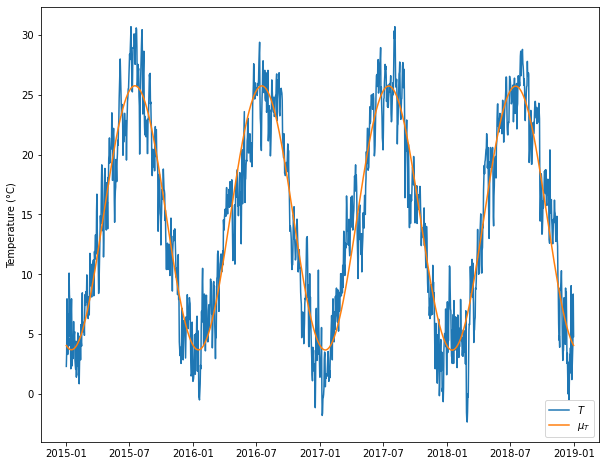}
    \end{minipage}
    \caption{Fitted deterministic curve $\mu_T(\cdot)$ for France (left) and Italy (right).} \label{fig:fit_temp}
\end{figure}
From here onwards, we will note $\tilde{T}$ the deseasonalised temperature such that $\tilde{T}_t = T_t-\mu_T(t)$. 

\subsection{Parameter estimation for the temperature residuals} 

Our first idea was to implement a multivariate NIG for $(\tilde{X}_t, \tilde{T}_t)$. We first then test if $(\tilde{T}_t)_{t\geq0}$ can also follow a Ornstein Uhlenbeck process with NIG noise and finally show that Gaussian noises are more stable and reliable. 

To estimate NIG parameters on $(\tilde{T}_t)_{t\geq0}$, we implement the CLS approach as in Subsection~\ref{subsec:est_NIG}.

\begin{table}[h]
\begin{minipage}{.5\linewidth}
      \centering
\begin{tabular}{ccccc}
Method & $\hat{\alpha}^T$ & $\hat{\beta}^T$   & $\hat{m}^T$ & $\hat{\delta}^T$ \\ \hline \hline
CLS & 14.095 &  $10^{-3}$   & -0.002   & 98.610 \\
\end{tabular}
\hspace{3mm}
\end{minipage}%
    \begin{minipage}{.5\linewidth}  
    \centering
    \begin{tabular}{ccccc}
Method & $\hat{\alpha}^T$ & $\hat{\beta}^T$   & $\hat{m}^T$ & $\hat{\delta}^T$ \\ \hline \hline
CLS & 38.152 &  $10^{-9}$  &  $10^{-9}$   & 111.119
\end{tabular}
\end{minipage}
\caption{Parameter estimation through the CLS for French (left) and North Italian (right) temperature.}
\label{tab:parameter_estimates_temp}
\end{table}

Table~\ref{tab:parameter_estimates_temp} shows the parameter estimations through CLS initiated with the MLE estimates. The results are quite unstable and particularly because $\beta$ shrinks towards 0.  Indeed, we are here confronted to the special case where $(\tilde{T}_t)_{t\geq0}$ nearly follows a Normal distribution. The analysis of quantile quantile plots in Figure~\ref{fig:qqplot_temp} show indeed that a Normal regression fits considerably well to the residuals regression~\ref{eq:LR_temp}. This result is aligned with Larsson's conclusions on German mean temperatures~\cite{larsson2023parametric}.

\begin{figure}[htb!]  
  \begin{minipage}[b]{0.5\linewidth}
    \centering
    \includegraphics[width=0.8\textwidth]{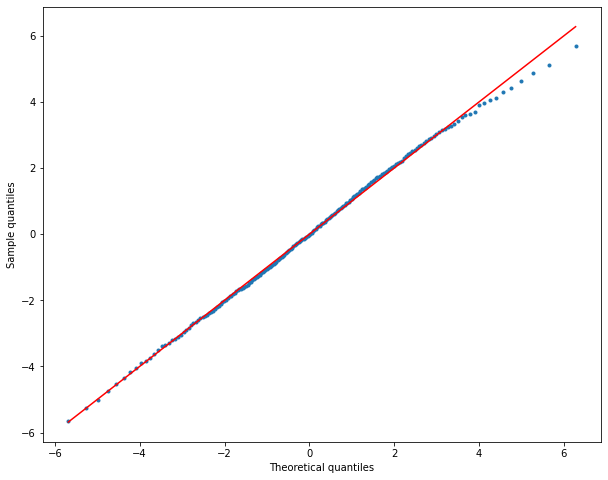}
  \end{minipage}
  \begin{minipage}[b]{0.5\linewidth}
    \centering
    \includegraphics[width=0.8\textwidth]{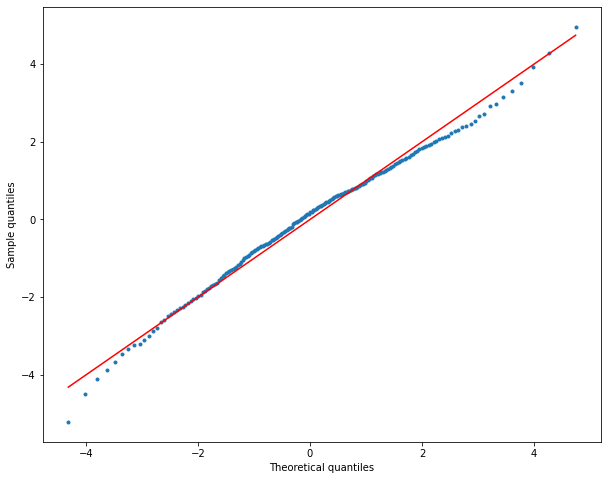}
  \end{minipage}
  \caption{Quantile quantile plots for residuals of regression~\ref{eq:LR_temp} compared with a theoretical quantiles of a normal distribution for Paris temperatures (left) and Milan temperatures (right).} \label{fig:qqplot_temp}
\end{figure}
Hence, we choose to stick to a Normal distribution for $(T_t)_{t\geq0}$. To estimate its parameters, we have that $W^T$ is a Brownian noise $\sim \mathcal{N}(m^T,\sigma_T^2)$. Hence, we have  $\int_t^{t+\Delta} e^{-\kappa_T(t+\Delta-v)} dW^T_v \sim \mathcal{N}(m^T \sqrt{\frac{1 - e^{-2\kappa_T\Delta}}{2\kappa_T}}, \sigma_T^2 \frac{1 - e^{-2\kappa_T\Delta}}{2\kappa_T})$. We can obtain the parameters of the Normal distribution by correcting the residuals of Regression~\eqref{eq:LR_temp} by the factor $\sqrt{\frac{2\kappa_T}{1 - e^{-2\kappa_T\Delta}}}$ . Table~\ref{tab:parameter_estimates_temp_gauss} summarized the results.

\begin{table}[h]
\begin{minipage}{.5\linewidth}
      \centering
\begin{tabular}{ccccc}
Method &  $\hat{m}^T$ & $\hat{\sigma}_T^2$ \\ \hline \hline
MLE & $10^{-15}$ & $2.413$
\end{tabular}
\end{minipage}%
    \begin{minipage}{.5\linewidth}  
    \centering
    \begin{tabular}{ccccc}
Method &  $\hat{m}^T$ & $\hat{\sigma}_T^2$ \\ \hline \hline
MLE & $10^{-16}$ & $ 1.846$
\end{tabular}
\end{minipage}
\caption{Parameter estimation through the maximum likelihood estimation for dynamic of temperature normally distributed for Paris (left) and Milan (right) temperature.}
\label{tab:parameter_estimates_temp_gauss}
\end{table}

\section{Towards a combined model for $(\tilde{X}_t, \tilde{T}_t)$}\label{Sec_modele2D}

In the previous section, we have considered separate models for the electricity spot price and the temperature. This has enabled us to estimate the trend functions $\mu_X(\cdot)$ and $\mu_T(\cdot)$ and the speeds of mean-reversion $\kappa_X$ and $\kappa_T$. We now consider the joint model~\eqref{eq:comb} with $\lambda \not = 0$. First, we show empirical results on the dependence between electricity prices and the temperature. Then, we propose an estimation procedure of the different parameters. Since the temperature follows an autonomous dynamics in Model~\eqref{eq:comb}, the estimation of $\kappa_T$, $\mu_T(\cdot)$ and $\sigma_T$ is unchanged from the previous section. By Proposition~\ref{prop_estimX}, the least square estimation of $\kappa_X$ and $\mu_X(\cdot)$ is also unchanged when $\lambda\not=0$. Therefore, this section focuses on the estimation of $\lambda$ and then on the estimation of the NIG parameters when $\lambda \not =0$.

\subsection{Test of dependence}
This section analyses the significance of the dependence structure between $(\tilde{X}_t)_{t\geq0}$ and $(\tilde{T}_t)_{t\geq0}$. For this we first estimate the Pearson correlation between the residuals $(\tilde{X}_{(i+1)\Delta}-e^{-\kappa_X \Delta} \tilde{X}_{i\Delta})$ and $(\tilde{T}_{(i+1)\Delta}-e^{-\kappa_T \Delta} \tilde{T}_{i\Delta})$. We obtain a correlation equal to $-0.087$ for France and $-0.043$ for north Italy, which suggests a small dependence.  Figure~\ref{fig:residuals_cloud} shows standardized residuals and ranked residuals plots for France and Italy. We cannot observe a clear dependence structure through these plots which supports the low correlations that we have obtained. 

To better analyse the dependence, Table~\ref{tab:freq_tercile} shows the frequencies of ranked residuals given a tercile classification. This time we can observe a slight anti-correlation as left top and bottom low corners are more populated than right top and bottom left corners in both cases.

\begin{figure}[htb!]  
  \begin{minipage}[b]{0.5\linewidth}
    \centering
    \includegraphics[width=0.8\linewidth]{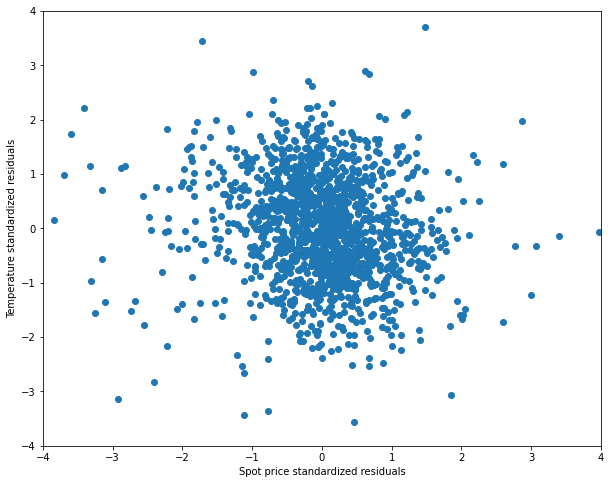} 
    \vspace{4ex}
  \end{minipage}
  \begin{minipage}[b]{0.5\linewidth}
    \centering
    \includegraphics[width=0.8\linewidth]{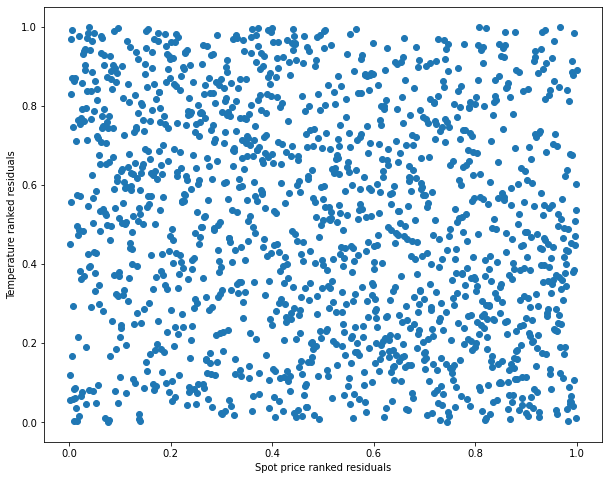} 
    \vspace{4ex}
  \end{minipage} 
  \begin{minipage}[b]{0.5\linewidth}
    \centering
    \includegraphics[width=0.8\linewidth]{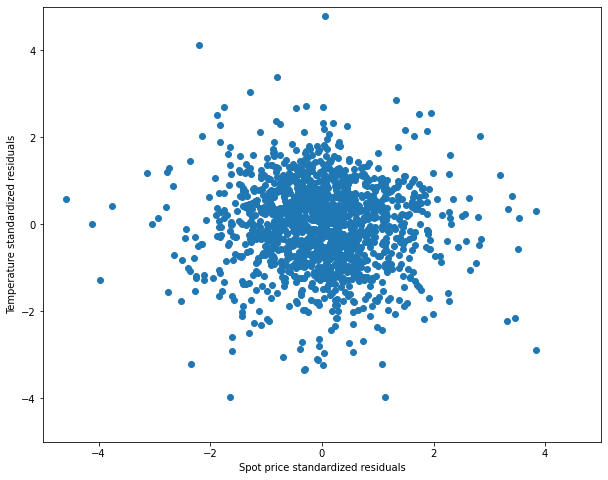} 
    \vspace{4ex}
  \end{minipage}
  \begin{minipage}[b]{0.5\linewidth}
    \centering
    \includegraphics[width=0.8\linewidth]{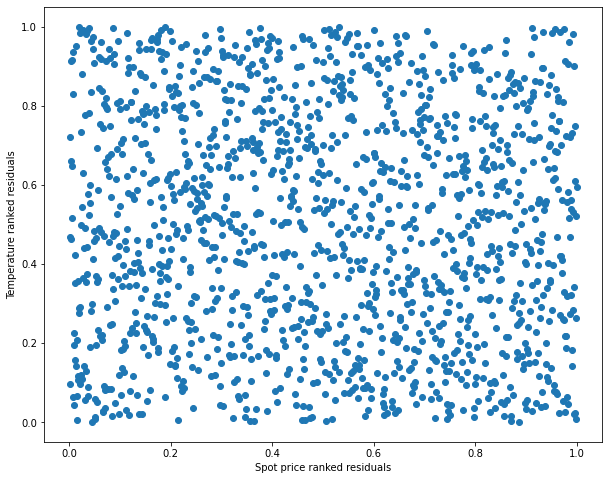}  
    \vspace{4ex}
  \end{minipage} 
  \caption{Standardized residuals (left) and ranked residuals (right) plots for France (top) and North Italy (bottom).}
  \label{fig:residuals_cloud}
\end{figure}

\begin{table}[h]
\begin{minipage}{.5\linewidth}
      \centering
\begin{tabular}{c|c|c}
197 & 162 & 126 \\ \hline 
163 & 160 & 162  \\ \hline 
124 & 163 & 197 
\end{tabular}
\end{minipage}%
    \begin{minipage}{.5\linewidth}  
    \centering
    \begin{tabular}{c|c|c}
165 & 173 & 149 \\ \hline
176 & 147 & 163 \\ \hline
145 & 166 & 175 
\end{tabular}
\end{minipage}
\caption{Observed frequencies by couple tercile for French (left) and Italian (right) coupled data.}
\label{tab:freq_tercile}
\end{table}

In order to move forward, we perform chi-square independence tests on the ranked residuals. For this we classify the ranked residuals based on quantiles, compute contingency tables with frequencies per coupled quantile classification and perform a chi-square independence test on these frequencies compared to a binomial distribution.

\begin{figure}[ht!]  
  \begin{minipage}[b]{0.5\linewidth}
    \centering
    \includegraphics[width=0.8\linewidth]{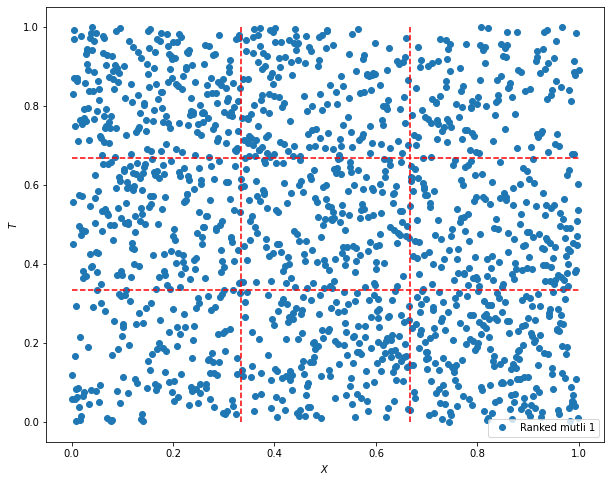} 
    \caption*{$p-value = 1.764e-06$} 
    \vspace{4ex}
  \end{minipage}
  \begin{minipage}[b]{0.5\linewidth}
    \centering
    \includegraphics[width=0.8\linewidth]{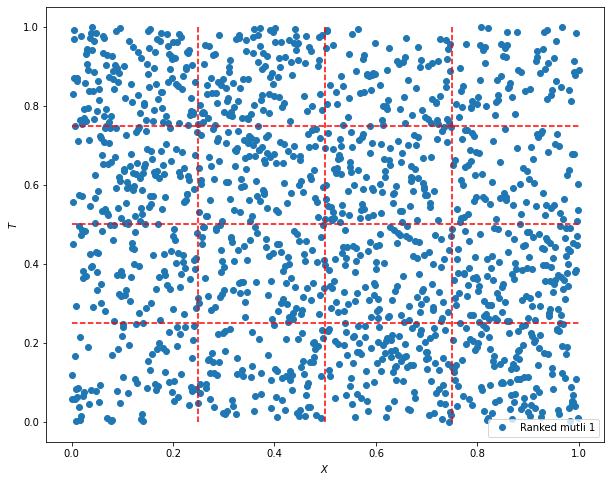} 
    \caption*{$p-value = 6.532e-07$} 
    \vspace{4ex}
  \end{minipage} 
  \begin{minipage}[b]{0.5\linewidth}
    \centering
    \includegraphics[width=0.8\linewidth]{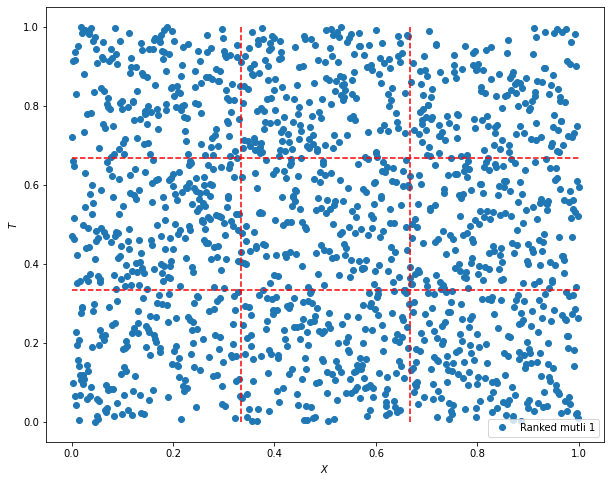} 
    \caption*{$p-value = 0.118$} 
    \vspace{4ex}
  \end{minipage}
  \begin{minipage}[b]{0.5\linewidth}
    \centering
    \includegraphics[width=0.8\linewidth]{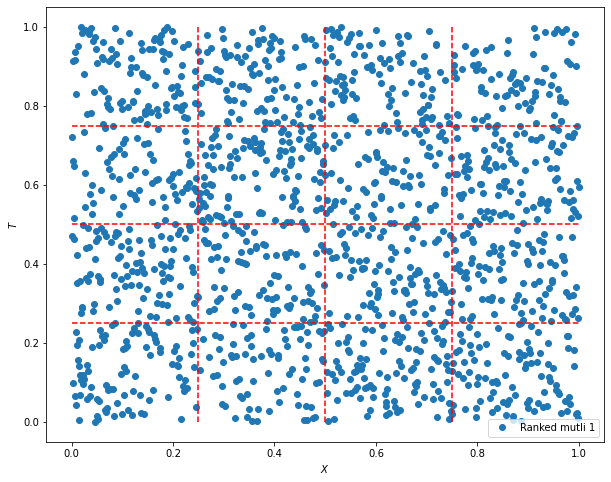} 
    \caption*{$p-value = 0.008$} 
    \vspace{4ex}
  \end{minipage} 
  \caption{$\chi^2$-test performed on ranked residuals for 9 (left) and 16 (right) category classification for France (top) and Italy (bottom).}
  \label{fig:chi_indep_test}
\end{figure}

Figure~\ref{fig:chi_indep_test} represents the results of the $\chi^2$ independence test performed on residuals. We can see that the two datasets do not show same results: the French dataset clearly rejects the independence hypotheses while the North Italian dataset only rejects the independence hypothesis for 16 categories. This motivates us to propose a combined model for $(\tilde{X}_t, \tilde{T}_t)$ that allows dependence on the residuals.

\subsection{Estimation of $\lambda$ and NIG parameters}

Let us recall that the parameters of the Temperature diffusion $\kappa_T$, $\mu_T(\cdot)$ and $\sigma_T$ can be estimated as in Section~\ref{sec:calibration}, as well $\kappa_X$ and $\mu_X(\cdot)$ by Proposition~\ref{prop_estimX}. We assume these parameters estimated, and focus first on the estimation of $\lambda$.

\paragraph{Estimation of $\lambda$}
From~\eqref{eq:comb_int}, we compute the covariance of the residuals:
\begin{align*}
    Cov\left(X_{t+\Delta}-\mu_X(t+\Delta)-e^{-\kappa_X \Delta}(X_t-\mu_X(t)),T_{t+\Delta}-\mu_T(t+\Delta)-e^{-\kappa_T \Delta}(T_t-\mu_T(t)) \right)\\
    =Cov\left(\lambda \sigma_T \int_t^{t+\Delta} e^{-\kappa_X(t+\Delta-v)} dW^T_v  + \int_t^{t+\Delta} e^{-\kappa_X(t+\Delta-v)} dL^X_v, \sigma_T \int_t^{t+\Delta} e^{-\kappa_T(t+\Delta-s)} dW^T_v \right) \\
    = \int_t^{t+\Delta} \sigma_T^2 \lambda e^{-(\kappa_X+\kappa_T)(t+\Delta-v)} dv
    = \sigma_T^2 \lambda \frac{1-e^{-(\kappa_{X}+ \kappa_{T})\Delta}}{\kappa_{X} + \kappa_{T}}.
\end{align*}
 Hence, discretizing for a time period $\Delta$, we get the following estimator
 \begin{equation}\label{eq:estim_lambda}
     \hat{\lambda} =  \frac{\hat{\kappa}_{X} + \hat{\kappa}_{T}}{\hat{\sigma}_T^2 (1-e^{-(\hat{\kappa}_{X}+ \hat{\kappa}_{T})\Delta})} \widehat{Cov},
 \end{equation}
 where $\widehat{Cov}$ is the usual covariance estimator between residuals.  
 Using values in Table~\ref{tab:parameter_estimates_temp_gauss}, we get the estimated value of $\lambda$ in Table~\ref{tab:lambda_estimates}.

\begin{table}[h]
\centering
    \begin{tabular}{c|c|cc}
    Market & $\hat{\lambda}$  & $\hat{\lambda}_{JFM}$  & $\hat{\lambda}_{JJA}$  \\ \hline \hline
    France & $ -0.007$ & $-0.0175$ & $0.0119$\\
    North Italy & $ -0.002$ & $-0.0136$ & $0.008$ \\
\end{tabular}
\caption{Estimated $\lambda$ of Model~\eqref{eq:comb} for France and North Italy.}
\label{tab:lambda_estimates}
\end{table}
    \red{ In Table~\ref{tab:lambda_estimates}, $\hat{\lambda}$, $\hat{\lambda}_{JFM}$ and $\hat{\lambda}_{JJA}$ correspond to the estimated values of $\lambda$ through the whole year, from January to March and from June to August respectively. 
    We can see that the estimation of $\lambda$ during winter months is larger in absolute value, meaning that the opposite dependence between electricity prices and temperature is higher during winter months than through the year. This is coherent as during winter the energy price is more related to heating activities and therefore a decrease in average temperature leads to an increase on electricity prices. On the contrary, during the summer, $\lambda$ becomes positive meaning that an increase on temperature leads to an increase on price. This is also coherent as in France and North Italy energy is used for cooling during summer. Overall the sign of $\lambda$ computed through out the year shows that the dependence on heating is more important than the dependence on cooling. This short analysis indicates that $\lambda$ may evolve periodically along the year. However, taking into account this feature in our model would lead to much more cumbersome developments in Section~\ref{Sec_risk_quanto} to explicit get formulas on quanto derivatives. This is left for further research. Besides, the formulas of Section~\ref{Sec_risk_quanto} already enables us to understand and handle the sensitivity of the derivatives to the parameter~$\lambda$.
    For the rest of the study, we use $\hat{\lambda}$ estimated on the whole year in our numerical experiments.}

 \paragraph{Estimation of NIG parameters}
 We will use, as in Section~\ref{sec:calibration} the CLS estimation method. For this, we write the (conditional) characteristic function of the log-prices:
\begin{equation} \label{eq:charac_nig_update}
     \psi_{X_t}(u; \Delta) := \E(e^{i u X_{t+\Delta}}|X_t) = { e^{iu (\mu_X(t+\Delta)+e^{-\kappa_X \Delta} (X_t - \mu_X(t))-{\frac {1}{2}}\lambda^{2}\sigma_T^{2}  \frac{1- e^{-2\kappa_X \Delta }}{2\kappa_X}  u^{2}}} \varphi(u ;\Delta).
\end{equation}
This gives immediately the characteristic function of the residuals $$\E[e^{i u ( \tilde{X}_{t+\Delta} -e^{-\kappa_X \Delta} \tilde{X}_t)}| \cF_t]= e^{-{\frac {1}{2}} \lambda^{2}\sigma_T^{2} \frac{1- e^{-2\kappa_X\Delta }}{2\kappa_X} u^{2}} \varphi(u ;\Delta).$$
Then, to apply Conditional Least Square Estimation, we update the objective function in Equation~\eqref{eq:charac} accordingly and seek to minimise the following quantity:
\begin{equation}
    \sum_u  \sum_{t=0}^{N-1}  \Big| e^{i u ( \tilde{X}_{t+\Delta} -e^{-\kappa_X \Delta} \tilde{X}_t)}- e^{-{\frac {1}{2}} \lambda^{2}\sigma_T^{2} \frac{1- e^{-2\kappa_X\Delta }}{2\kappa_X} u^{2}} \varphi(u ;\Delta) \Big|^2.
    \label{eq:charac2}
\end{equation}
Table~\ref{tab:parameter_estimates_updated} summarizes the NIG parameter estimations through CLS. We can observe that Table~\ref{tab:parameter_estimates_updated} is very close to Table~\ref{tab:parameter_estimates} which is expected as the parameter $\lambda$ is quite small.

\begin{table}[ht!]
\begin{minipage}{.5\linewidth}
      \centering
\begin{tabular}{ccccc}
Method & $\hat{\alpha}^X$ & $\hat{\beta}^X$   & $\hat{m}^X$ & $\hat{\delta}^X$ \\ \hline \hline
CLS & 4.189 & -0.379 & 0.011 & 0.125 \\
\end{tabular}
\end{minipage}%
    \begin{minipage}{.5\linewidth}  
    \centering
    \begin{tabular}{ccccc}
Method & $\hat{\alpha}^X$ & $\hat{\beta}^X$   & $\hat{m}^X$ & $\hat{\delta}^X$ \\ \hline \hline
CLS & 13.621 & 0.003 & $10^{-3}$ & 0.151 \\
\end{tabular}
\end{minipage}
\caption{Parameter estimation through the CLS  for French (left) and Italian (right) energy log spot price.}
\label{tab:parameter_estimates_updated}
\end{table}

We then turn to the goodness of fit of the estimated combined Model~\eqref{eq:comb}. Figure~\ref{fig:chi_test_m2} and \ref{fig:chi_test_m2_it} represent $\chi^2$ test performed between the historical and simulated distributions on the (2-dimensional) empirical copula between temperature and electricity spot price residuals. To ensure reliability of the results, we perform $1,000,000$ simulations and rescale the frequencies to compare with observed frequencies.  We can see that the test does not globally reject the null hypotheses which means that the dependence is correctly reproduced by Model~\eqref{eq:comb} for both French and North Italian data.

Finally we analyse the standard deviation explained by temperature component in Model~\eqref{eq:comb}. We use the ratio below:
\begin{equation*}
    \frac{|\sigma_T \lambda|}{\sqrt{\sigma_T^2 \lambda^2 + \frac{\delta \alpha^2}{\gamma^3} } }
\end{equation*}
we found that $9.43\%$ and $4.34\%$ of the standard deviation of the random term of the log energy spot price is explained by the temperature component for French and North Italian data correspondingly. This is small but not negligible, especially for handling the risk of derivatives as shown in the next section. 

\begin{figure}[htb!]  
  \begin{minipage}[b]{0.5\linewidth}
    \centering
    \includegraphics[width=0.8\linewidth]{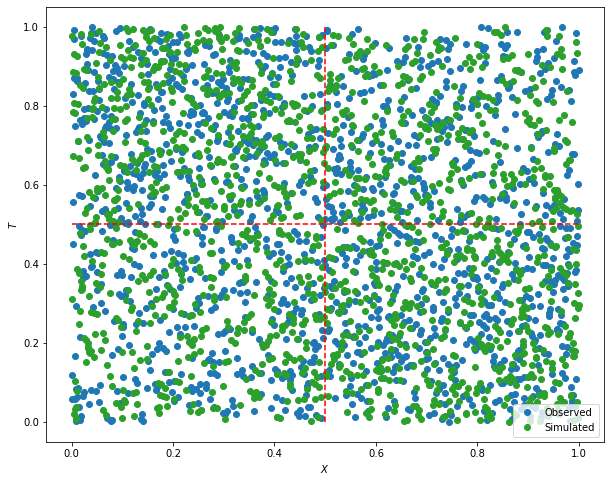} 
    \caption*{$p-value = 0.219$} 
    \vspace{4ex}
  \end{minipage}
  \begin{minipage}[b]{0.5\linewidth}
    \centering
    \includegraphics[width=0.8\linewidth]{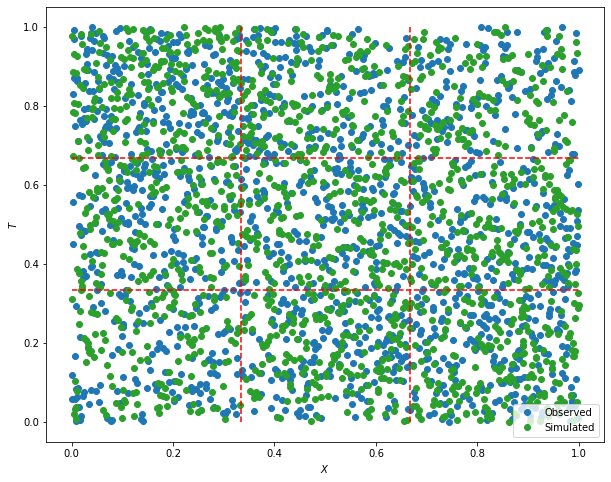} 
    \caption*{$p-value =  0.892$} 
    \vspace{4ex}
  \end{minipage} 
  \begin{minipage}[b]{0.5\linewidth}
    \centering
    \includegraphics[width=0.8\linewidth]{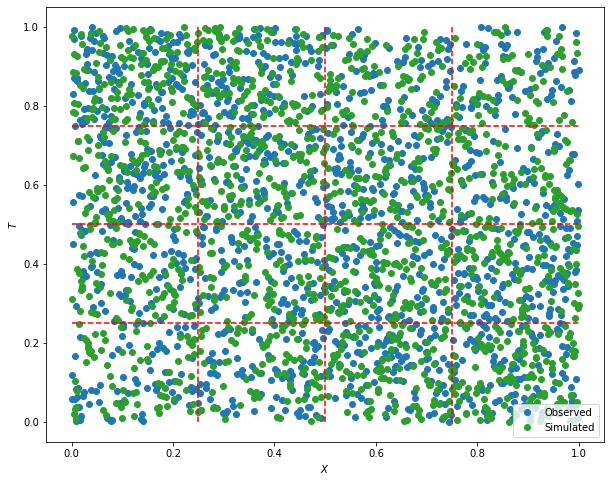} 
    \caption*{$p-value =0.616$} 
    \vspace{4ex}
  \end{minipage}
  \begin{minipage}[b]{0.5\linewidth}
    \centering
    \includegraphics[width=0.8\linewidth]{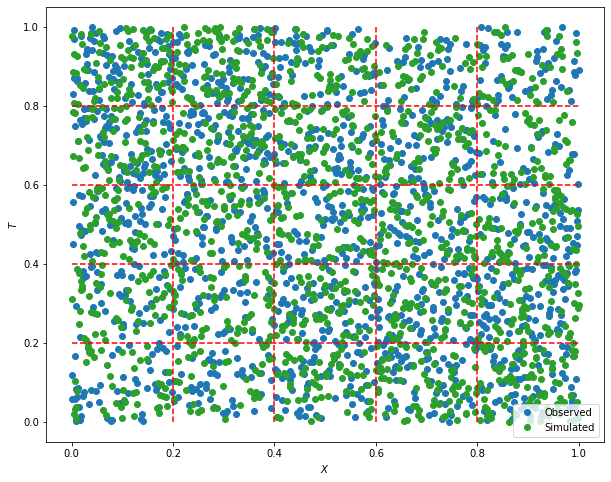} 
    \caption*{$p-value = 0.124$} 
    \vspace{4ex}
  \end{minipage} 
  \caption{From top left to bottom right, $\chi^2$ test performed on the distributions of real (blue) and simulated (green - based on $1,000,000$ simulations) ranked residuals for 4 and 25 categories for French data. }
  \label{fig:chi_test_m2}
\end{figure}

\begin{figure}[htb!]  
  \begin{minipage}[b]{0.5\linewidth}
    \centering
    \includegraphics[width=0.8\linewidth]{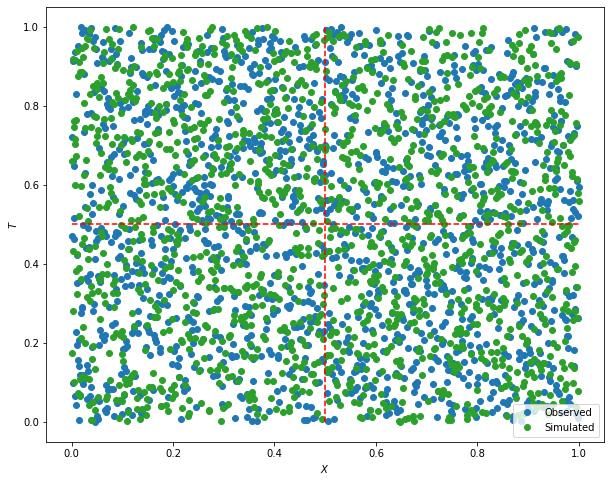} 
    \caption*{$p-value = 0.961$} 
    \vspace{4ex}
  \end{minipage}
  \begin{minipage}[b]{0.5\linewidth}
    \centering
    \includegraphics[width=0.8\linewidth]{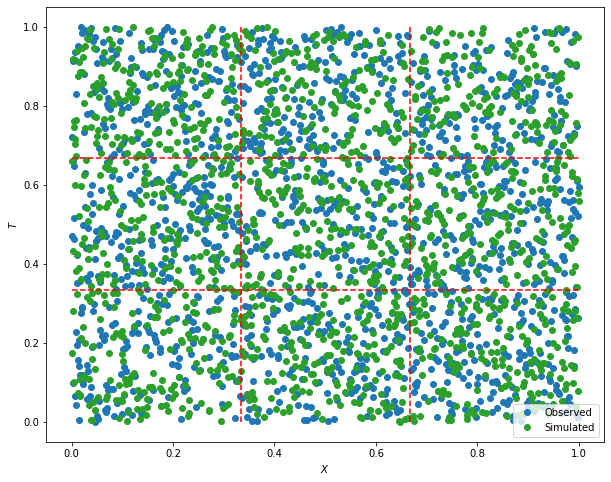} 
    \caption*{$p-value = 0.834$} 
    \vspace{4ex}
  \end{minipage} 
  \begin{minipage}[b]{0.5\linewidth}
    \centering
    \includegraphics[width=0.8\linewidth]{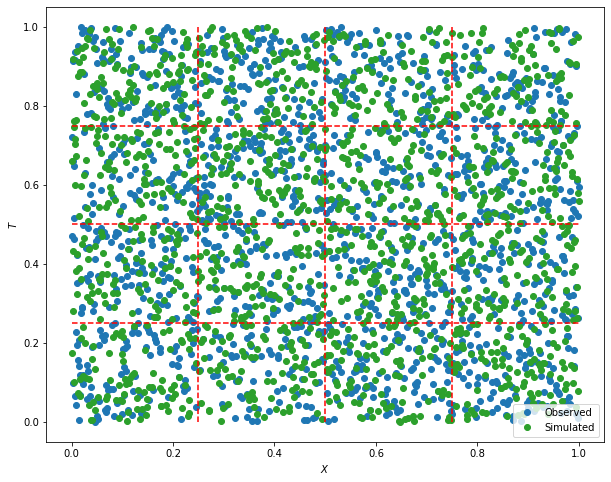} 
    \caption*{$p-value =0.198$} 
    \vspace{4ex}
  \end{minipage}
  \begin{minipage}[b]{0.5\linewidth}
    \centering
    \includegraphics[width=0.8\linewidth]{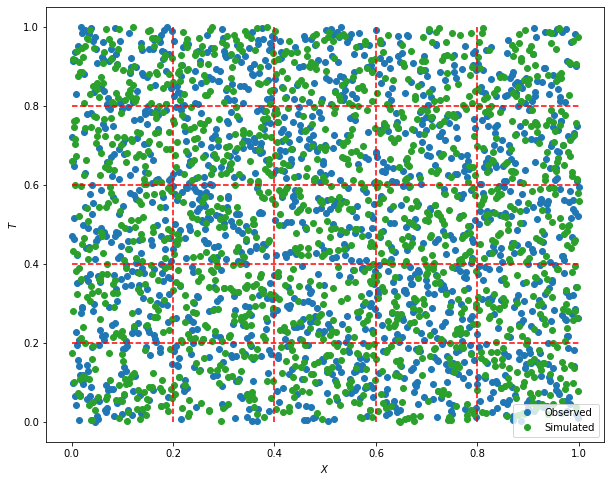} 
    \caption*{$p-value = 0.7$} 
    \vspace{4ex}
  \end{minipage} 
  \caption{From top left to bottom right, chi-square test performed on the distributions of real (blue) and simulated (green - based on $1,000,000$ simulations) ranked residuals for 4 and 25 categories for North Italian data. }
  \label{fig:chi_test_m2_it}
\end{figure}

\section{Handling the risk of quanto derivatives}\label{Sec_risk_quanto}

On the previous sections we have developed a combined model for energy spot price and temperature in order to price financial derivatives combining both parameters. The objective of this section is to apply this model for quanto valuation and hedging. 

\subsection{An overview on quanto design and risk valuation}

Quantos are derivative contracts of which payoff depends a double trigger, meaning the claim depends on the value at maturity of two indices. Their main interest relies on their capacity to hedge simultaneously volumetric risk, linked to weather conditions, and price risk, represented by the energy price. Quanto derivatives are defined, like weather and energy derivatives, over a time period $[t_1, t_2]$ such that the payoff of the contract will depend on the value of the underlying during all this risk period. While there exist large studies on quanto risk valuation involving different commodities~\cite{zhang1997exotic}, there exist little literature on the structuration and pricing of quanto products mixing commodity and weather inputs \cite{caporin2012model} \cite{benth2015pricing} \cite{kufakunesu2022sensitivity}. \\

On the one hand, there is no clear consensus on the structuring of the products. 
For Benth \red{et al.} \cite{benth2015pricing}  and \red{Kafakunesu} \cite{kufakunesu2022sensitivity}, the payoff structure is applied to an aggregate of the two underlyings:
\begin{equation} \label{eq:payoff_Benth}
    Payoff := f \left(\sum_{t=t_1}^{t_2} g_S (S_t), \sum_{t=t_1}^{t_2} g_T (T_t) \right)
\end{equation}
where $f$ represents the payoff function and $g_S$ and $g_T$ represent transformations of the initial inputs. The dates $t_1<t_2$ indicates two days, and the summation is made on each day between $t_1$ and $t_2$ (including these days). 
In particular, for Caporin~\cite{caporin2012model}, $f$ takes the form of a product and integrates common derivative payoff functions such that:
\begin{equation} 
    Payoff := f_S \left(\sum_{t=t_1}^{t_2} g_S (S_t)\right) \times f_T\left( \sum_{t=t_1}^{t_2} g_T (T_t) \right)
\end{equation}
Here, $f_S$ and $f_T$ can correspond to the payoff function of put and call options, capped linear for swaps and identity for futures. Commonly $f_T$ can correspond to the formula enabling to compute the Heating Degree Days (HDD) such that $f_T(\cdot) = (\bar{T} - T_t)^+ $, with $\bar{T}=18^\circ C$. 
Finally, some practitioners seem to favour another definition~\cite{cucu2016managing}. In these cases the payoff function is directly applied to the daily values such that:
\begin{equation} \label{eq:eqcucu}
    Payoff := \sum_{t=t_1}^{t_2} f_S( S_t) \times f_T(T_t )
\end{equation}
We will focus on this latter definition as it answers to practitioner's needs. In addition, we consider price settlement takes place at $t_0$ and payoff payment at $t_2$.  

We now discuss briefly the valuation of these products. Under classic risk-neutral pricing theory for financial derivatives we would like to write the price as follows:
 \begin{equation} \label{eq:Epayoff}
 \E_{\mathbb{Q}} \left(D(t_0,t_2) \sum_{t=t_1}^{t_2} f_S( S_t) \times f_T(T_t )\right)
 \end{equation}
where $\mathbb{Q}$ corresponds to the risk-neutral probability and $D(t_0,t_2)$ is a discount factor between $t_0$ and $t_2$.  

However, temperature is not an asset traded on markets, and hence risk neutral theory cannot be applied. A possible way to get around this is to work with payoffs  written on futures contracts, as in the works of Benth \red{et al.} \cite{benth2015pricing}  and \red{Kafakunesu} \cite{kufakunesu2022sensitivity}. However, futures on temperature are still not liquid and usual quanto payoffs are written on the temperature itself, not on the future contracts. Here, we will rather work on the real-world probability world and analyse the payoff distribution in this framework. We will be particularly focus on the average payoff, which we discuss in the next subsection. In addition, we consider a unit discount rate (i.e. $D(t_0,t_2)=1$) as the time span $t_2-t_0$ is rather short and there is no clear alignment on which rate should be used, see e.g.~\cite{benth2007putting}~\cite{brody2002dynamical}. Note also that if the discount rate is deterministic or independent of $(S,T)$, the formulas below are still valid up to a constant factor.

In the following section we develop explicit formulas for different payoff and compare then with Monte Carlo simulations. We also consider the possibility to hedge statically~\eqref{eq:eqcucu} with electricity and temperature derivatives, and analyse numerically the hedging error distribution.

\subsection{Expected values of some standard payoffs} \label{subsec:expected_values}
This section concentrates on our ability to get explicit values for the average payoff of different financial instruments. Namely, we will consider the payoff~\eqref{eq:eqcucu} for different choices of functions $f_S$ and $f_T$. Let $t_0$ denote the present date, we consider two dates $t_1$ and $t_2$ such that $t_0<t_1<t_2$ and want to determine:
\begin{equation} 
    \E \left( \sum_{t=t_1}^{t_2} f_S( S_t) \times f_T(T_t ) \mid \mathcal{F}_{t_0}\right).
\end{equation}
The dates $t_1$ and $t_2$ indicate days, and the summation is made on all days between $t_1$ and $t_2$ including them. We also apply these formulas and compare the results to prices computed through Monte Carlo simulation. The Monte Carlo discretization schemes can be found in Appendix~\ref{appendix:Simu_MNIG} and are applied to parameters in Table~\ref{tab:parameters_simu} and $\mu_X$ and $\mu_T$ as estimated for France with Formulas~\eqref{eq:def_hlambda} and \eqref{eq:def_hzeta}.

\begin{table}[ht!]
\centering
\begin{tabular}{lllllllllll}
$\kappa_X$ & $\beta^X_0$ &	$\alpha^X_1$ &$\beta^X_1$	& $\alpha^X_{DoW0}$ & $\alpha^{X,DoW}_{1}$ &	$\alpha^{X,DoW}_2$ & $\alpha^{X,DoW}_{3}$ &	$\alpha^{X,DoW}_{4}$ & $\alpha^{X,DoW}_{5}$ &	$\alpha^{X,DoW}_{6}$  \\ \hline \hline
0.226 &  0.0003 &	-0.165 &	 0.187 & 3.523 & 3.594 & 3.594 & 3.589 &3.566 &	3.370 &	3.175   
\end{tabular}
\begin{tabular}{llllll}
$\kappa_T$ &$\alpha^T_0$ & $\beta^T_0$ &	$\alpha^T_1$ &$\beta^T_1$	 \\ \hline \hline
0.254 &   6.578 &  0.00004   & -4.139& -6.959
\end{tabular}
\begin{tabular}{llllllll}
$\alpha^X$ & $\beta^X$  & $m^X$    & $\delta^X$ &   $\sigma_T$ & $\lambda$  \\ \hline \hline
4.189 &-0.379 &  0.011 & 0.125 &    2.413& -0.007    
\end{tabular}
\caption{Parameters used for numerical experiments. These parameters are the ones obtained with the estimation on French electricity and temperature data.} \label{tab:parameters_simu}
\end{table}

\paragraph{Forwards/Futures}
First, let consider $t_0<t_1<t_2$ and a future derivative between times $t_1$ and $t_2$. The payoff is given by~\eqref{eq:eqcucu} with $f_T(\cdot)=f_S(\cdot)=\cdot$. We consider the average payoff at time~$t_0$ under the historical probability:
\begin{align*}
    \mathcal{F}(t_1,t_2) = \E \Big( \sum_{t=t_1}^{t_2}  S_t \times T_t \mid \mathcal{F}_{t_0}\Big)
\end{align*}

\begin{prop}
Under  Model~\eqref{eq:comb} and for $t \in [t_1, t_2]$, we have:
\begin{equation}\label{eq:forward}
    \begin{aligned}
    \mathcal{F}(t_1,t_2)  &= \sum_{t=t_1}^{t_2} \Bigg[ \exp\Bigg(\mu_X(t) + e^{-\kappa_X (t-t_0) }(X_{t_0}-\mu_X(t_0)) \Bigg) \varphi(-i ; t -t_0)\\
    &  \Bigg( (\mu_T(t) + e^{-\kappa_T (t-t_0) }(T_{t_0}-\mu_T(t_0)))  e^{\frac{1}{2} k_X(t-t_0)^2\lambda^2 \sigma_T^2  } \\
    +& \lambda \sigma_T^2 k_{XT}^2(t-t_0) e^{ \frac{1}{2}\lambda^2 \sigma_T^2 k_{X}(t-t_0)^2}     \Bigg)  \Bigg]
    \end{aligned}
\end{equation}
\noindent where $\varphi$ is the characteristic function defined in Equation~\eqref{eq:charac_nig} and

\begin{equation} \label{eq:k}
    k_T(\Delta)=\sqrt{\frac{1 - e^{-2\kappa_T\Delta}}{2 \kappa_T}} \mbox{, }k_X(\Delta)=\sqrt{\frac{1 - e^{-2\kappa_X\Delta}}{2 \kappa_X}} \mbox{ and } k_{XT}(\Delta)=\sqrt{\frac{1 - e^{-(\kappa_X+\kappa_T)\Delta}}{\kappa_X+\kappa_T}}.
\end{equation}

\end{prop}

This result is a direct consequence of Proposition~\ref{prop:eXt} that calculates explicitly $\E ( S_t \times T_t \mid \mathcal{F}_{t_0})$. 
Figure~\ref{fig:forward_swap} shows the price of monthly futures computed 30 days in advance for France data. The prices have been computed with 100,000 Monte Carlo simulations and with explicit formula in Equation~\eqref{eq:forward}. We can see that both methodologies provide same results. However, computation with Equation~\eqref{eq:forward} is around 650 times faster. Additionally we compute the price with $\lambda=0$. Although it is close to the precedent case, the price computed with  $\lambda=0$ is usually not within the confidence intervals of the Monte Carlo simulations.

\paragraph{Swap} Swaps are derivatives where the payoff is given by~\eqref{eq:eqcucu}, that is, $f_T(\cdot)=(\bar{T}-\cdot)$ and $f_S(\cdot)=(\cdot-\Bar{S})$. Here we suppose strikes $\bar{S}$ and $\bar{T}$ are given. For applications, and on the following of the paper, $\Bar{T}=18^\circ C$ and $\Bar{S}=50$ EUR/MWh. While taking $\Bar{T}=18^\circ C$ is a market standard, taking $\Bar{S}=50$ is our choice. We decide to take a strike which is relatively into the money and repeat the exercise for other strikes ($\Bar{S}=40, 60$), the below conclusions remain unchanged.

We define the swap's average payoff under historical probability $\mathcal{S}(t_1,t_2)$ as follows:

\begin{align*} 
    \mathcal{S}(t_1,t_2) = \E \Big( \sum_{t=t_1}^{t_2}  (S_t - \bar{S}) (\bar{T}-T_t) \mid \mathcal{F}_{t_0}\Big)
\end{align*}

\begin{prop}
Under  Model~\eqref{eq:comb} and for $t \in [t_1, t_2]$, we have
    \begin{equation}\label{eq:swap}
        \begin{aligned}
        \mathcal{S}(t_1,t_2) &= \sum_{t=t_1}^{t_2} \bar{T} \psi_{X_{t_0}}(-i ; t -t_0) -\mathcal{F}(t_1,t_2) -\bar{S} \bar{T} (t_2-t_1+1)\\
        &+ \bar{S} \sum_{t=t_1}^{t_2} (\mu_T(t) + e^{-\kappa_T (t-t_0)} (T_{t_0}- \mu_T(t_0)))
        \end{aligned}
    \end{equation}
    \noindent where $\psi$ is the characteristic function defined in Equation~\eqref{eq:charac_nig_update}, $\mathcal{F}(t_1,t_2)$ as defined in Equation~\eqref{eq:forward} and $k_T(\cdot)$ , $k_X(\cdot)$ and $k_{XT}(\cdot)$ are defined in Equation~\eqref{eq:k}.
\end{prop}

This result is a direct consequence of Proposition~\ref{prop:swap} that calculates $\E ( (S_t-\bar{S})(\bar{T}-  T_t) \mid \mathcal{F}_{t_0})$. Figure~\ref{fig:forward_swap} shows the price of monthly swaps computed 30 days in advance for France data. The prices have been computed with 100,000 Monte Carlo simulations and with explicit formula in Equation~\eqref{eq:swap}. As for forwards, we can see that both methods provide similar results while using Formula~\eqref{eq:swap} is around 300 times faster. Again, prices computed for $\lambda=0$ are close but out of the confidence interval of the Monte Carlo simulations.

\begin{figure}[ht!]
    \begin{minipage}[b]{0.45\linewidth}
        \centering
        \includegraphics[width=\textwidth]{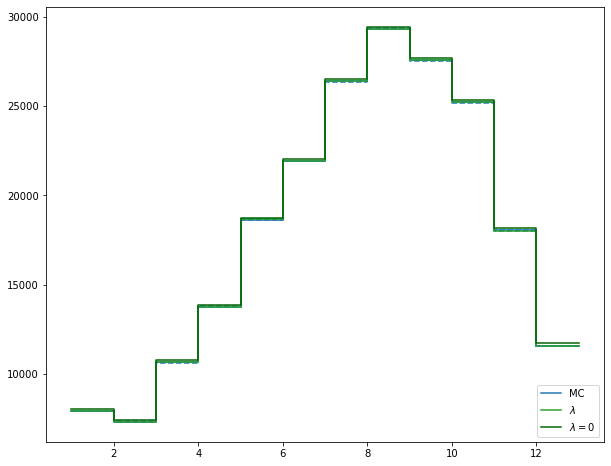}
    \end{minipage}
    \hspace{0.3cm}
    \begin{minipage}[b]{0.45\linewidth}
        \centering
        \includegraphics[width=\textwidth]{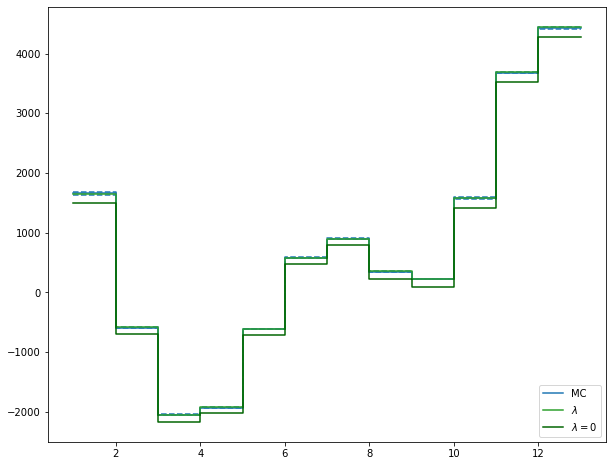}
    \end{minipage}
    \caption{Forward (left) and Swap (right) prices computed with 100,000 simulation-Monte Carlo (blue) and Equations~\eqref{eq:forward} and~\eqref{eq:swap} (green) methods. Each contract lasts one month of 2018. Time $t_0$ corresponds to $30$ days ahead of the first day of the month, $t_1$ to the first day of the month and $t_2$ to the last day of the month.} \label{fig:forward_swap}
\end{figure}

\paragraph{Single sided options $\mathcal{E}$-HDD and $\mathcal{E}$-CDD}
We now focus on put and call options on temperature. These are defined by the payoff~\eqref{eq:eqcucu} with $f_S(x)=x$ and either $f_T(\cdot) := (\bar{T} - \cdot )^+$ for a put option or $f(\cdot) := (\cdot - \bar{T})^+$ for a call option. We call these products single sided options, since the option brings only on the temperature. We define the average payoff under historical probability of a single sided option $\mathcal{E-}HDD(t_1,t_2)$ as follows:
\begin{align}
    \mathcal{E-}HDD(t_1,t_2) = \E \Big( \sum_{t=t_1}^{t_2}  S_t (\bar{T}-T_t)^+ \mid \mathcal{F}_{t_0}\Big) = \sum_{t=t_1}^{t_2}  \E ( S_t (\bar{T}-T_t)^+ \mid \mathcal{F}_{t_0})
\end{align}
Let first consider each term separately. We have:
\begin{align}\label{decompo_put}
    \E ( S_t (\bar{T}-T_t)^+ \mid \mathcal{F}_{t_1}) &= \int_{T^0}^{\bar{T}} \E ( e^{X_t} \mathbbm{1}_{T_t \leq u }  \mid \mathcal{F}_{t_1}) du,
\end{align}
where $T^0=-\infty$. Note that the Gaussian model for the temperature allows in principle any real temperature. In practice, with the estimated parameters, the probability of having extreme temperatures is very infinitesimal. We can thus use~\eqref{decompo_put} with $T^0=-273.15$ (the absolute zero temperature) or $T^0=-100$ with a negligible error. 

\begin{prop}
Under  Model~\eqref{eq:comb} and for $t \in [t_1, t_2]$, we have
    \begin{equation} \label{eq:eHDD}
    \begin{aligned}
    \mathcal{E-}HDD(t_1,t_2) &= \sum_{t=t_1}^{t_2}  \Bigg[  \psi_{X_{t_0}}(-i ; t -t_0) \times \\
    & \int_{T^0}^{\bar{T}}  \Phi\Big( \frac{u - (\mu_T(t) + e^{-\kappa_T (t-t_0)}(T_{t_0}-\mu_T(t_0))) }{\sigma_Tk(t-t_0)} - \lambda \sigma_T \frac{k_{XT}^2(t-t_0)}{k_{T}(t-t_0)} \Big) du \Bigg]
    \end{aligned}
    \end{equation}
where $\psi$ is the characteristic function defined in Equation~\eqref{eq:charac_nig_update}, $\Phi$ is the cumulative function of the standard Gaussian distribution and  $k_T(\cdot)$ , $k_X(\cdot)$ and $k_{XT}(\cdot)$ are as in Equation~\eqref{eq:k}.
\end{prop}

\begin{proof}
    For $t \in [t_1, t_2]$, we use~\eqref{decompo_put} and apply then Proposition~\ref{prop:eX1}.
\end{proof}

Similarly, we can calculate explicitly the average payoff of a $\mathcal{E}$-CDD defined as:
\begin{align} 
    \mathcal{E-}CDD(t_1,t_2) = \E \Big( \sum_{t=t_1}^{t_2}  S_t (T_t-\bar{T})^+ \mid \mathcal{F}_{t_0}\Big). 
\end{align}

\begin{prop}
Under  Model~\eqref{eq:comb} and for $t \in [t_1, t_2]$, we have:
\begin{equation} \label{eq:eCDD}
    \begin{aligned}
    \mathcal{E-}CDD(t_1,t_2) &= \sum_{t=t_1}^{t_2}  \Bigg[  \psi_{X_{t_0}}(-i ;t -t_0) \times \\
    &   \int_{\bar{T}}^{T^{m}}  \Phi\Big( \lambda \sigma_T \frac{k_{XT}^2(t-t_0)}{k_{T}(t-t_0)} - \frac{u - (\mu_T(t) + e^{-\kappa_T (t-t_0)}(T_{t_0}-\mu_T(t_0))) }{\sigma_T k_T(t-t_0) }   \Big) du \Bigg]
\end{aligned}
\end{equation}
where $T^m=+\infty$, $\psi$ is the characteristic function defined in Equation~\eqref{eq:charac_nig_update}, $\Phi$ is the cumulative function of the standard Gaussian distribution, $k_T(\cdot)$, $k_X(\cdot)$ and $k_{XT}(\cdot)$ are as in Equation~\eqref{eq:k}.
\end{prop}
\noindent In practice, \eqref{eq:eCDD} can be used with $T_m=100$  with a negligible error, the probability of having temperature above 100°C being infinitesimal. 

Figures in~\ref{fig:eHDD_eCDD} show the price of monthly $\mathcal{E}$-HDD and $\mathcal{E}$-CDD computed 30 days in advance for French data. The prices have been computed with 100,000 Monte Carlo simulations and with explicit formula in Equations~\eqref{eq:eHDD} and~\eqref{eq:eCDD}. We can see that both methodologies provide same results. However, computation with formulas is, again, around 40 times faster for $\mathcal{E}$-HDD and 3 times faster for $\mathcal{E}$-CDD. In addition, we compute the price with $\lambda=0$, the impact on prices is visual for $\mathcal{E}$-CDD as the prices are out of the confidence interval of the Monte Carlo simulations. This shows the significance of $\lambda$ on the valuation of derivatives.

\begin{figure}[ht!]
    \begin{minipage}[b]{0.45\linewidth}
        \centering
        \includegraphics[width=\textwidth]{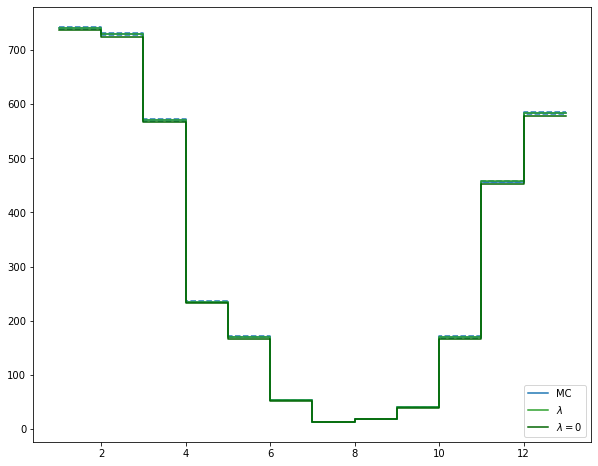}
    \end{minipage}
    \hspace{0.3cm}
    \begin{minipage}[b]{0.45\linewidth}
        \centering
        \includegraphics[width=\textwidth]{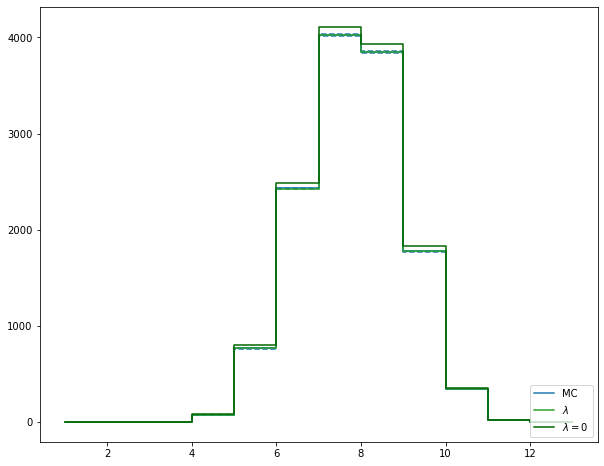}
    \end{minipage}
    \caption{$\mathcal{E}$-HDD (left) and $\mathcal{E}$-CDD (right) prices computed with 100,000 simulation-Monte Carlo (blue) and Equations~\eqref{eq:eHDD} and~\eqref{eq:eCDD} (green) methods. Each contract lasts a month of 2018. $t_0$ corresponds to $30$ days ahead of the first day of the month, $t_1$ to the first day of the month and $t_2$ to the last day of the month. The computation of the derivatives through the formulas around 40 times faster for $\mathcal{E}$-HDD and 3 times faster for $\mathcal{E}$-CDD. } \label{fig:eHDD_eCDD}
\end{figure}

\paragraph{Quanto options} Let now consider double sided options given by Equation~\eqref{eq:eqcucu} with the following payoff function specifications $f_S(\cdot)=(\cdot-\bar{S})^+$ and $f_T(\cdot)=(\bar{T}-\cdot)^+$. We consider then the average payoff function:
\begin{align}
    \mathcal{Q}(t_1,t_2) = \E \Big( \sum_{t=t_1}^{t_2}   (S_t - \bar{S})^+  (\bar{T}-T_t)^+ \mid \mathcal{F}_{t_0}\Big)
\end{align}
Given our Model~\eqref{eq:comb}, there is no explicit formula for $\mathcal{Q}$ up to our knowledge. However, we suggest to perform a Taylor's series expansion to the first order on $\lambda$ given that $\lambda$ is quite small. The next proposition gives semi-explicit formulas for this expansion. 

\begin{prop}
Under  Model~\eqref{eq:comb} and for $t \in [t_1, t_2]$, we have the following Taylor's series expansion:
    \begin{equation} \label{eq:quanto_DL}
        \begin{aligned}
              \mathcal{Q}(t_1,t_2) &=  \sum_{t=t_1}^{t_2} \Bigg( \E_{\lambda=0}( (S_{t}- \bar{S})^+\mid \mathcal{F}_{t_0}) \times
              \Bigg(\Big(\Bar{T} - \mu_T(t) - e^{-\kappa_T(t-t_0)} (T_{t_0} - \mu_T(t_0)) \Big) \times \\
             &\Phi\Big(\frac{\Bar{T} - \mu_T(t) - e^{-\kappa_T(t-t_0)} (T_{t_0} - \mu_T(t_0) )}{\sigma_T k_T(t-t_0)}\Big) \\
            &+\frac{\sigma_T k_T(t-t_0)}{\sqrt{2 \pi}} \exp\Big(-\frac{1}{2} \Big( \frac{\Bar{T} - \mu_T(t) - e^{-\kappa_T(t-t_0)} (T_{t_0} - \mu_T(t_0) )}{\sigma_T k_T(t-t_0)}\Big)^2\Big) \Bigg) \\
            &- \Big( \E_{\lambda=0} ( (S_{t}- \bar{S})^+\mid \mathcal{F}_{t_0}) + \bar{S} \mathbb{P}_{\lambda=0} \left(   S_{t} \geq \bar{S}   \mid \mathcal{F}_{t_0} \right ) \Big) \times \\
            &\sigma_T^2 k_{XT}(t-t_0)^2 \Phi\Big(\frac{\bar{T}-\mu_T(t) - e^{-\kappa_T (t-t_0)}(T_{t_0} -\mu_T(t_0))}{-\sigma_T k_T(t-t_0)}\Big) \Big)\lambda \Bigg) + o(\lambda)
        \end{aligned}
\end{equation}
where $\Phi$ is the cumulative distribution function of the standard Gaussian distribution,  $k_T(\cdot)$ , $k_X(\cdot)$ and $k_{XT}(\cdot)$ are as in Equation~\eqref{eq:k}.
\end{prop}

\noindent This result is a direct application of Proposition~\ref{prop:Taylor_expansion}. Note that $ \E_{\lambda=0}( (S_{t}- \bar{S})^+\mid \mathcal{F}_{t_0})$ (resp. $\mathbb{P}_{\lambda=0} \left( S_{t} \geq \bar{S}   \mid \mathcal{F}_{t_0} \right ) $) can be computed efficiently as in Equation~\eqref{eq:carr_madan} (resp. Equation~\eqref{eq:gil-pelaez}).

Figure~\ref{fig:quanto_month} shows quanto prices computed with Equation~\eqref{eq:quanto_DL} and 100,000 simulation-Monte Carlo simulations. We can see that the first order Taylor development is sufficient as prices are always within Monte Carlo confidence intervals. We compare this with a quanto approached keeping only the first term of Taylor development which is equivalent to $\lambda=0$. This time, average payoff values are close but out of the Monte Carlo confidence intervals. Explicit formula computation remains faster than Monte Carlo simulations however it gets less attractive than for previous derivatives as it includes several numerical integrations.

\begin{figure}[ht!]
    \begin{minipage}[b]{0.45\linewidth}
        \centering
        \includegraphics[width=\textwidth]{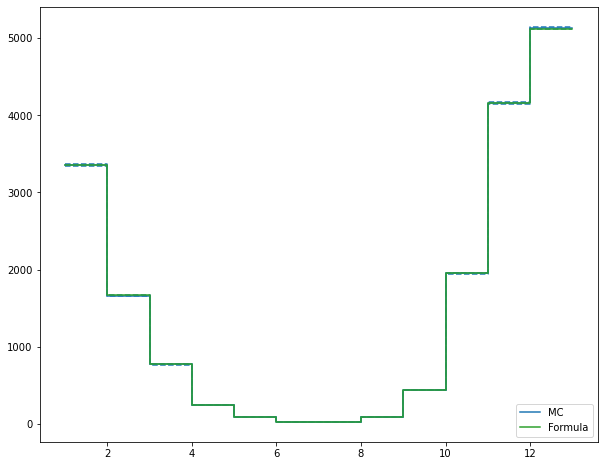}
    \end{minipage}
    \hspace{0.3cm}
    \begin{minipage}[b]{0.45\linewidth}
        \centering
        \includegraphics[width=\textwidth]{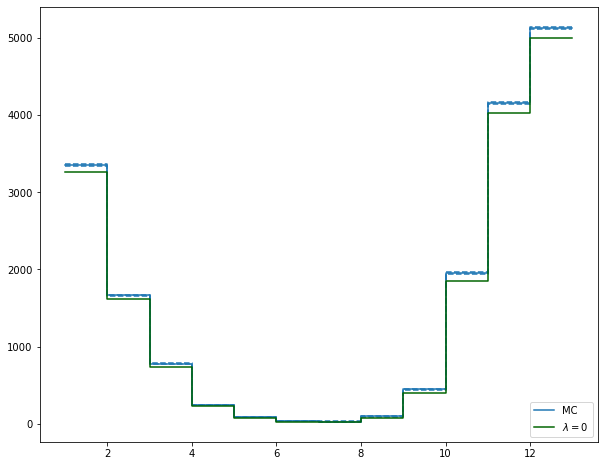}
    \end{minipage}
    \caption{Quanto prices computed with 100,000 simulation-Monte Carlo (blue) and Equation~\eqref{eq:quanto_DL} (green) methods. On the left the price corresponds to Formula~\eqref{eq:quanto_DL}. On the right only the first term of the Taylor development in Equation~\eqref{eq:quanto_DL} is considered. This is equivalent to consider $\lambda=0$. Each contract lasts a month of 2018. Time $t_0$ corresponds to $30$ days ahead of the first day of the month, $t_1$ to the first day of the month and $t_2$ to the last day of the month. The computation of the derivatives through the formulas is around 6 times faster than using Monte Carlo simulations. } \label{fig:quanto_month}
\end{figure}

To sum up, we have developed in this subsection explicit or semi-explicit formulas for futures, swaps, single-sided and double sided options given Model~\eqref{eq:comb}. These formulas are verified through Monte Carlo simulations and remain faster to use than Monte Carlo techniques. Finally, approaching the formulas with $\lambda=0$ provides results close to the simulated ones but out of their confidence intervals showing the significance of the dependence between electricity and temperature on the risk associated to these derivatives.

\subsection{Static hedging of hybrid derivatives}

In the following section, we leverage explicit formulas developed in the above subsection and discuss potential statistic hedging strategies for $\mathcal{E}$-HDD and double-sided quantos. 

We remind that hedging challenges are key for portfolio and risk managers to handle risk aggregation questions and meet solvency constraints. The aim of this subsection is to show that while quantos are somehow exotic derivatives, they can be hedged through more common derivatives increasing their attractiveness and risk understanding.

\subsubsection{Risk decomposition of $\mathcal{E}$-HDD}
Let first focus on $\mathcal{E}$-HDD and take a self-financing portfolio approach. We suppose that we are at time $t$ and that we want to find the static portfolio made with single $HDD$, future on electricity and cash that minimizes the square hedging error at $t+\Delta$ where $\Delta>0$:
\begin{equation} \label{eq:min_eHDD}
    \E[(S_{t+\Delta}(\bar{T}-T_{t+\Delta})^+-c^0_{t,t+\Delta}-c^1_{t,t+\Delta}(\bar{T}-T_{t+\Delta})^+-c^2_{t,t+\Delta} S_{t+\Delta})^2|\mathcal{F}_t].
\end{equation}

\begin{remark}
The risk related to the future on the electricity spot can in principle be hedged dynamically on the electricity market.  Instead, the risk related to the elementary HDD $(\bar{T}-T_t)^+$ cannot  be hedged.    
\end{remark}
This is a quadratic function with respect to $(c^0_{t,t+\Delta},c^1_{t,t+\Delta},c^2_{t,t+\Delta})$ and the first order condition leads to:
\begin{equation} \label{eq:quad_eHDD}
\begin{aligned}
\left[\begin{matrix}
 1 & \E[(\bar{T}-T_{t+\Delta})^+|\mathcal{F}_t] & \E[S_{t+\Delta}|\mathcal{F}_t]   \\    \E[(\bar{T}-T_{t+\Delta})^+|\mathcal{F}_t]  & \E[((\bar{T}-T_{t+\Delta})^+)^2|\mathcal{F}_t] & \E[S_{t+\Delta}(\bar{T}-T_{t+\Delta})^+|\mathcal{F}_t]\\ 
            \E[S_{t+\Delta}|\mathcal{F}_t]  & \E[S_{t+\Delta}(\bar{T}-T_{t+\Delta})^+|\mathcal{F}_t] & \E[S_{t+\Delta}^2|\mathcal{F}_t]
    \end{matrix}\right]\left[\begin{matrix}c^0_{t,t+\Delta} \\c^1_{t,t+\Delta} \\c^2_{t,t+\Delta} \end{matrix}\right] \\ \\
    =\left[\begin{matrix} \E[S_{t+\Delta}(\bar{T}-T_{t+\Delta})^+ |\mathcal{F}_t]  \\\E[S_{t+\Delta}((\bar{T}-T_{t+\Delta})^+)^2 |\mathcal{F}_t] \\\E[S_{t+\Delta}^2(\bar{T}-T_{t+\Delta})^+ |\mathcal{F}_t] \end{matrix}\right]
\end{aligned}
\end{equation}

\begin{prop}\label{prop_hedge_eHDD}
    Under  Model~\eqref{eq:comb} and for $\Delta >0$, the vector  $(c^0_{t,t+\Delta},c^1_{t,t+\Delta},c^2_{t,t+\Delta})$ minimising the function~\eqref{eq:min_eHDD} is the unique solution of the linear equation~\eqref{eq:quad_eHDD}, whose components can be explicitly or semi-explicitly calculated.
\end{prop}
    
\begin{proof}
    The fact that the linear quadratic problem~\eqref{eq:min_eHDD} boils down to~\eqref{eq:quad_eHDD} is standard, one only has to check that the matrix on the left-hand side is invertible so that there is a unique solution. This matrix is a (conditional) covariance matrix: it is invertible, otherwise we could find $\cF_t$-measurable coefficients $(c^0_{t,t+\Delta},c^1_{t,t+\Delta},c^2_{t,t+\Delta})$ such that $c^0_{t,t+\Delta}+c^1_{t,t+\Delta}(\bar{T}-T_{t+\Delta})^+ +c^2_{t,t+\Delta} S_{t+\Delta}=0$, which is clearly impossible from~\eqref{eq:comb_int}.

    We now recall how to calculate explicitly or semi-explicitly: $ \E[(\bar{T}-T_{t+\Delta})^+|\mathcal{F}_t] $ is given by Proposition~\ref{prop:E_T_positive},  $\E[((\bar{T}-T_{t+\Delta})^+)^2|\mathcal{F}_t]$ by Proposition~\ref{prop:E_T_2}, $\E[S_{t+\Delta}|\mathcal{F}_t]$ and $\E[S^2_{t+\Delta}|\mathcal{F}_t]$ can be calculated with the characteristic function in Equation~\eqref{eq:charac_nig},  $\E[S_{t+\Delta}(\bar{T}-T_{t+\Delta})^+|\mathcal{F}_t]$ by Proposition~\ref{prop:eX1} and $\E[S_{t+\Delta}^2(\bar{T}-T_{t+\Delta})^+|\mathcal{F}_t]$  by Proposition~\ref{prop:eS2}. 
    Finally we compute $\E[S_{t+\Delta}((\bar{T}-T_{t+\Delta})^+)^2|\mathcal{F}_t]$ by using Proposition~\ref{prop:S_T^2}.
\end{proof}

Figures~\ref{fig:portfolio_jan} show the results of the daily portfolio optimisation during 31 days starting the 1st January 2018.  First,  we can observe that $c^0$ and $c^1$ present a seasonality. This is explained by the integration of weekend days where there is small energy trade while the constant $c^0$ and $c^1$ are related to non-seasonal tools. This seasonality is not present in $c^2$ as the instrument hedged by the $c^2$ has the same seasonality as the output. Second, we can also comment on the signs of $c^1$ and $c^2$ which are both positive. This shows that both instruments are used to hedge the output $S_{t+\Delta}(\bar{T}-T_{t+\Delta})^+-c^0_{t,t+\Delta}-c^1_{t,t+\Delta}(\bar{T}-T_{t+\Delta})^+-c^2_{t,t+\Delta} S_{t+\Delta}$. The cash quantity $c^0$ accommodates to respond to the minimisation. Third, we can observe that for $\lambda$ non-zero a small share of $c^0$ is reported to $c^2$ as we are leveraging the dependence structure between the energy and the temperature.

\begin{figure}[ht!]
    \begin{minipage}[b]{0.32\linewidth}
        \centering
        \includegraphics[width=\textwidth]{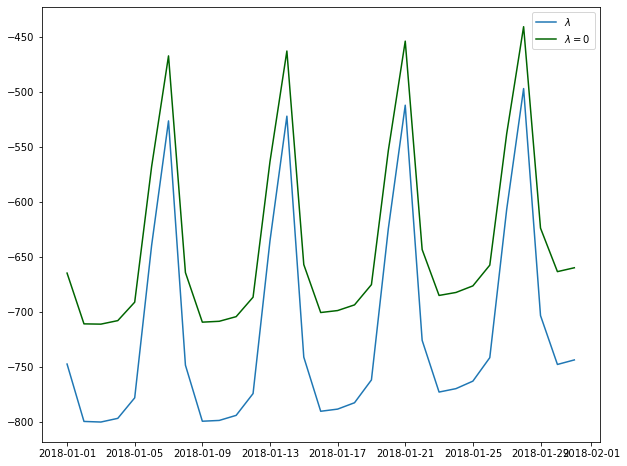}
    \end{minipage}
    \hspace{0.0cm}
    \begin{minipage}[b]{0.31\linewidth}
        \centering
        \includegraphics[width=\textwidth]{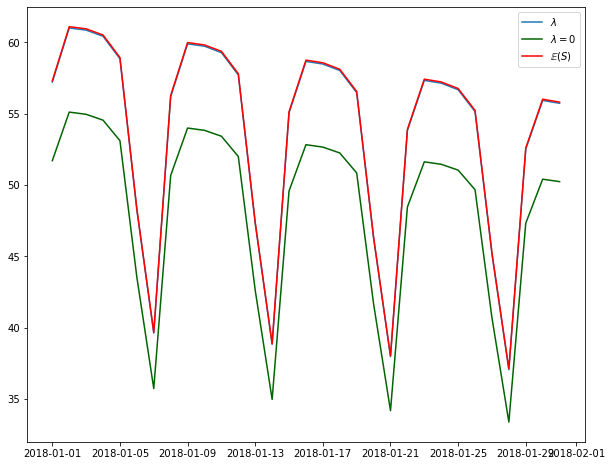}
    \end{minipage}
    \begin{minipage}[b]{0.32\linewidth}
        \centering
        \includegraphics[width=\textwidth]{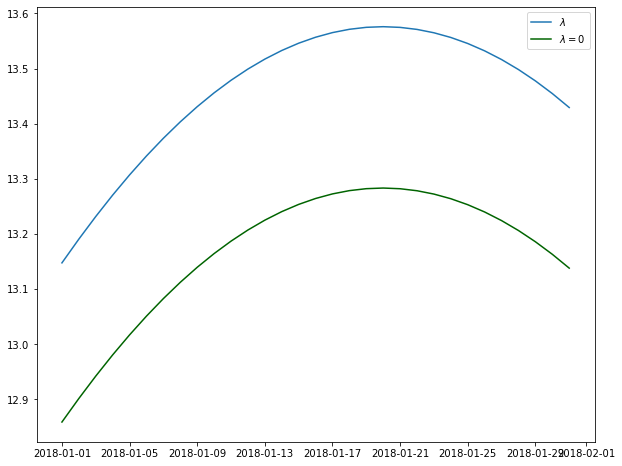}
    \end{minipage}
    \caption{From top left to bottom right, $c^0_{t_0,t_1+i\Delta}$, $c^1_{t_0,t_1+i\Delta}$ and $c^2_{t_0,t_1+i\Delta}$  starting from 1st January 2018 ($t_1$) and with $t_0 = t_1 - 30$, $\Delta=1$ and $i=0,\dots,30$.} \label{fig:portfolio_jan}
\end{figure}

Figure~\ref{fig:hist_portfolio} shows the empirical density of the portfolio only composed by $S_{t+\Delta}(\bar{T}-T_{t+\Delta})^+$ and including hedging for the month of January and May 2018. We compare the $100,000$ Monte Carlo simulations of the portfolio with and without hedging. We can see that the hedging is efficient as the average of the hedge portfolio is 0 while the average of the portfolio without hedging is negative. In addition, the hedging strategy also decreases the variances of a portfolio payoff as the portfolio without hedging is clearly more spread than the one hedged. In the example of January (resp. May), the average payoff of the $\mathcal{E}$-HDD is $22,055$ (resp. $3,014$) and its standard deviation is $4,127$ (resp. $1,487$). In contrast, the average PnL of the static hedging portfolio is 
$-0.754$ (resp. $-0.389$ in May) and its standard deviation is $534$ (resp. $276$).\\

Furthermore, we analyse the impact of $\lambda <0$ supposing a portfolio hedged by using the model with $\lambda=0$. While the hedging of the portfolio still is effective the average PnL of the portfolio  in January (resp. May) is $-137.638$  (resp. $-67.734$) with $\lambda=0$ instead of $-0.754$ (resp. $-0.389$) with the correct value of~$\lambda$, and the corresponding standard deviation is $597$ (resp. $282$) with $\lambda=0$ instead of $534$ (resp. $276$)  with the correct value of $\lambda$. This shows that $\lambda$ has some influence on the quality of the hedge, and the portfolio hedging effectiveness when using options on temperature and energy as hedging instruments.

\begin{figure}[ht!]
    \begin{minipage}[b]{0.45\linewidth}
        \centering
        \includegraphics[width=\textwidth]{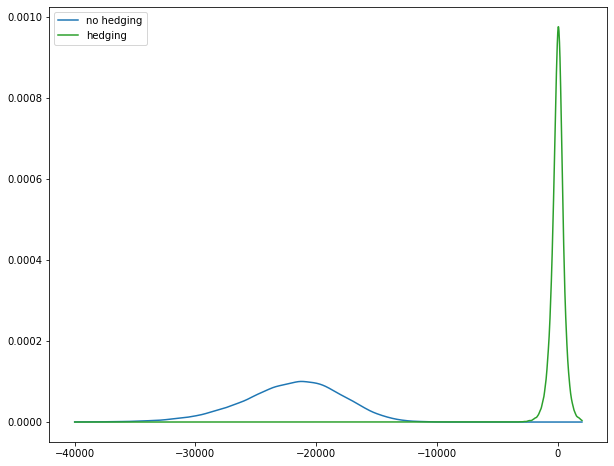}
    \end{minipage}
    \hspace{1.3cm}
    \begin{minipage}[b]{0.45\linewidth}
        \centering
        \includegraphics[width=\textwidth]{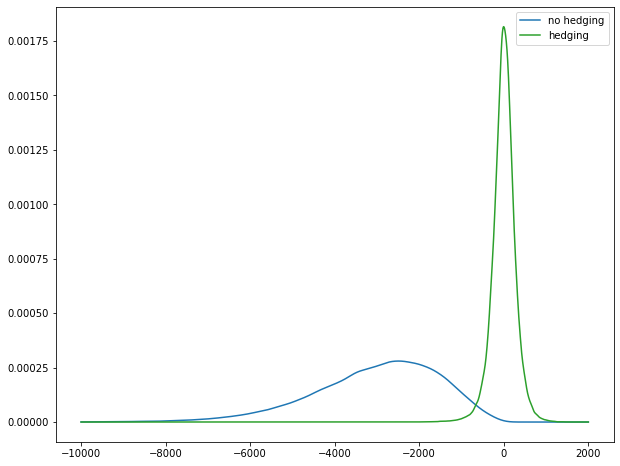}
    \end{minipage}
    \caption{Empirical density of $\sum_{i=1}^{31}-S_{t+i\Delta}(\bar{T}-T_{t_1+i\Delta})^+$ (blue) and $ \sum_{i=1}^{31} c^0_{t_0,t_1 +(i-1)\Delta}+c^1_{t_0,t_1 +(i-1)\Delta}(\bar{T}-T_{t_1+i\Delta})^+ +c^2_{t_0,t_1 +(i-1)\Delta} S_t-S_{t_1+i\Delta}(\bar{T}-T_{t_1+i\Delta})^+$ (green) for portfolio optimisation starting on 1st January 2018 (for $t_1$ on the left) and 1st May 2018 (for $t_1$ on the right), with $t_0 =t_1-30$ and lasting the whole month.} \label{fig:hist_portfolio}
\end{figure}

\subsubsection{Risk decomposition of quantos}
We replicate the above exercise for portfolios including quantos. We suppose that we are at time $t$ and that we want to find the static portfolio made with single $HDD$, puts on electricity and cash that minimizes the square hedging error:
\begin{equation} \label{eq:min_quanto}
    \E[((S_{t+\Delta}- \Bar{S})^+(\bar{T}-T_{t+\Delta})^+-d^0_{t,t+\Delta}-d^1_{t,t+\Delta}(\bar{T}-T_{t+\Delta})^+-d^2_{t,t+\Delta} (S_{t+\Delta}- \Bar{S})^+)^2|\mathcal{F}_t].
\end{equation}

This is again a quadratic function with respect to $(d^0_{t,t+\Delta},d^1_{t,t+\Delta},d^2_{t,t+\Delta})$ and the first order condition leads to:
\begin{equation} \label{eq:quad_quantos}
\begin{aligned}
\left[\begin{matrix}
 1 & \E[(\bar{T}-T_{t+\Delta})^+|\mathcal{F}_t] & \E[(S_{t+\Delta}-\Bar{S})^+|\mathcal{F}_t]   \\    \E[(\bar{T}-T_{t+\Delta})^+|\mathcal{F}_t]  & \E[((\bar{T}-T_{t+\Delta})^+)^2|\mathcal{F}_t] & \E[(S_{t+\Delta}-\Bar{S})^+(\bar{T}-T_{t+\Delta})^+|\mathcal{F}_t]\\ 
            \E[(S_{t+\Delta}-\Bar{S})^+|\mathcal{F}_t]  & \E[(S_{t+\Delta}-\Bar{S})^+(\bar{T}-T_{t+\Delta})^+|\mathcal{F}_t] & \E[((S_{t+\Delta}-\Bar{S})^+)^2|\mathcal{F}_t]
    \end{matrix}\right]\left[\begin{matrix}d^0_{t,t+\Delta} \\d^1_{t,t+\Delta} \\d^2_{t,t+\Delta} \end{matrix}\right] \\
    =\left[\begin{matrix} \E[(S_{t+\Delta}-\Bar{S})^+(\bar{T}-T_{t+\Delta})^+ |\mathcal{F}_t]  \\\E[(S_{t+\Delta}-\Bar{S})^+((\bar{T}-T_{t+\Delta})^+)^2 |\mathcal{F}_t] \\ \E[((S_{t+\Delta}-\Bar{S})^+)^2(\bar{T}-T_{t+\Delta})^+ |\mathcal{F}_t] \end{matrix}\right]
\end{aligned}
\end{equation}

\begin{prop}\label{prop_hedge_quanto}
    Under  Model~\eqref{eq:comb} and for $\Delta >0$, the vector  $(d^0_{t,t+\Delta},d^1_{t,t+\Delta},d^2_{t,t+\Delta})$ minimising the quadratic criterion~\eqref{eq:min_quanto} is the unique solution of the linear system~\eqref{eq:quad_quantos}. The first order Taylor development when $\lambda\to 0$ of all components of this linear system can be explicitly or semi-explicitly calculated.
\end{prop}

\begin{proof}
    The arguments assuring the existence of a unique minimizer are the same as in Proposition~\ref{prop_hedge_eHDD}.
    All the terms above have already been implemented in the Subsection~\ref{subsec:expected_values} except for $\E[((S_{t+\Delta}-\Bar{S})^+)^2|\mathcal{F}_t]$, $\E[(S_{t+\Delta}-\Bar{S})^+((\bar{T}-T_{t+\Delta})^+)^2 |\mathcal{F}_t]$ and $\E[((S_{t+\Delta}-\Bar{S})^+)^2(\bar{T}-T_{t+\Delta})^+ |\mathcal{F}_t]$. The calculation of the first one can be made by using the Carr-Madan approach as presented in Proposition~\ref{prop_carrmadan_sq}, while the Taylor developments of the two other terms are given respectively by Proposition~\ref{prop:CallS_T^2} and Proposition~\ref{prop:S^2_T}.
\end{proof}

Figure~\ref{fig:portfolio_jan_quanto} shows the coefficients $d^0_{t,t+\Delta}$, $d^1_{t,t+\Delta}$ and $d^2_{t,t+\Delta}$ evolution on the month of January 2018. We can observe  that $d^0_{t,t+\Delta}$ and $d^1_{t,t+\Delta}$ are, as before, weekly seasonal. In addition, in this case $d^0_{t,t+\Delta}$ and $d^1_{t,t+\Delta}$ are very close for both $\lambda$ nill and negative. For $d^2_{t,t+\Delta}$, we only get the weekly phenomenon when $\lambda$ is negative.

Figure~\ref{fig:hist_portfolio_quanto} shows the empirical density of the PnL of the portfolio for January and May 2018. In both cases we observe a significant hedging effect. The hedging effect is more important in January as the double condition on the options are hit more frequently. In May, the quanto does not claim so often and the PnL Monte Carlo simulations are closer to 0. In both cases, the PnL of the portfolio is reduced on average from $-3,358$ to $0.208$ in January (from $-89.174$ to $-0.113$ in May) and on standard deviation from $2,197$ to $391$  in January (from $177$  to $98$ in May).

We again compare the hedging obtained with the exact value of $\lambda<0$ (coupled model) and the one obtained with $\lambda=0$ (independent dynamics). In January (resp. May), the average PnL of the portfolio with $\lambda$ is $0.208$ (resp. $-0.113$) instead of $-93.072$ (resp. $-12.283$) for $\lambda=0$ for January. The standard deviation with $\lambda$ is  $391$ (resp. $99$) instead of $394$ (resp. $100$) for $\lambda=0$ respectively. This again shows that the parameter $\lambda$ has some notable influence on the portfolio hedging: while $\lambda$ is small, the hedging constructed with this value performs better than the one using $\lambda=0$.

\begin{figure}[htb!]
    \begin{minipage}[b]{0.32\linewidth}
        \centering
        \includegraphics[width=\textwidth]{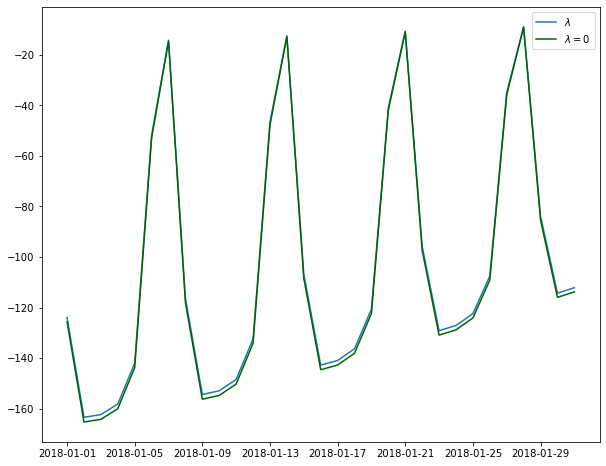}
    \end{minipage}
    \hspace{0.0cm}
    \begin{minipage}[b]{0.31\linewidth}
        \centering
        \includegraphics[width=\textwidth]{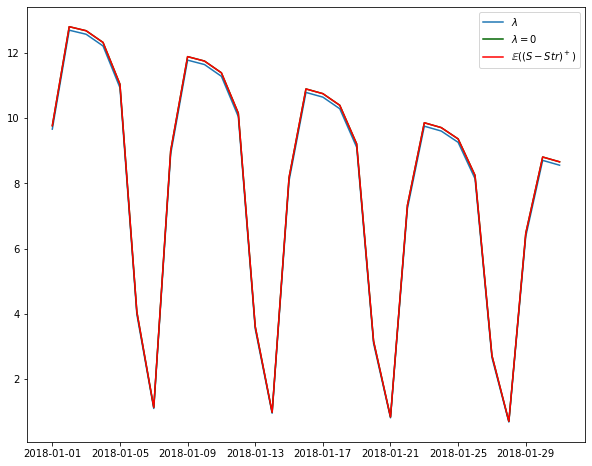}
    \end{minipage}
    \begin{minipage}[b]{0.32\linewidth}
        \centering
        \includegraphics[width=\textwidth]{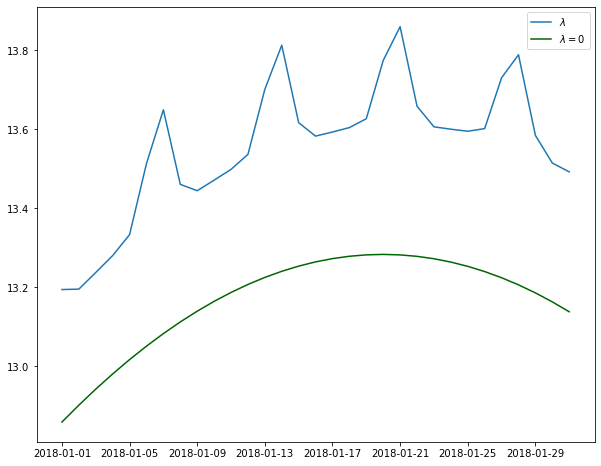}
    \end{minipage}
    \caption{From top left to bottom right, $d^0_{t_0,t_1+i\Delta}$, $d^1_{t_0,t_1+i\Delta}$and $d^2_{t_0,t_1+i\Delta}$ starting from 1st January 2018 ($t_1$), with $t_0 =t_1-30$ and with $t_0 = t_1 - 30$, $\Delta=1$ and $i=0,\dots,30$.} \label{fig:portfolio_jan_quanto}
\end{figure}

\begin{figure}[h!]
    \begin{minipage}[b]{0.45\linewidth}
        \centering
        \includegraphics[width=\textwidth]{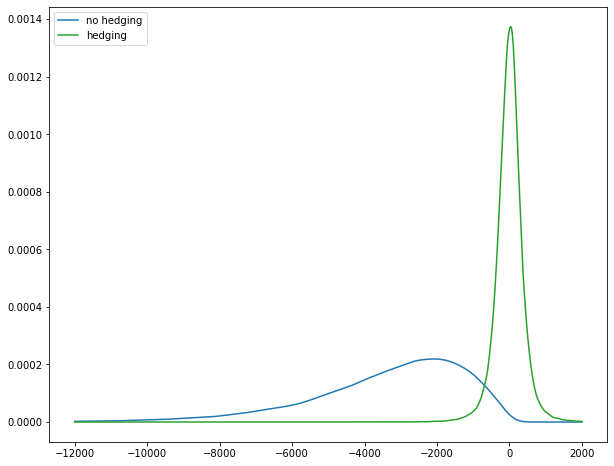}
    \end{minipage}
    \hspace{1.3cm}
    \begin{minipage}[b]{0.45\linewidth}
        \centering
        \includegraphics[width=\textwidth]{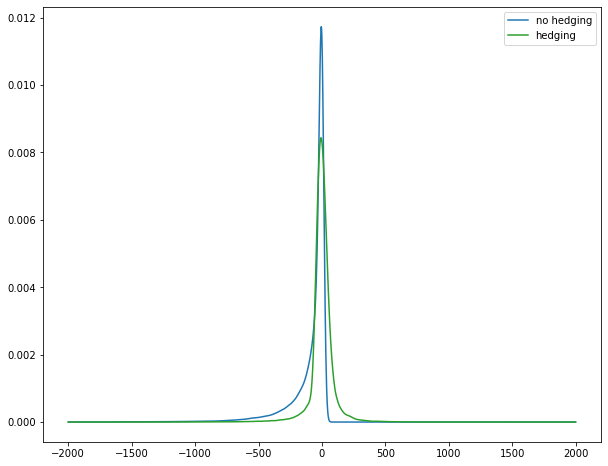}
    \end{minipage}
    \caption{Empirical density of $\sum_{i=1}^{31}-(S_{t_1+i\Delta}-\Bar{S})^+(\bar{T}-T_{t_1+i\Delta})^+$ (blue) and $ \sum_{i=1}^{31} d^0_{t_0,t_1 +(i-1)\Delta}+d^1_{t_0,t_1 +(i-1)\Delta}(\bar{T}-T_{t_1+i\Delta})^+ +d^2_{t_0,t_1 +(i-1)\Delta} (S_{t_1+i\Delta}-\Bar{S})^+-(S_{t_1+i\Delta}-\Bar{S})^+(\bar{T}-T_{t_1+i\Delta})^+$ (green) for portfolio optimisation starting on 1st January 2018 (for $t_1$ on the left) and 1st May 2018 (for $t_1$ on the right), with $t_0=t_1-30$ and lasting the whole month.} \label{fig:hist_portfolio_quanto}
\end{figure}

To sum up, the portfolio study above shows that one can effectively hedge $\mathcal{E}$-HDD and double-sided option quantos through single index-based derivatives. This single-index-based derivatives are traded in open market which ease accessibility and decrease operational costs. For energy derivatives, we can even consider these markets as liquid and suppose a perfect hedging of this risk component. Understanding this risk decomposition is key for risk managers and risk transfer businesses as it enables to gain comfort on the product and ensure meeting their own solvency constraints. Given this is a highly regulated economic sector, regulators are also concerned about this hedging capacity.

Finally, in this section we address the pricing of quanto derivatives on temperature and electricity. We explore different payoff functions and develop explicit pricing formulas for swaps, futures, single sided quanto options and double sided quanto options. These formulas are verified through Monte Carlo simulations. Finally we explore the possibility to statistically hedge single and double sided quanto options. We obtain an efficient daily risk decomposition of these derivatives leading to a averaged-simulated complete hedging of these derivatives. This capacity to hedge these derivatives is key to confirm the efficiency and market viability of these products.
    
\bibliographystyle{plain} 
\bibliography{references.bib} 

\appendix

\section{The Normal Inverse Gaussian (NIG) distribution} \label{appendix:NIG}

This paragraph recalls the parametrization of NIG distributions and some elementary properties. The NIG distribution is a generalised hyperbolic distribution introduced by Barndorff-Nielsen~\cite{barndorff1997normal}. Its density function is defined as follow:
\begin{equation} \label{eq:NIG}
    f(x ; \alpha, \beta, \delta, m) = {\frac {\alpha \delta K_{1}\left(\alpha {\sqrt {\delta ^{2}+(x-m )^{2}}}\right)}{\pi {\sqrt {\delta ^{2}+(x-m )^{2}}}}}\;e^{\delta \gamma +\beta (x-m )}, x \in \R
    \end{equation} 
where $m \in \R$ is the location of the density, $\beta \in \R$, $\alpha > |\beta|$, and $\delta\in \R$ the scale and $K_{1}$ denotes a modified Bessel function of the second kind. We denote by $NIG(\alpha,\beta,\delta,m)$ this law and set $\gamma = \sqrt{\alpha^2 - \beta^2}>0$. The moments and the characteristic function are known explicitly: for $X\sim NIG(\alpha,\beta ,\delta,m)$, we have
\begin{align*}
   & \E(X) = m + \frac{\delta \beta}{\gamma}  ,  &&    Var(X) = \frac{\delta \alpha^2}{\gamma^3} \\
    &Skewness(X) = \frac{3 \beta}{\alpha \sqrt{\gamma \delta}} , &&     Ex. Kurtosis(X) = 3 \frac{1+ 4 \beta^2/ \alpha^2}{\delta \gamma}.
\end{align*}
The characteristic function is given by
$$\E[e^{iuX}]= e^{i u m + \delta \left(\gamma - \sqrt{\alpha^2- (\beta+ iu )^2} \right)},  \ u \in \R.$$

\section{CLS estimator of the drift parameters of the log spot price process} \label{appendix:T}

We consider Model~\eqref{eq:comb}, and we want to estimate the mean-reversion parameter $\kappa_X$ as well as the parameters defining the function $\mu_X(\cdot)$ given by~\eqref{eq:mu}.

The goal of this appendix is to compute the conditional least squares estimator of $(\kappa_X, \beta^X_0, \alpha^X_1, \beta^X_1, \alpha^{X,DoW}_0,\dots,\alpha^{X,DoW}_6  )$ and to prove the next proposition. We note ${\alpha}^{X,DoW}=(\alpha^{X,DoW}_0,\dots,\alpha^{X,DoW}_6)$ and define, for $j\in \N$, $\alpha^{X,DoW}_j=\alpha^{X,DoW}_{\tilde{j}}$, with $\tilde{j}\in \{0,\dots 6\}$ such that $j=\tilde{j} \mod 7$. 

\begin{prop}\label{prop_estimX}
Let $\Xi_{i\Delta} = (i\Delta,  X_{i\Delta}, \sin(\xi  i\Delta), \cos(\xi  i\Delta), \left( \mathbbm{1}_{\{DoW(i\Delta) = j  \}}\right)_{j=0,\dots,6} ) \in \R^4\times \{0,1\}^7$ for $i \in \N$  with $(X_t)_{t \ge 0}$ following the dynamics of~\eqref{eq:comb} and $\Delta>0$. We assume that $\sum_{i=0}^{N-1} \Xi_{i\Delta} \Xi^\top_{i\Delta}$ is invertible and define 
 \begin{equation}\label{def_hlambda_app}
\hat{\eta} = (\hat{\eta}_1,\dots \hat{\eta}_{11})^\top= \left( \sum_{i=0}^{N-1} \Xi_{i\Delta} \Xi^\top_{i\Delta}  \right)^{-1} \left(\sum_{i=0}^{N-1} \Xi_{i\Delta}  X_{(i + 1)\Delta} \right). 
\end{equation} 
If $\hat{\eta}_2 \in (0,1)\cup (1,+\infty)$, the solution of the minimisation problem

\begin{mini}|s| 
{\kappa_X, \beta^X_0, \alpha^X_1, \beta^X_1, {\alpha}^{X,DoW}}{\sum_{i=0}^{N-1} \left( X_{(i + 1)\Delta} - \mathbb{E}  [X_{(i + 1)\Delta} | X_{i \Delta} ] \right)^2}{}{}
\label{minim_X}
\end{mini}

\noindent is given by

\begin{equation*}
\begin{cases}
  \hat{\kappa}_X &= - \ln{\hat{\eta} _2} \\
  \hat{\beta}^X_0&= \frac{\hat{\eta}_1}{1-\hat{\eta}_2} \\
  \hat{\alpha}^X_1 &= \frac{\hat{\eta}_3 (\cos(\xi )-e^{-\hat{\kappa}_X \Delta}) + \hat{\eta}_4 \sin(\xi \Delta )}{(\cos(\xi \Delta)-e^{-\hat{\kappa}_X \Delta})^2 + \sin^2(\xi \Delta )}\\
  \hat{\beta}^X_1 &= \frac{\hat{\eta}_4 (\cos(\xi )-e^{-\hat{\kappa} \Delta }) - \hat{\eta}_3 \sin(\xi \Delta )}{(\cos(\xi \Delta )-e^{-\hat{\kappa}\Delta })^2 + \sin^2(\xi \Delta)}\\
  \hat{\alpha}^{X,DoW}_{j}&= \frac{1}{1-e^{- 7\hat{\kappa}_X \Delta }} \sum_{k=0}^6 (\hat{\eta}^{DoW}_{j + k} - \hat{\beta}_0) e^{-(6-k)\hat{\kappa}_X \Delta }, \\
\end{cases}
\end{equation*}
with $(\hat{\eta}^{DoW}_0,\dots,\hat{\eta}^{DoW}_6)=(\hat{\eta}_5,\dots,\hat{\eta}_{11})$ and $\hat{\eta}^{DoW}_j=\hat{\eta}^{DoW}_{\tilde{j}}$, with $\tilde{j}\in \{0,\dots 6\}$ such that $j=\tilde{j} \mod 7$.

\label{prop-ap1}
\end{prop}

\begin{proof} From~\eqref{eq:comb_int}, we get
$$\E[X_{t+\Delta}|\cF_t]=X_t e^{-\kappa_X \Delta} +\mu_X(t+\Delta)-\mu_X(t)e^{-\kappa_X\Delta}$$
by using the martingale property of the stochastic integral and the fact that $L^X$ is centered. 
We now use trigonometric identities to get
\begin{align*}
\mu_X(t+\Delta)-e^{-\kappa_X \Delta}\mu_X(t)=& \beta_0 (t + \Delta) - \beta_0 e^{-\kappa_X  \Delta}t  + \alpha_1 \sin(\xi (t+\Delta)) - \alpha_1 e^{-\kappa_X \Delta} \sin(\xi t) \\
& + \beta_1 \cos(\xi (t+\Delta)) - \beta_1 e^{-\kappa_X \Delta} \cos(\xi t) + \sum_{j=0}^6 \alpha_j^{DoW} \mathbbm{1}_{\{DoW(t+\Delta)= j  \}} \\
&- \alpha_j^{DoW} e^{-\kappa_X  \Delta} \mathbbm{1}_{\{\lfloor DoW(t)= j \}} \\
&=  \eta_1 t + \eta_3 \sin(\xi t) + \eta_4  \cos(\xi t) + \sum_{j=0}^6 \eta_j^{DoW} \mathbbm{1}_{\{\lfloor DoW(t)= j  \}},
\end{align*}

\noindent with

\begin{equation}
  \begin{cases}
  \eta_1 &= \beta_0(1-e^{-\kappa_X \Delta}) \\
  \eta_2 &= e^{-\kappa_X \Delta} \\
  \eta_3 &=\alpha_1 (\cos(\xi \Delta)-e^{-\kappa_X \Delta})-\beta_1\sin (\xi \Delta) \\
  \eta_4 &=\alpha_1 \sin(\xi \Delta)+\beta_1(\cos (\xi \Delta)- e^{-\kappa_X \Delta})\\
  \eta_j^{DoW} &= \alpha_{j+1}^{DoW}- \alpha_j^{DoW} e^{-\kappa_X \Delta} + \beta_0 \Delta  \text{ with convention } \alpha_7^{DoW}=\alpha_0^{DoW},\\
  \end{cases} 
\end{equation} 
where $\eta_2$ is set to have $\E[X_{(i+1)\Delta}|\cF_{i\Delta}]=\eta^\top \Xi_{i\Delta}$, i.e. is the regression coefficient with respect to~$X_{i\Delta}$. The minimization problem~\eqref{minim_X} is then equivalent to
\begin{mini*}|s|
{\eta \in \R^{11}}{\sum_{i=0}^{N-1} \left( X_{(i + 1)\Delta} - \eta^\top \Xi_{i\Delta}   \right)^2}{}{}.
\end{mini*}

\noindent This corresponds to a linear regression, whose solution is given by~\eqref{def_hlambda_app}. When $\eta_2 \in (0,1)$, the system can be inverted, and the claim follows easily. 
\end{proof}
\noindent Let us note here that $\hat{\eta}^\top \Xi_{i\Delta}$ can then be seen as the estimation of $\mathbb{E}  [X_{(i+1) \Delta} | X_{i \Delta} ]$.

\section{Simulation of Model~\eqref{eq:comb} and associated characteristic function} \label{appendix:Simu_MNIG}

Let recall~\eqref{eq:comb_int}:
\begin{equation} 
   \begin{dcases}
    X_{t+\Delta}-\mu_X(t+\Delta)  &= e^{-\kappa_X \Delta }(X_t-\mu_X(t)) + \lambda \sigma_T \int_t^{t+\Delta} e^{-\kappa_X (t+\Delta-v)} dW^T_v   + \int_t^{t+\Delta} e^{-\kappa_X (t+\Delta-v)} dL^X_v\\
    T_{t+\Delta}-\mu_T(t +\Delta)  &= e^{-\kappa_T \Delta}(T_t-\mu_T(t)) + \sigma_T \int_t^{t+\Delta} e^{-\kappa_T (t+\Delta-v)}  dW^T_v.
 \end{dcases} 
\end{equation}
The simulation algorithm is the following:
\begin{enumerate}
    \item Simulate $ N_1 \sim \mathcal{N}(0, 1)$ and $ N_2 \sim \mathcal{N}(0, 1)$. 
    \item Simulate $Z^X\sim NIG(\alpha^X, \beta^X,\delta^X, -\frac{\delta^X \beta^X}{\gamma^X})$ (we work with centered NIG distributions). 
    \item Simulate $(\tilde{X}_{t+\Delta}, \tilde{T}_{t+\Delta} )$ given $(\tilde{X}_t, \tilde{T}_t )$ using the below scheme
    \begin{equation*} 
   \begin{dcases}
    X_{t+\Delta} &= \mu_X(t+\Delta) + e^{-\kappa_X \Delta}(X_t-\mu_X(t)) +  \lambda \sigma_T \sqrt{\frac{1- e^{-2 \kappa_X\Delta}}{2 \kappa_X} } N_1+  e^{-\kappa_X \Delta /2}Z^X \\
    T_{t+\Delta}  &= \mu_T(t+\Delta) + e^{-\kappa_T \Delta}(T_t-\mu_T(t)) +  \sigma_T \sqrt{\frac{1- e^{-2 \kappa_T \Delta}}{2 \kappa_T}} (\rho N_1+\sqrt{1-\rho^2}N_2),
 \end{dcases} 
\end{equation*}

with $\rho$ defined as in Proposition~\ref{prop:vector}.
\end{enumerate}
Note that this is the exact scheme for $T$, and the only discretization error on~$X$ comes from the approximation of $\int_t^{t+\Delta} e^{-\kappa_X (t+\Delta-v)} dL^X_v$. When comparing the pricing by Monte-Carlo with the formulas using the Fourier transform as in Section~\ref{Sec_risk_quanto}, it is then worth to use the characteristic function associated to this discretization scheme. This avoids to have a bias between both methods. Namely, we use for $t_0<t$ such that $t-t_0$ is a multiple of the discretization step $\Delta$
\begin{equation} \label{eq:charac_nig_update_approx}
   \begin{aligned}  \hat{\psi}_{X_t}(u; t_0-t) =& { e^{iu (\mu_X(t)+e^{-\kappa_X (t-t_0)} (X_t - \mu_X(t))-{\frac {1}{2}}\lambda^{2}\sigma_T^{2}  \frac{1- e^{-2\kappa_X (t-t_0) }}{2\kappa_X}  u^{2}}} \times \\
   &\exp\Bigg(\Delta \sum_{\ell=0}^{\frac{t-t_0}\Delta -1}\Bigg(i u m^X e^{-\kappa_X (\ell+1/2) \Delta} + \delta^X \gamma^X  - \delta^X \sqrt{(\alpha^X)^2- (\beta^X+ iue^{-\kappa_X (\ell+1/2) \Delta})^2}  \Bigg) \Bigg).
   \end{aligned}
\end{equation}
instead of~\eqref{eq:charac_nig_update}. 

\section{Proofs of the results of Section~\ref{Sec_risk_quanto}} 

\subsection{Results on the dependence between $X$ and $T$}
From~\eqref{eq:comb_int}, we are interested in the law of $(\int_t^{t+ \Delta} e^{-\kappa_X(t+\Delta-v)} dW^T_v, \int_t^{t+ \Delta} e^{-\kappa_T(t +\Delta-v)} dW^T_v )$ that captures the dependence between $X_{t+\Delta}$ and $T_{t+\Delta}$ given $X_t$ and $T_t$.

\begin{prop} \label{prop:vector}
The random vector $\begin{pmatrix}
            \int_t^{t+ \Delta} e^{-\kappa_X(t+\Delta-v)} dW^T_v  \\
            \int_t^{t+ \Delta} e^{-\kappa_T(t+\Delta-v)} dW^T_v 
    \end{pmatrix}$ is a centered Gaussian vector with covariance matrix

\begin{equation*} 
    \bfK(\Delta):=\left[\begin{matrix}k_{X}^2(\Delta) & k_{XT}^2(\Delta)\\
           k_{XT}^2(\Delta) & k_T^2(\Delta)
\end{matrix}\right]=\left[\begin{matrix}
           \frac{1 - e^{-2\kappa_X\Delta}}{2 \kappa_X} & \frac{1 - e^{-(\kappa_X+\kappa_T)\Delta}}{\kappa_X+ \kappa_T}\\
           \frac{1 - e^{-(\kappa_X+\kappa_T)\Delta}}{\kappa_X+ \kappa_T} &\frac{1 - e^{-2\kappa_T\Delta}}{2 \kappa_T}
    \end{matrix}\right].
\end{equation*} 
    It has the same law as
\begin{align*} 
    \begin{pmatrix}
           k_{X}(\Delta) \varrho \\
           k_{T}(\Delta)
    \end{pmatrix} G +
    \begin{pmatrix}
            k_{X}(\Delta) \sqrt{1-\varrho^2}\\
            0
    \end{pmatrix} G^{\perp} 
  \end{align*}
where $G$ and $G^{\perp}  \sim \mathcal{N}(0,1)$ are independent and $\varrho = \frac{k_{XT}^2(\Delta)}{k_{X}(\Delta) k_{T}(\Delta)}  \in [0,1]$. 
\end{prop}

\begin{proof} The Brownian motion is a Gaussian process, which gives the Gaussian property. We use then the Itô isometry to get the covariance matrix. 
\end{proof}
From Proposition~\ref{prop:vector}, we can quickly get the following corollary.
\begin{corollary} \label{corollary:projection}
Conditionally on  $(\int_t^{t+ \Delta} e^{-\kappa_T(t+\Delta-v)} dW^T_v)$,  $(\int_t^{t+ \Delta} e^{-\kappa_X(t+\Delta-v)} dW^T_v)$ follows a Gaussian distribution with mean
\begin{equation*}
    \E \Big(\int_t^{t+ \Delta} e^{-\kappa_X(t+\Delta-v)} dW^T_v \Big| \int_t^{t+ \Delta} e^{-\kappa_T(t+\Delta-v)} dW^T_v \Big) = \frac{2 \kappa_T}{1 - e^{-2\kappa_T\Delta}} \frac{1 - e^{-(\kappa_X + \kappa_T )\Delta}}{\kappa_X + \kappa_T} \int_t^{t+ \Delta} e^{-\kappa_T(t+\Delta-v)} dW^T_v
\end{equation*}
and variance
$$ \frac{1 - e^{-2\kappa_X\Delta}}{2 \kappa_X}-\Big(\frac{1 - e^{-(\kappa_X + \kappa_T )\Delta}}{\kappa_X + \kappa_T}\Big)^2 \frac{2 \kappa_T}{(1 - e^{-2\kappa_T\Delta})}.$$
\end{corollary}

\subsection{Identities on the normal distribution}
We note $\Phi$ the cumulative distribution function of the normal distribution $\mathcal{N}(0,1)$.
\begin{lemma} \label{lemma:partie_positive}
    Let $G\sim \mathcal{N}(0,1)$, $a \in \R$, $b>0$. We have 
    $$\E[(a+b G)^+]=a \Phi(a/b)+\frac{b}{\sqrt{2 \pi}} e^{-a^2/(2 b^2)}.$$
\end{lemma}
\begin{proof}
We have $\E[(a+b G)^+]=\int_{-a/b}^\infty (a+ b x) \frac{e^{-x^2/2}}{\sqrt{2 \pi}}dx=a \Phi(a/ b)+\frac{b}{\sqrt{2 \pi}} e^{-a^2/(2 b^2)} $.    
\end{proof}

\begin{lemma} \label{lemma:gauss_comb0}
Let $a\in \R$, $b\in \R$ and $G\sim \mathcal{N}(0,1)$.
We have
\begin{equation} 
    \E[bG(a+bG)^+] = b^2\Phi(a/|b|)
\end{equation}
\end{lemma}

\begin{proof}
It is sufficient to prove the result for $b>0$ since $bG\overset{law}=|b|G$.
Let $d=a/b$. We have $\E[bG(a+bG)^+]=b^2 \E[G(d+G)^+]$ and
\begin{align*}
 \E[G(d+G)^+]&=\int_{-d}^\infty x (d+x)\frac{e^{-\frac{x^2}2}}{\sqrt{2\pi}}dx = \Phi(d).\qedhere
\end{align*}
\end{proof}

\begin{lemma} \label{lemma:gauss_comb}
Let $a\in \R$, $b>0$ and $G\sim \mathcal{N}(0,1)$.
We have
$$\E[((a+bG)^+)^2]=(a^2+b^2)\Phi(a/b)+\frac{ab}{\sqrt{2 \pi}}e^{-\frac{a^2}{2b^2}}.   $$
\end{lemma}
\begin{proof}
Let $d=a/b$. We have $\E[((a+bG)^+)^2]=b^2\E[((d+G)^+)^2]$ and
\begin{align*}
 \E[((d+G)^+)^2]&=\int_{-d}^\infty (d^2+2dx +x^2)\frac{e^{-\frac{x^2}2}}{\sqrt{2\pi}}dx = d^2\Phi(d)+ \frac{2d}{\sqrt{2\pi}}e^{-\frac{d^2}2} -\frac{d}{\sqrt{2\pi}}e^{-\frac{d^2}2}+\Phi(d).\qedhere
\end{align*}
\end{proof}

\begin{lemma} \label{lemma:gauss_comb2}
Let $a\in \R$, $b\in \R$ and $G\sim \mathcal{N}(0,1)$.
We have
\begin{equation} 
    \E[bG((a+bG)^+)^2] = |b|^3 \Big( \sqrt{\frac{2}{\pi}} e^{-\frac{1}{2}(\frac{a}{|b|})^2} + 2 \frac{a}{|b|} \Phi(\frac{a}{|b|})\Big)
\end{equation}
\end{lemma}

\begin{proof}
Since $bG$ has the same law as $|b|G$, it is sufficient to consider the case $b> 0$. We have $\E[bG((a+bG)^+)^2]=b^3 \E[G((d+G)^+)^2]$ with $d=\frac{a}{b}$. We then get by integration by parts and then Lemma~\ref{lemma:partie_positive}
\begin{align*}
\E[G((d+G)^+)^2]&=\int_{-d}^\infty  (d + x)^2 x\frac{e^{-\frac{x^2}2}}{\sqrt{2\pi}}dx =2\int_{-d}^\infty  (d + x) \frac{e^{-\frac{x^2}2}}{\sqrt{2\pi}}dx\\
& = 2  \E[(d+G)^+] = \sqrt{\frac{2}{\pi}} e^{-\frac{d^2}{2}} + 2d  \Phi(d).\qedhere
\end{align*}
\end{proof}

\subsection{Computations with Fourier transform}

In the following section, we will develop some conditional expectations for different derivatives. However, some derivatives do not admit explicit formulas and were computed through inverse Fourier methods.

We first use Carr Madan formula~\cite[Equations (5) and (6)]{carr1999option} to compute $\E (( S_{t+\Delta}- \bar{S})^+\mid \mathcal{F}_{t})$ for $t\ge 0, \Delta>0$:
\begin{equation} \label{eq:carr_madan}
    \E (( S_{t+\Delta}- \bar{S})^+ \mid \mathcal{F}_{t}) = \frac{\exp(-\alpha \bar{X})}{\pi} \int_0^{\infty} e^{-i \Bar{X} v} \frac{\psi_{X_{t}}(v-(\alpha+1)i)}{\alpha^2 +\alpha - v^2 +i(2\alpha +1)v} dv,
\end{equation}
where $\psi$ corresponds to the characteristic function as in Equation~\eqref{eq:charac_nig_update}, $\alpha>0$ and $\Bar{X}:= \ln(\Bar{S})$. Note that from~\eqref{eq:charac_nig_update}, we have $\E[S_{t+\Delta}^{1+\alpha}|\cF_t]<\infty$ a.s. and is equal to~$\psi_{X_{t}}(-(\alpha+1)i; \Delta)$. In practice, we take $\alpha=0.5$ for~\eqref{eq:carr_madan} as well as for Proposition~\ref{prop_carrmadan_sq} below.\\

Second, we apply Gil-Pelaez~\cite{gil1951note} inversion formula to compute $\mathbb{P} \left(   S_{t+\Delta} \geq \bar{S}   \mid \mathcal{F}_{t} \right )$.

\begin{equation} \label{eq:gil-pelaez}
    \mathbb{P}_{\lambda=0} \left(   S_{t+\Delta} \geq \bar{S}   \mid \mathcal{F}_{t} \right ) = \frac{1}{2} + \frac{1}{\pi} \int_0^{\infty} \mathcal{R}\Big(\frac{e^{-iv \bar{X}}\psi_{X_{t}}(v)}{iv} \Big)dv
\end{equation}
where $\mathcal{R}$ denotes the real part, $\psi$ the characteristic function as in Equation~\eqref{eq:charac_nig_update} and $\Bar{X}:= \ln(\Bar{S})$. \\

Third, we leverage again Carr Madan~\cite{carr1999option} approach to compute $\E ((( S_{t+\Delta}- \bar{S})^+)^2\mid \mathcal{F}_{t})$.

\begin{prop}\label{prop_carrmadan_sq}
    Under  Model~\eqref{eq:comb}, we have
    \begin{equation} \label{eq:carr_madan_square}
    \begin{aligned}
          &  \E ((( S_{t+\Delta}- \bar{S})^+)^2 \mid \mathcal{F}_{t})\\ &= \frac{\exp(-\alpha \bar{X})}{\pi} \int_0^{\infty} e^{-i \Bar{X} v}  \psi(v-(\alpha+2)i; \Delta) \Big(\frac{1}{\alpha+ iv} - 2 \frac{1}{\alpha+1+iv} + \frac{1}{\alpha+2+iv}\Big)  dv
    \end{aligned}
    \end{equation} 
    where $\psi$ corresponds is the characteristic function as in Equation~\eqref{eq:charac_nig_update}, $\alpha>0$ and $\Bar{X}:= \ln(\Bar{S})$.
\end{prop}
\begin{proof}
    We prove the result for a random variable $Y$ with distribution $\mu$ on~$\R$ such that $\E[e^{(2+\alpha)Y}]<\infty$. We define
    \begin{equation*}
        C(k) = \E\Big[\Big((e^Y-e^k)^+\Big)^2\Big]  =\int_k^{\infty} (e^y-e^k)^2 \mu(dy)
    \end{equation*}
 and $c(k) = e^{\alpha k } C(k)$. The function  $c$ is nonnegative and integrable on~$\R$ since
 $$\int_{\R} c(k)dk\le \int_{\R} e^{\alpha k} \E[e^{2Y}\mathbbm{1}_{Y>k}]dk=\frac{\E[e^{(2+\alpha)Y]}}{\alpha}<\infty ,$$ by Fubini's theorem.  Following Carr Madan~\cite{carr1999option}, as $c$ is integrable, we can define its inverse Fourier transform $\Tilde{\psi}: \mathbb{C} \longrightarrow \mathbb{C}$ such that:
 \begin{equation}
     \begin{aligned}
         \Tilde{\psi}(v) &= \int_{-\infty}^{\infty} e^{ivk} c(k) dk \\
         &= \int_{-\infty}^{\infty} e^{ivk} \int_k^{\infty} e^{-\alpha k } (e^k-e^y)^2 \mu(dy) dk \\
         &= \int_{-\infty}^{\infty}  \Big(\frac{1}{\alpha+ iv} - 2 \frac{1}{\alpha+1+iv} + \frac{1}{\alpha+2+iv}\Big)e^{(\alpha +2 +iv)y} \mu(dy)\\
         &= \psi(v-(\alpha+2)i) \Big(\frac{1}{\alpha+ iv} - 2 \frac{1}{\alpha+1+iv} + \frac{1}{\alpha+2+iv}\Big),
     \end{aligned}
 \end{equation}
 by using Fubini's theorem. 
We have $|\tilde{\psi}(v)|\le \E[e^{(2+\alpha)Y}] \frac{2}{|\alpha +iv||\alpha+1+iv||\alpha+2+iv|}$, and thus $\Tilde{\psi(v)}$ is integrable on~$\R$ and bounded. We get then the claim by Fourier inversion and using that $\E[S_{t+\Delta}^{2+\alpha}|\cF_t]=\psi_{X_{t}}(-(\alpha+2)i)<\infty$ a.s. 
\end{proof}
\red{\begin{remark}
 Let us mention that for the purpose of this paper, we only calculate options at a given strike $\bar{S}$ and therefore use a basic numerical integration to calculate the integrals of Equations~\eqref{eq:carr_madan},~\eqref{eq:gil-pelaez} and~\eqref{eq:carr_madan_square}. If we were interested in calculating values at different strikes, we could easily adapt the FFT pricing method of Carr and Madan~\cite{carr1999option} or alternative methods suggested by Fang and Oosterlee~\cite{FaOo} and Kirkby~\cite{Kirkby} to name a few. The Hilbert transform approach developed by Feng and Linetsky~\cite{FeLi} could also be used for more involved options such as Barrier options. 
\end{remark}
}

\subsection{Results to calculate the average payoffs of derivatives}
\subsubsection{Future}

\begin{prop} \label{prop:eXt} 
Under  Model~\eqref{eq:comb}, we have
\begin{align*}
    \E ( S_{t+\Delta} T_{t+\Delta} \mid \mathcal{F}_t) &= \exp\Bigg(\mu_X(t+\Delta) + e^{-\kappa_X \Delta }(X_t-\mu_X(t)) \Bigg) \varphi(-i ;\Delta) \times \\
    &  \Bigg( (\mu_T(t+\Delta) + e^{-\kappa_T \Delta }(T_t-\mu_T(t)))  e^{\frac{1}{2} k_X(\Delta)^2\lambda^2 \sigma_T^2  } + \lambda \sigma_T^2 k_{XT}^2(\Delta) e^{ \frac{1}{2}\lambda^2 \sigma_T^2 k_{X}(\Delta)^2}    \Bigg) 
\end{align*}
where $\varphi$ is the characteristic function defined in Equation~\eqref{eq:charac_nig} and  $k_T(\cdot)$ , $k_X(\cdot)$ and $k_{XT}(\cdot)$ are as in Equation~\eqref{eq:k}.
\end{prop}

\begin{proof}
From~\eqref{eq:comb_int}, we get
\begin{align*}
    \E ( e^{X_{t+\Delta}} T_{t+\Delta} \mid \mathcal{F}_t) &= \E\Bigg(\exp\Bigg(\mu_X(t+\Delta) + e^{-\kappa_X \Delta}(X_{t+\Delta}-\mu_X(t+\Delta))+ \lambda \sigma_T \int_t^{t+\Delta} e^{-\kappa_X(t+\Delta-u)}  dW^T_u \\
    &+ \int_t^{t+\Delta} e^{-\kappa_X(t+\Delta-u)}dL^X_u \Bigg) T_{t+\Delta} \Bigg| \mathcal{F}_t\Bigg) \\
    &= \exp\Bigg(\mu_X(t+\Delta) + e^{-\kappa_X \Delta }(X_t-\mu_X(t)) \Bigg) \E(e^{  \int_t^{t+\Delta} e^{-\kappa_X(t+\Delta-u)}dL^X_u }) \times \\
    & \E\Bigg(e^{\lambda \sigma_T \int_t^{t+\Delta} e^{-\kappa_X(t+\Delta-u)}  dW^T_u} T_{t+\Delta} \Bigg| T_t\Bigg), \\
\end{align*}
since $\int_t^{t+\Delta} e^{-\kappa_X(t+\Delta-u)}dL^X_u$ is independent of $\mathcal{F}_t$ and $T_t$.
The first term is deterministic, the second term is equal to $\varphi(-i ;\Delta)$ by~\eqref{eq_def_phi}. We use~\eqref{eq:comb_int} to write as follows the third term:
\begin{align*}
    \E\Bigg(e^{\lambda \sigma_T \int_t^{t+\Delta} e^{-\kappa_X(t+\Delta-u)}  dW^T_u} T_{t+\Delta} \Bigg| T_t\Bigg) &= (\mu_T(t+\Delta) + e^{-\kappa_T \Delta}(T_t-\mu_T(t))) \E\Bigg(e^{\lambda \sigma_T \int_t^{t+\Delta} e^{-\kappa_X(t+\Delta-u)}  dW^T_u}  \Bigg) \\
    &+ \sigma_T \E\Bigg(e^{\lambda \sigma_T \int_t^{t+\Delta} e^{-\kappa_X(t+\Delta-u)}  dW^T_u}  \int_t^{t+\Delta} e^{-\kappa_T(t+\Delta-u)}  dW^T_u \Bigg).
\end{align*}
From Proposition~\ref{prop:vector}, we get $\E\left(e^{\lambda \sigma_T \int_t^{t+\Delta} e^{-\kappa_X(t+\Delta-u)}  dW^T_u}  \right) = e^{\frac{1}{2} k_{X}(\Delta)^2 \lambda^2 \sigma_T^2  }$ and
\begin{align*}
    \E\Bigg(e^{\lambda \sigma_T \int_t^{t+\Delta} e^{-\kappa_X(t+\Delta-u)}  dW^T_u} & \int_t^{t+\Delta} e^{-\kappa_T(t+\Delta-u)}  dW^T_u \Bigg) = \E\left(e^{\lambda \sigma_T(\frac{k_{XT}^2(\Delta)}{k_{T}(\Delta)}  G +  \frac{\sqrt{k_{X}(\Delta)^2 k_{T}(\Delta)^2 -k_{XT}^4(\Delta)}}{k_{T}(\Delta)}  G^{\perp})} k_{T}(\Delta) G\right) \\
    &=  \E\left(e^{ \lambda \sigma_T \frac{\sqrt{k_{X}(\Delta)^2 k_{T}(\Delta)^2 -k_{XT}^4(\Delta)}}{k_{T}(\Delta)} G^{\perp}} \right) k_{T}(\Delta) \E\left( G e^{\lambda \sigma_T \frac{k_{XT}^2(\Delta)}{k_{T}(\Delta)} G} \right).
\end{align*}
Since $\E[e^{xG}]=e^{x^2/2}$ and $\E[Ge^{xG}]=xe^{x^2/2}$, we get the claim.
\end{proof}

\subsubsection{Swap}

\begin{prop} \label{prop:swap}
Under  Model~\eqref{eq:comb} and for $\Delta>0$ , we have
\begin{align*}
    \E ((S_{t+\Delta} - \bar{S}) (\bar{T}-T_{t+\Delta}) \mid \mathcal{F}_{t}) &=  \bar{T}  \psi_{X_{t}}(-i ;\Delta) -\mathcal{F}(t,t+ \Delta) -\bar{S} \bar{T} + \bar{S}  (\mu_T(t) + e^{-\kappa_T \Delta} (T_{t}- \mu_T(t)))
\end{align*}
where $\psi$ is the characteristic function defined in Equation~\eqref{eq:charac_nig_update}, $\mathcal{F}(t,t +\Delta)$ as in Equation~\eqref{eq:forward} and  $k_X(\cdot)$ is as in Equation~\eqref{eq:k}.
\end{prop}
    
\begin{proof}
    Let first develop the formula:
\begin{align*}
    \E ( (S_{t+\Delta}- \bar{S}) (\bar{T}- T_{t+\Delta}) \mid \mathcal{F}_{t}) &= \bar{T} \E ( S_{t+\Delta}\mid \mathcal{F}_{t}) - \E ( S_{t+\Delta} T_{t+\Delta}\mid \mathcal{F}_{t}) -\bar{S} \bar{T} + \bar{S} E ( T_{t+\Delta} \mid \mathcal{F}_{t})
\end{align*}
We have $\E ( S_{t+\Delta} T_{t+\Delta}\mid \mathcal{F}_{t}) $ from Equation~\eqref{eq:forward} and Proposition~\ref{prop:eXt}.
The first term is equal to:
\begin{align*}
 \Bar{T}\E(S_{t+\Delta} \mid \mathcal{F}_{t}) &= \Bar{T} \psi_{X_t}(-i ;\Delta)
\end{align*}
and, by using the independence between $L^X$ and $W^T$, we get the forth term:
\begin{align*}
    \bar{S} E ( T_{t+\Delta} \mid \mathcal{F}_{t}) = \bar{S} (\mu_T(t+\Delta) + e^{-\kappa_T \Delta} (T_{t}- \mu_T(t))).
\qedhere\end{align*}
\end{proof}

\subsubsection{Quanto}
\begin{prop} \label{prop:E_T_positive}
Under  Model~\eqref{eq:comb} and for $\Delta>0$ , we have
    \begin{equation*}
        \begin{aligned}
            \E[(\Bar{T} -T_{t+\Delta})^+ \mid \cF_t ] 
            &= \Big(\Bar{T} - \mu_T(t+\Delta) - e^{-\kappa_T\Delta} (T_t - \mu_T(t) \Big) \Phi \left(\frac{\Bar{T} - \mu_T(t+\Delta) - e^{-\kappa_T\Delta} (T_t - \mu_T(t) )}{\sigma_T k_T(\Delta)} \right) \\
            &+\frac{\sigma_T k_T(\Delta)}{\sqrt{2 \pi}} \exp\Big(-\frac{1}{2} \Big( \frac{\Bar{T} - \mu_T(t+\Delta) - e^{-\kappa_T\Delta} (T_t - \mu_T(t) )}{\sigma_T k_T(\Delta)}\Big)^2\Big)
        \end{aligned}
    \end{equation*}
\end{prop}

\begin{proof}
    Under  Model~\eqref{eq:comb}, the distribution of $\Bar{T} -T_{t+\Delta}$ given $\cF_t$ is $\mathcal{N}(\Bar{T} - \mu_T(t+\Delta) - e^{-\kappa_T\Delta} (T_t - \mu_T(t) ), \sigma_T^2 k_T(\Delta)^2)$. We then apply Lemma~\ref{lemma:partie_positive} to obtain the result.
\end{proof}

\begin{prop}\label{prop:Taylor_expansion}
    Under  Model~\eqref{eq:comb} and for $\Delta>0$,  we can write the following Taylor expansion on $\lambda$:
    \begin{equation*}
        \begin{aligned}
             \E ( (S_{t+\Delta}- \bar{S})^+ (\bar{T}-T_{t+\Delta})^+ \mid \mathcal{F}_{t}) &=  \E_{\lambda=0}( (S_{t+\Delta}- \bar{S})^+\mid \mathcal{F}_{t}) \times\\
             & \Big(\Big(\Bar{T} - \mu_T(t+\Delta) - e^{-\kappa_T\Delta} (T_t - \mu_T(t)) \Big) \times \\
             &\Phi\Big(\frac{\Bar{T} - \mu_T(t+\Delta) - e^{-\kappa_T\Delta} (T_t - \mu_T(t) )}{\sigma_T k_T(\Delta)}\Big) \\
            &+\frac{\sigma_T k_T(\Delta)}{\sqrt{2 \pi}} \exp\Big(-\frac{1}{2} \Big( \frac{\Bar{T} - \mu_T(t+\Delta) - e^{-\kappa_T\Delta} (T_t - \mu_T(t) )}{\sigma_T k_T(\Delta)}\Big)^2\Big) \Big) \\
            &- \Big( \E_{\lambda=0} ( (S_{t+\Delta}- \bar{S})^+\mid \mathcal{F}_{t}) + \bar{S} \mathbb{P}_{\lambda=0} \left(   S_{t+\Delta} \geq \bar{S}   \mid \mathcal{F}_{t} \right ) \Big) \times \\
            &\sigma_T^2 k_{XT}(\Delta)^2 \Phi\Big(\frac{\bar{T}-\mu_T(t+\Delta) - e^{-\kappa_T \Delta}(T_t -\mu_T(t))}{\sigma_T k_T(\Delta)}\Big) \Big)\lambda + o(\lambda)
        \end{aligned}
    \end{equation*}
    where  $k_T(\cdot)$ , $k_X(\cdot)$ and $k_{XT}(\cdot)$ are as in Equation~\eqref{eq:k}.
\end{prop}

\begin{proof}
    Under  Model~\eqref{eq:comb} and for $\Delta>0$ , we want to compute the first order Taylor expansion of $ \E ( (S_{t+\Delta}- \bar{S})^+ (\bar{T}-T_{t+\Delta})^+ \mid \mathcal{F}_{t})$.
    For $\lambda=0$,
    \begin{equation*}
        \begin{aligned}
        \E ( (S_{t+\Delta}- \bar{S})^+ (\bar{T}-T_{t+\Delta})^+ \mid \mathcal{F}_{t}) &= \E (( S_{t+\Delta}- \bar{S})^+ \mid \mathcal{F}_{t}) \E( ( \bar{T}-T_{t+\Delta})^+ \mid \mathcal{F}_{t})
         \end{aligned}
    \end{equation*}
    We have $\E (( \bar{T}-T_{t+\Delta})^+ \mid \mathcal{F}_{t})$ from Proposition~\ref{prop:E_T_positive} and $\E (( S_{t+\Delta}- \bar{S})^+ \mid \mathcal{F}_{t})$ from Equation~\eqref{eq:carr_madan}.
    
    Let now consider the derivative of $ \E ( (S_{t+\Delta}- \bar{S})^+ (\bar{T}-T_{t+\Delta})^+ \mid \mathcal{F}_{t})$ in $\lambda=0$:
    \begin{equation*}
        \begin{aligned}
            \left. {\frac{d }{d \lambda}} \right|_{\lambda=0} \E ( (S_{t+\Delta}- \bar{S})^+ (\bar{T}-T_{t+\Delta})^+ \mid \mathcal{F}_{t}) &= \E_{\lambda=0} \left(  \mathbbm{1}_{ S_{t+\Delta} \geq \bar{S} } \sigma_T \int_t^{t+\Delta} e^{-\kappa_X(t+\Delta-s)}dW_s e^{X_{t+\Delta}} (\bar{T}-T_{t+\Delta})^+ \mid \mathcal{F}_{t} \right ) \\
            &= \E_{\lambda=0} \left(  \mathbbm{1}_{ S_{t+\Delta} \geq \bar{S} }  e^{X_{t+\Delta}} \mid \mathcal{F}_{t} \right ) \times \\
            & \E_{\lambda=0} \left( \sigma_T  \int_t^{t+\Delta} e^{-\kappa_X(t+\Delta-s)}dW_s (\bar{T}-T_{t+\Delta})^+ \mid \mathcal{F}_{t} \right )
        \end{aligned}
    \end{equation*}
    Let consider the first term,
    $$ \E_{\lambda=0} \left(  \mathbbm{1}_{ S_{t+\Delta} \geq \bar{S} }  S_{t+\Delta}  \mid \mathcal{F}_{t} \right )= \E_{\lambda=0} ( (S_{t+\Delta}- \bar{S})^+\mid \mathcal{F}_{t}) + \bar{S} \mathbb{P}_{\lambda=0} \left(   S_{t+\Delta} \geq \bar{S}   \mid \mathcal{F}_{t} \right )$$
    The first element is computed with Equation~\eqref{eq:carr_madan} and the second with Equation~\eqref{eq:gil-pelaez}.
    
    We can now develop the second term by using Proposition~\ref{prop:vector} and~\eqref{eq:comb_int}:
    \begin{equation} \label{calc_LemD2a}
        \begin{aligned}
            &\E\left( \sigma_T  \int_t^{t+\Delta} e^{-\kappa_X(t+\Delta-s)}dW^T_s (\bar{T}-T_{t+\Delta})^+ \mid \mathcal{F}_{t} \right ) \\
            &=  \frac{k_{XT}(\Delta)^2}{k_T(\Delta)^2} \E\left( \sigma_T   \int_t^{t+\Delta} e^{-\kappa_T(t+\Delta-s)}dW^T_s (\bar{T}-T_{t+\Delta})^+ \mid \mathcal{F}_{t} \right ) \\
            &= \frac{k_{XT}(\Delta)^2}{k_T(\Delta)^2}\E \Bigg( \sigma_T  \int_t^{t+\Delta} e^{-\kappa_T(t+\Delta-s)}dW^T_s  \\
            &\phantom{\frac{k_{XT}(\Delta)^2}{k_T(\Delta)^2} -----}\times
            \left(\bar{T}-\mu_T(t+\Delta) - e^{-\kappa_T \Delta}(T_t -\mu_T(t)) 
            - \sigma_T \int_t^{t+\Delta} e^{-\kappa_T(t+\Delta-s)}dW^T_s  \right)^+ \Bigg| \mathcal{F}_{t}  \Bigg). 
        \end{aligned}
    \end{equation}
    Now we apply Lemma~\ref{lemma:gauss_comb0} with $a=\bar{T}-\mu_T(t+\Delta) - e^{-\kappa_T \Delta}(T_t -\mu_T(t))$ and  $b=-\sigma_T k_T(\Delta)$ to get
    \begin{equation}\label{calc_LemD2b}
        \begin{aligned}
            \E\left( \sigma_T  \int_t^{t+\Delta} e^{-\kappa_T(t+\Delta-s)}dW^T_s (\bar{T}-T_{t+\Delta})^+ \bigg| \mathcal{F}_{t} \right ) 
            &=  -\sigma_T^2 k_T(\Delta)^2 \Phi\left(\frac{\bar{T}-\mu_T(t+\Delta) - e^{-\kappa_T \Delta}(T_t -\mu_T(t))}{\sigma_T k_T(\Delta)} \right)
        \end{aligned}
    \end{equation}
\end{proof}

\subsection{Results for static hedging portfolios}

\subsubsection{Results for $\mathcal{E}$-HDD}

\begin{prop}\label{prop:E_T_2}
Under  Model~\eqref{eq:comb}, we have
\begin{equation*}
    \begin{aligned}
        \E[((\bar{T}-T_{t+\Delta})^+)^2\mid \cF_t] &=\Big((\bar{T}- \mu_T(t+\Delta) - e^{-\kappa_T\Delta} (T_t - \mu_T(t) ))^2+ \sigma_T^2 k_T(\Delta)^2\Big)\times \\
        & \Phi \Big( \frac{\Bar{T} - \mu_T(t+\Delta) - e^{-\kappa_T\Delta} (T_t - \mu_T(t) )}{\sigma_T k_T(\Delta)} \Big) \\
        & + \frac{(\Bar{T} - \mu_T(t+\Delta) - e^{-\kappa_T\Delta} (T_t - \mu_T(t) ))\sigma_T k_T(\Delta)}{\sqrt{2 \pi}}\times \\
        &\exp\Big(-\frac{1}{2} \Big(\frac{\Bar{T} - \mu_T(t+\Delta) - e^{-\kappa_T\Delta} (T_t - \mu_T(t) )}{\sigma_T k_T(\Delta)}\Big)^2 \Big). 
    \end{aligned}
\end{equation*}

\end{prop}

\begin{proof}
Since $\bar{T}-T_{t+\Delta}$ follows a Gaussian distribution given $\cF_t$, we have an explicit formula by using Lemma~\ref{lemma:gauss_comb}.
\end{proof}

\begin{prop} \label{prop:eX1}
Under  Model~\eqref{eq:comb}, we have
\begin{align*}
    \E ( e^{X_{t+\Delta}} \mathbbm{1}_{T_{t+\Delta} \leq u}  \mid \mathcal{F}_t) &= \psi_{X_t}(-i;\Delta)  \Phi\Big( \frac{\tilde{u}(T_t)}{k_T(\Delta) } - \lambda \sigma_T \frac{k_{XT}^2(\Delta)}{k_{T}(\Delta)} \Big).
\end{align*}
\begin{align*}
    \E ( e^{X_{t+\Delta}} \mathbbm{1}_{ u \leq T_{t+\Delta} }  \mid \mathcal{F}_t) =  \psi_{X_t}(-i;\Delta)  \Phi\Big( \lambda \sigma_T \frac{k_{XT}^2(\Delta)}{k_{T}(\Delta)} - \frac{\tilde{u}(T_t)}{k_T(\Delta) } \Big).
\end{align*}
where $\psi$ is the characteristic function defined in Equation~\eqref{eq:charac_nig_update},  $k_T(\cdot)$ , $k_X(\cdot)$ and $k_{XT}(\cdot)$ are as in Equation~\eqref{eq:k} and $\tilde{u}(T_t) = \frac{u - (\mu_T(t+\Delta) + e^{-\kappa_T \Delta}(T_t-\mu_T(t))) }{\sigma_T}$.
\end{prop}

\begin{proof}

From~\eqref{eq:comb_int}, we can write

\begin{align*}
    \E ( e^{X_{t+\Delta}} \mathbbm{1}_{T_{t+\Delta} \leq u}  \mid \mathcal{F}_t) &= \exp\Bigg(\mu_X(t) + e^{-\kappa_X \Delta}(X_t-\mu_X(t)) \Bigg) \E(e^{  \int_t^{t+\Delta} e^{-\kappa_X(t+\Delta-u)}dL^X_u }) \times \\
    &\E\Bigg(e^{\lambda \sigma_T \int_t^{t+\Delta} e^{-\kappa_X(t+\Delta-u)}  dW^T_u} \mathbbm{1}_{T_{t+\Delta} \leq u} \Bigg| \mathcal{F}_t\Bigg) 
\end{align*}
The second term is $\varphi(-i;\Delta)$, see~\eqref{eq_def_phi}. We now consider the last term:
\begin{align*}
    \E\Bigg(e^{\lambda \sigma_T \int_t^{t+\Delta} e^{-\kappa_X(t+\Delta-s)}  dW^T_s} \mathbbm{1}_{T_{t+\Delta} \leq u} \Bigg| \mathcal{F}_t\Bigg) 	&= \E\Bigg(e^{\lambda \sigma_T \int_t^{t+\Delta} e^{-\kappa_X(t+\Delta-s)}  dW^T_s} \mathbbm{1}_{ \int_t^{t+\Delta} e^{-\kappa_X(t+\Delta-s)}  dW^T_s  \leq \tilde{u}(T_t) } \Bigg| \mathcal{F}_t\Bigg) 
\end{align*}
where $\tilde{u}(T_t) = \frac{u - (\mu_T(t+\Delta) + e^{-\kappa_T \Delta}(T_t-\mu_T(t))) }{\sigma_T}  $. For $\tilde{u} \in \R$ and $G,G^{\perp}\sim \mathcal{N}(0,1)$ independent, we have:
\begin{align*}
    \E\Bigg(e^{\lambda \sigma_T(\frac{k_{XT}^2(\Delta)}{k_{T}(\Delta)}  G +  \frac{\sqrt{k_{X}(\Delta)^2 k_{T}(\Delta)^2 -k_{XT}^2(\Delta)}}{k_{T}(\Delta)}  G^{\perp})} \mathbbm{1}_{  k_T(\Delta) G \leq \tilde{u}  }  \Bigg) 
    	&= e^{ \frac{\lambda^2 \sigma_T^2}{2} \frac{k_{X}(\Delta)^2 k_{T}(\Delta)^2 -k_{XT}^4(\Delta)}{k_{T}^2(\Delta)} }  \times\\
     &e^{ \frac{(\lambda \sigma_T )^2}{2} (\frac{k_{XT}^2(\Delta)}{k_{T}(\Delta)})^2} \Phi\Big( \frac{\tilde{u}}{k_T(\Delta) } - \lambda \sigma_T \frac{k_{XT}^2(\Delta)}{k_{T}(\Delta)} \Big),
\end{align*}
because $\E[e^{xG}\mathbbm{1}_{G\le a}]=e^{\frac{x^2}{2}} \Phi(a-x)$ for $x,a\in \R$. Since $\tilde{u}(T_t)$ is $\mathcal{F}_t$-measurable and the variables $\int_t^{t+\Delta} e^{-\kappa_X(t+\Delta-u)}  dW^T_u$ and $\int_t^{t+\Delta} e^{-\kappa_T(t+\Delta-u)}  dW^T_u$ are independent of $\mathcal{F}_t$, we get the claim by applying Proposition~\ref{prop:vector}. 
\end{proof}

\begin{prop} \label{prop:eS2}
Under  Model~\eqref{eq:comb}, we have
\begin{align*}
    \E ( e^{2 X_{t+\Delta}} \mathbbm{1}_{T_{t+\Delta} \leq u}  \mid \mathcal{F}_t) &= \exp\Bigg(2 \mu_X(t+\Delta) + 2e^{-\kappa_X \Delta}(X_t-\mu_X(t)) \Bigg) \times \\
    & \exp \Bigg( 2 m^X \frac{1- e^{-\kappa_X\Delta}}{\kappa_X} + \delta^X \gamma^X \Delta - \delta^X \int_{0}^{\Delta} \sqrt{(\alpha^X)^2- (\beta^X+ 2 e^{-\kappa_X(\Delta- v)})^2} dv \Bigg) \times \\
    & e^{ 4 \frac{\lambda^2 \sigma_T^2}{2} k_{X}(\Delta)^2 }  \Phi \Big( \frac{\tilde{u}(T_t)}{k_T(\Delta) } - 2 \lambda \sigma_T \frac{k_{XT}^2(\Delta)}{k_{T}(\Delta)} \Big).
\end{align*}
where $k_T(\cdot)$ , $k_X(\cdot)$ and $k_{XT}(\cdot)$ are as in Equation~\eqref{eq:k} and $\tilde{u}(T_t) = \frac{u - (\mu_T(t+\Delta) + e^{-\kappa_T \Delta}(T_t-\mu_T(t))) }{\sigma_T}$.
\end{prop}

\begin{proof}

From~\eqref{eq:comb_int}, we can write
\begin{align*}
    \E ( e^{2 X_{t+\Delta}} \mathbbm{1}_{T_{t+\Delta} \leq u}  \mid \mathcal{F}_t) &= \exp\Bigg(2\mu_X(t+\Delta) + 2e^{-\kappa_X \Delta}(X_t-\mu_X(t)) \Bigg) \E(e^{  2\int_t^{t+\Delta} e^{-\kappa_X(t+\Delta-u)}dL^X_u }) \times \\
    &\E\Bigg(e^{2 \lambda \sigma_T \int_t^{t+\Delta} e^{-\kappa_X(t+\Delta-u)}  dW^T_u} \mathbbm{1}_{T_{t+\Delta} \leq u} \Bigg| \mathcal{F}_t\Bigg) 
\end{align*}
The second term is $\varphi(-i;\Delta)$ by~\eqref{eq_def_phi}, and the last one is
\begin{align*}
    \E\Bigg(e^{2 \lambda \sigma_T \int_t^{t+\Delta} e^{-\kappa_X(t+\Delta-s)}  dW^T_s} \mathbbm{1}_{T_{t+\Delta} \leq u} \Bigg| \mathcal{F}_t\Bigg) 	&= \E\Bigg(e^{2 \lambda \sigma_T \int_t^{t+\Delta} e^{-\kappa_X(t+\Delta-s)}  dW^T_s} \mathbbm{1}_{  \int_t^{t+\Delta} e^{-\kappa_X(t+\Delta-s)}  dW^T_s \leq \tilde{u}(T_t) } \Bigg| \mathcal{F}_t\Bigg) 
\end{align*}
where $\tilde{u}(T_t) = \frac{u - (\mu_T(t+\Delta) + e^{-2 \kappa_T \Delta}(T_t-\mu_T(t))) }{\sigma_T}  $. We now calculate for $\tilde{u} \in \R$, and $G,G^{\perp}\sim \mathcal{N}(0,1)$ independent: 
\begin{align*}
    \E\Bigg(e^{2\lambda \sigma_T(\frac{k_{XT}^2(\Delta)}{k_{T}(\Delta)}  G +  \frac{\sqrt{k_{X}(\Delta)^2 k_{T}(\Delta)^2 -k_{XT}^2(\Delta)}}{k_{T}(\Delta)}  G^{\perp})} \mathbbm{1}_{  k_T(\Delta) G \leq \tilde{u}  }  \Bigg) 
    	&= e^{ \frac{4 \lambda^2 \sigma_T^2}{2} \frac{k_{X}(\Delta)^2 k_{T}(\Delta)^2 -k_{XT}^4(\Delta)}{k_{T}^2(\Delta)} }  \times\\
     &e^{ \frac{(2 \lambda \sigma_T )^2}{2} (\frac{k_{XT}^2(\Delta)}{k_{T}(\Delta)})^2} \Phi \left( \frac{\tilde{u}}{k_T(\Delta) } - 2 \lambda \sigma_T \frac{k_{XT}^2(\Delta)}{k_{T}(\Delta)} \right).
\end{align*}
Since $\tilde{u}(T_t)$ is $\mathcal{F}_t$-measurable and the variables $\int_t^{t+\Delta} e^{-\kappa_X(t+\Delta-u)}  dW^T_u$ and $\int_t^{t+\Delta} e^{-\kappa_T(t+\Delta-u)}  dW^T_u$ are independent of $\mathcal{F}_t$, we get the claim by applying Proposition~\ref{prop:vector}. 
\end{proof}

\begin{prop}\label{prop:S_T^2}
Under  Model~\eqref{eq:comb}, we have
$$\E[S_{t+\Delta}((\bar{T}-T_{t+\Delta})^+)^2|\mathcal{F}_t] = 2 \int_{T^0}^{\bar{T}} (\bar{T}-u) \E[ S_{t+\Delta}\mathbbm{1}_{u\le T_{t+\Delta}} |\mathcal{F}_t]du,$$ 
with $T_0=-\infty$.
\end{prop}
\noindent We write this result with $T^0$, because for numerical purposes we use $T^0=-273.15$ or $T^0=-100$. Note that  $\E[ S_{t+\Delta}\mathbbm{1}_{u\le T_{t+\Delta}} |\mathcal{F}_t]$ can be calculated by using Proposition~\ref{prop:eX1}.
\begin{proof}
    We have 
    \begin{align*}
        ((\bar{T}-T_{t+\Delta})^+)^2 &=\int_{T^0}^{\bar{T}}\int_{T^0}^{\bar{T}} \mathbbm{1}_{ T_{t+\Delta} \le u}\mathbbm{1}_{ T_{t+\Delta} \le v} du dv\\
        &=2\int_{T^0}^{\bar{T}}\int_{T^0}^{\bar{T}} \mathbbm{1}_{ T_{t+\Delta} \le u}\mathbbm{1}_{ T_{t+\Delta} \le v} \mathbbm{1}_{u\le v} du dv =2\int_{T^0}^{\bar{T}} (\bar{T}-u)\mathbbm{1}_{ T_{t+\Delta} \le u} du,
    \end{align*} 
    and therefore
    $$\E[S_{t+\Delta}((\bar{T}-T_{t+\Delta})^+)^2|\mathcal{F}_t]= 2 \int_{T^0}^{\bar{T}} (\bar{T}-u) \E[ S_{t+\Delta}\mathbbm{1}_{ T_{t+\Delta} \le u} |\mathcal{F}_t]du. \qedhere$$
\end{proof}

\subsubsection{Results for two-sided quantos}
\begin{prop}\label{prop:CallS_T^2}
Under  Model~\eqref{eq:comb}, we have the following Taylor decomposition:
\begin{equation*}
    \begin{aligned}
        \E[(S_{t+\Delta}-\Bar{S})^+((\bar{T}-T_{t+\Delta})^+)^2|\mathcal{F}_t] &= \E_{\lambda=0}[(S_{t+\Delta}-\Bar{S})^+|\mathcal{F}_t]\E[((\bar{T}-T_{t+\Delta})^+)^2|\mathcal{F}_t] \\ &-  \lambda \Big( \E_{\lambda=0}[(S_{t+\Delta}-\Bar{S})^+|\mathcal{F}_t]+\Bar{S}\mathbb{P}_{\lambda=0}[(S_{t+\Delta}-\Bar{S})^+|\mathcal{F}_t] \Big) \times\\
        &\sigma_T^3 k_T(\Delta) k_{XT}^2(\Delta) \Big( \sqrt{\frac{2}{\pi}} e^{-\frac{1}{2}(\frac{\Bar{T} - \mu(t+\Delta) -e^{-\kappa_T }\Delta (T_t - \mu(t))}{\sigma_T k_T(\Delta)} )^2} \\
        &+2\frac{\Bar{T} - \mu(t+\Delta) -e^{-\kappa_T }\Delta (T_t - \mu(t))}{\sigma_T k_T(\Delta)} \times \\
        &  \Phi \Big(\frac{\Bar{T} - \mu(t+\Delta) -e^{-\kappa_T }\Delta (T_t - \mu(t))}{\sigma_T k_T(\Delta)}\Big) \Big) + o(\lambda)
    \end{aligned} 
\end{equation*}
where   $k_T(\cdot)$, $k_{XT}(\cdot)$  are as in Equation~\eqref{eq:k}.
\end{prop}
\noindent Note that $\E_{\lambda=0}[(S_{t+\Delta}-\Bar{S})^+|\mathcal{F}_t]$ and $\mathbb{P}_{\lambda=0} \left(   S_{t+\Delta} \geq \bar{S}   \mid \mathcal{F}_{t} \right )$ are computed through Equation~\eqref{eq:carr_madan} and \eqref{eq:gil-pelaez} respectively, and that $\E_{\lambda=0}[((\bar{T}-T_{t+\Delta})^+)^2|\mathcal{F}_t]$ can be calculated by using Proposition~\ref{prop:E_T_2}.
\begin{proof}
As for $\E[(S_{t+\Delta}-\Bar{S})^+]$ we will perform a Taylor decomposition of $\E[(S_{t+\Delta}-\Bar{S})^+((\bar{T}-T_{t+\Delta})^+)^2|\mathcal{F}_t]$.
Let consider the 0 order term,
$$\E_{\lambda=0}[(S_{t+\Delta}-\Bar{S})^+((\bar{T}-T_{t+\Delta})^+)^2|\mathcal{F}_t] = \E_{\lambda=0}[(S_{t+\Delta}-\Bar{S})^+|\mathcal{F}_t]\E[((\bar{T}-T_{t+\Delta})^+)^2|\mathcal{F}_t].$$
Now, let us compute the derivative at $\lambda= 0$:
\begin{equation*}
        \begin{aligned}
            \left. {\frac{d }{d \lambda}} \right|_{\lambda=0}  \E\left[(S_{t+\Delta}-\Bar{S})^+((\bar{T}-T_{t+\Delta})^+)^2\bigg|\mathcal{F}_t\right]= \E\left[ \sigma_T \int_t^{t+\Delta} e^{-\kappa_X(t+\Delta-v)} dW^T_vS_{t+\Delta} \mathbbm{1}_{S_{t+\Delta} \geq \bar{S}} ((\bar{T}-T_{t+\Delta})^+)^2\bigg|\mathcal{F}_t\right] \\
            = \mathbb{E}_{\lambda=0} \left(   S_{t+\Delta} \mathbbm{1}_{S_{t+\Delta} \geq \bar{S} }  \mid \mathcal{F}_{t} \right ) \E\left[ \sigma_T \int_t^{t+\Delta} e^{-\kappa_X(t+\Delta-v)} dW^T_v ((\bar{T}-T_{t+\Delta})^+)^2 \bigg|\mathcal{F}_t \right],
        \end{aligned}
    \end{equation*}
and $\mathbb{E}_{\lambda=0} [  S_{t+\Delta} \mathbbm{1}_{S_{t+\Delta} \geq \bar{S} } |\mathcal{F}_t]=\E_{\lambda=0}[(S_{t+\Delta}-\Bar{S})^+|\mathcal{F}_t]+\Bar{S}\mathbb{P}_{\lambda=0}[(S_{t+\Delta}-\Bar{S})^+|\mathcal{F}_t]$.\\
By Proposition~\ref{prop:vector}, we have
$$\E[ \sigma_T \int_t^{t+\Delta} e^{-\kappa_X(t+\Delta-v)} dW^T_v ((\bar{T}-T_{t+\Delta})^+)^2|\mathcal{F}_t]=\frac{k_{XT}^2(\Delta)}{k_{T}^2(\Delta)}\E\left[ \sigma_T \int_t^{t+\Delta} e^{-\kappa_X(t+\Delta-v)} dW^T_v ((\bar{T}-T_{t+\Delta})^+)^2\big|\mathcal{F}_t\right].$$ We can now use Lemma~\ref{lemma:gauss_comb2} with $a=\Bar{T} - \mu(t+\Delta) -e^{-\kappa_T \Delta } (T_t - \mu(t))$ and $b=-\sigma_T k_T(\Delta)$ to get
\begin{equation*}
    \begin{aligned}
        \E \left[ \sigma_T \int_t^{t+\Delta} e^{-\kappa_T(t+\Delta-v)} dW^T_v ((\bar{T}-T_{t+\Delta})^+)^2 \bigg|\mathcal{F}_t\right] &= - \sigma_T^3  k_T(\Delta)^3 \Big( \sqrt{\frac{2}{\pi}} e^{-\frac{1}{2}(\frac{\Bar{T} - \mu(t+\Delta) -e^{-\kappa_T \Delta} (T_t - \mu(t))}{\sigma_T k_T(\Delta)})^2} \\
        &+2\frac{\Bar{T} - \mu(t+\Delta) -e^{-\kappa_T \Delta} (T_t - \mu(t))}{\sigma_T k_T(\Delta)} \times\\
        &  \Phi \Big(\frac{\Bar{T} - \mu(t+\Delta) -e^{-\kappa_T \Delta} (T_t - \mu(t))}{\sigma_T k_T(\Delta)}\Big) \Big)  \qedhere
    \end{aligned}
\end{equation*}
\end{proof}

\begin{prop}\label{prop:S^2_T}
Under  Model~\eqref{eq:comb}, we have the following Taylor expansion:
\begin{equation*}
    \begin{aligned}
        \E[((S_{t+\Delta}-\Bar{S})^+)^2(\bar{T}-T_{t+\Delta})^+|\mathcal{F}_t] &=
        \E_{\lambda=0}[((S_{t+\Delta}-\Bar{S})^+)^2|\mathcal{F}_t]\E[(\bar{T}-T_{t+\Delta})^+|\mathcal{F}_t] \\
        & - 2   \Big(\E_{\lambda=0}[((S_{t+\Delta}-\Bar{S})^+)^2|\mathcal{F}_t]+\Bar{S}\E_{\lambda=0}[(S_{t+\Delta}-\Bar{S})^+|\mathcal{F}_t] \Big) \times\\
        &\sigma_T^2 k_{XT}(\Delta)^2 \Phi\Big(\frac{\bar{T}-\mu_T(t+\Delta) - e^{-\kappa_T \Delta}(T_t -\mu_T(t))}{\sigma_T k_T(\Delta)}\Big) \lambda + o(\lambda)
    \end{aligned}
\end{equation*}
where  $k_T(\cdot)$ and $k_{XT}(\cdot)$ are  as in Equation~\eqref{eq:k}. 
\end{prop}
\noindent Note that $\E_{\lambda=0}[(S_{t+\Delta}-\Bar{S})^+|\mathcal{F}_t]$ and $\E_{\lambda=0}[((S_{t+\Delta}-\Bar{S})^+)^2|\mathcal{F}_t]$ can be computed through Equation~\eqref{eq:carr_madan}  and~\eqref{eq:carr_madan_square} respectively, $\E_{\lambda=0}[(\bar{T}-T_{t+\Delta})^+|\mathcal{F}_t]$ is given by Proposition~\ref{prop:E_T_positive} and   $\E_{\lambda=0}[((\bar{T}-T_{t+\Delta})^+)^2|\mathcal{F}_t]$ can be calculated by using  Proposition~\ref{prop:E_T_2}.
\begin{proof}
For $\lambda=0$, we have
$$\E_{\lambda=0}[((S_{t+\Delta}-\Bar{S})^+)^2(\bar{T}-T_{t+\Delta})^+|\mathcal{F}_t] = \E_{\lambda=0}[((S_{t+\Delta}-\Bar{S})^+)^2|\mathcal{F}_t]\E_{\lambda=0}[(\bar{T}-T_{t+\Delta})^+|\mathcal{F}_t].$$
Now, let us compute the derivative in $\lambda=0$:
\begin{equation*}
        \begin{aligned}
            \left. {\frac{d }{d \lambda}} \right|_{\lambda=0}  \E[((S_{t+\Delta}-\Bar{S})^+)^2(\bar{T}-T_{t+\Delta})^+|\mathcal{F}_t] = \E_{\lambda=0}\bigg[ \sigma_T \int_t^{t+\Delta} e^{-\kappa_X(t+\Delta-v)} dW^T_v\times\\
            2 S_{t+\Delta}(S_{t+\Delta}-\Bar{S})^+ (\Bar{T}-T_{t+\Delta})^+\bigg|\mathcal{F}_t\bigg] \\
         =2 \E_{\lambda=0}[S_{t+\Delta}(S_{t+\Delta}-\Bar{S})^+|\mathcal{F}_t] \times \\
         \E\left[ \sigma_T \int_t^{t+\Delta} e^{-\kappa_X(t+\Delta-v)} dW^T_v (\Bar{T}-T_{t+\Delta})^+\big| \mathcal{F}_t\right].
        \end{aligned}
\end{equation*}
The calculation of $\E[ \sigma_T \int_t^{t+\Delta} e^{-\kappa_X(t+\Delta-v)} dW^T_v (\Bar{T}-T_{t+\Delta})^+|\mathcal{F}_t]$ has already been done in Equations~\eqref{calc_LemD2a} and~\eqref{calc_LemD2b}.
\end{proof}

\end{document}